\RequirePackage{amsmath}
\documentclass{llncs}
\usepackage{amsmath}
\usepackage{amssymb}
\usepackage{url}
\usepackage[bookmarks,bookmarksdepth=2,pdfusetitle]{hyperref}
\usepackage{graphicx,multirow,wasysym}
\usepackage{lmodern}
\usepackage{amsfonts,amsmath}
\usepackage{paralist}
\usepackage{enumerate}
\usepackage{bbm}
\usepackage{tikz}
\usepackage{verbatim}
\usepackage[shortlabels]{enumitem}
\usepackage{comment}
\usepackage{booktabs}
\usepackage{lipsum}
\usepackage{tabularx}
\usepackage{xmpmulti}
\usepackage{multicol}
\usepackage{float}
\floatstyle{boxed}
\restylefloat{figure}
\begin{document}



\newcommand{\mbbN}{\mathbb{N}}
\newcommand{\mcA}{\mathcal{A}}
\newcommand{\mcB}{\mathcal{B}}
\newcommand{\mcS}{\mathcal{S}}
\newcommand{\mcD}{\mathcal{D}}

\newcommand{\ms}{\mathsf} 
\newcommand{\mbb}{\mathbb}
\newcommand{\mb}{\mathbf}
\newcommand{\mc}{\mathcal}
\newcommand{\Z}{\mathbb{Z}}
\newcommand{\rdec}{\mathsf{rdec}}

\pagestyle{plain}
\mainmatter
\title{Accountable Tracing Signatures from Lattices}

\author{San Ling, Khoa Nguyen, Huaxiong Wang, Yanhong Xu}
\institute{
Division of Mathematical Sciences, \\
School of Physical and Mathematical Sciences,\\
Nanyang Technological University, Singapore.\\
\texttt{\{lingsan,khoantt,hxwang,xu0014ng\}@ntu.edu.sg}
}
\maketitle

\begin{abstract}
Group signatures allow users of a group to sign messages anonymously in the name of the group, while incorporating a tracing mechanism to revoke anonymity and identify the signer of any message. Since its introduction by Chaum and van Heyst (EUROCRYPT 1991), numerous proposals have been put forward, yielding various improvements on security, efficiency and functionality. However, a drawback of traditional group signatures is that the opening authority is given too much power, i.e., he can indiscriminately revoke anonymity and there is no mechanism to keep him accountable.
To overcome this problem, Kohlweiss and Miers (PoPET 2015) introduced the notion of accountable tracing signatures ($\mathsf{ATS}$) - an enhanced group signature variant in which the opening authority is kept accountable for his actions. Kohlweiss and Miers demonstrated a generic construction of $\mathsf{ATS}$ and put forward a concrete instantiation based on number-theoretic assumptions. To the best of our knowledge, no other $\mathsf{ATS}$ scheme has been known, and the problem of instantiating $\mathsf{ATS}$ under post-quantum assumptions, e.g., lattices, remains open to date.


~~In this work, we provide the first lattice-based accountable tracing signature scheme. The scheme satisfies the security requirements suggested by Kohlweiss and Miers, assuming the hardness of the Ring Short Integer Solution ($\mathsf{RSIS}$) and  the Ring Learning With Errors ($\mathsf{RLWE}$) problems. At the heart of our construction are a lattice-based key-oblivious encryption scheme and a zero-knowledge argument system allowing to prove that a given ciphertext is a valid $\mathsf{RLWE}$ encryption under some hidden yet certified key. 
These technical building blocks may be of independent interest, e.g., they can be useful for the design of other lattice-based privacy-preserving protocols.


\end{abstract}

\section{Introduction}

Group signature is a fundamental cryptographic primitive introduced by Chaum and van Heyst~\cite{CV91}. It allows members of a group to anonymously sign messages on behalf of the group, but to prevent abuse of anonymity, there is an opening authority (\textsf{OA}) who can identify the signer of any message. While such a tracing mechanism is necessary to ensure user accountability, it grants too much power to the opening authority. Indeed, in traditional models of group signatures, e.g.,~\cite{BMW03,KTY04,BS04-VLR,BSZ05,KY06,SSEHO12,BCCGG16}, the \textsf{OA} can break users' anonymity whenever he wants, and we do not have any method to verify whether this trust is well placed or not.

One existing attempt to restrict the \textsf{OA}'s power is the proposal of group signatures with message-dependent opening (MDO)~\cite{SEHK+12}, in which the \textsf{OA} can only identify the signers of messages admitted by an additional authority named admitter. However, this solution is still unsatisfactory. Once the \textsf{OA} has obtained admission to open a specific message, he can identify all the users, including some innocent ones, who have ever issued signatures on this specific message. Furthermore, by colluding with the admitter, the \textsf{OA} again is able to open all signatures.

To tackle the discussed above problem, Kohlweiss and Miers~\cite{KM15} put forward the notion of accountable tracing signatures ($\mathsf{ATS}$), which is an enhanced variant of group signatures that has an additional mechanism to make the \textsf{OA} accountable. In an $\mathsf{ATS}$ scheme,  the role of the \textsf{OA} is incorporated into that of the group manager (\textsf{GM}), and there are two kinds of group users: traceable ones and non-traceable ones. Traceable users are treated as in traditional group signatures, i.e., their anonymity can be broken by the \textsf{OA}/\textsf{GM}.  Meanwhile,  it is infeasible for anyone, including the \textsf{OA}/\textsf{GM}, to trace signatures generated by non-traceable users. When a user joins the group, the \textsf{OA}/\textsf{GM} first has to determine whether this user is traceable and then he issues a corresponding (traceable/nontraceable) certificate to the user. In a later phase, the \textsf{OA}/\textsf{GM} reveals which user he deems traceable using an ``accounting'' algorithm, yielding an intriguing method to enforce his accountability.

As an example, let us consider the surveillance controls of a building, which is implemented using an \textsf{ATS} scheme. On the one hand, the customers in this building would like to have their privacy protected as much as possible. On the other hand, the police who are conducting security check in this building would like to know as much as they can. To balance the interests of these two parties, the police can in advance narrow down some suspects and asks the \textsf{OA}/\textsf{GM} to make these suspected users traceable and the remaining non-suspected users non-traceable. To check whether the suspects  entered the building, the police can ask the \textsf{OA}/\textsf{GM} to open all signatures that were used for authentication at the entrance. Since only the suspects are traceable,  the group manager can only identify them if they indeed entered this building. However, if a standard group signature scheme~(e.g.,~\cite{ACJT00,BMW03,BBS04,BSZ05}) were used, then the privacy of innocent users would be seriously violated. In this situation, one might think that a traceable signature scheme, as suggested by Kiayias, Tsiounis and Yung~\cite{KTY04}, would work.  By requesting a user-specific trapdoor from the \textsf{OA}/\textsf{GM}, the police can trace all the signatures created by the suspects. However, this only achieves privacy of innocent users against  the \emph{police}, but not against the \emph{group authorities}. In fact, in a traceable signature scheme, the \textsf{OA}/\textsf{GM} has the full power to identify the signers of all signatures and hence can violate the privacy of all users without being detected. In contrast, if an $\mathsf{ATS}$ scheme is used, then the \textsf{OA}/\textsf{GM} must later reveal which user he chose to be traceable, thus enabling his accountability.

In~\cite{KM15}, besides demonstrating the feasibility of \textsf{ATS} under generic assumptions, Kohlweiss and Miers also presented an instantiation based on number-theoretic assumptions, which remains the only known concrete \textsf{ATS} construction to date. This scheme, however, is vulnerable against quantum computers due to Shor's algorithm~\cite{Shor94}.
For the sake of not putting all eggs in one basket, it is therefore tempting to build schemes based on post-quantum foundations. In this paper, we investigate the design of accountable tracing signatures based on lattice assumptions, which are currently among the most viable foundations for post-quantum cryptography. Let us now take a look at the closely related and recently active topic of lattice-based group signatures.

\smallskip

\noindent
{\sc Lattice-based group signatures. }
The first lattice-based group signature scheme was introduced by Gordon, Katz and Vaikuntanathan in 2010~\cite{GKV10}. Subsequently, numerous schemes offering improvements in terms of security and efficiency have been proposed~\cite{CNR12,LLLS13-Asiacrypt,LNW15,NZZ15,LLNW16,LLMNW16-dgs,BCN18,PinoLS18}. Nevertheless, regarding the supports of advanced functionalities, lattice-based group signatures are still way behind their number-theoretic-based counterparts. Indeed, there have been known only a few lattice-based schemes~\cite{LNRW18,LMN16,LLMNW16-dgs,LNWX17,LNWX18} that depart from the BMW model~\cite{BMW03} - which deals solely with static groups and which may be too inflexible to be considered for a wide range of real-life applications. In particular, although there was an attempt~\cite{LMN16} to restrict the power of the \textsf{OA} in the MDO sense, the problem of making the \textsf{OA} accountable in the context of lattice-based group signatures is still open. This somewhat unsatisfactory state-of-affairs motivates our search for a lattice-based instantiation of \textsf{ATS}. As we will discuss below, the technical road towards our goal is not straightforward: there are challenges and missing building blocks along the way.

\medskip

\noindent
{\sc Our Results and Techniques.} In this paper, we introduce the first lattice-based accountable tracing signature scheme. The scheme satisfies the security requirements suggested by Kohlweiss and Miers~\cite{KM15}, assuming  the hardness of the Ring Short Integer Solution ($\mathsf{RSIS}$) problem and  the Ring Learning With Errors ($\mathsf{RLWE}$) problem.  As all other known lattice-based group signatures, the security of our scheme is analyzed in the random oracle model. For a security parameter $\lambda$, our $\mathsf{ATS}$ scheme features group public key size and user secret key size  $\widetilde{\mathcal{O}}(\lambda)$. However, the accountability of the \textsf{OA}/\textsf{GM} comes at a price: the signature size is of order $\widetilde{\mathcal{O}}(\lambda^2)$ compared with $\widetilde{\mathcal{O}}(\lambda)$ in a recent scheme by Ling et al.~\cite{LNWX18}.


Let us now give an overview of our techniques. First, we recall that in an ordinary group signature scheme~\cite{BMW03,BSZ05}, to enable traceability, the user is supposed to encrypt his identifying information and prove the well-formedness of the resulting ciphertext. In an $\mathsf{ATS}$ scheme, however, not all users are traceable. We thus would need a mechanism to distinguish between traceable users and non-traceable ones. A possible method is to let traceable users encrypt their identities under a public key ($\mathsf{pk}$) such that only the \textsf{OA}/\textsf{GM} knows the underlying secret key ($\mathsf{sk}$), while for non-traceable users, no one knows the secret key. However, there seems to be no incentive for users to deliberately make themselves traceable. We hence should think of a way to choose traceable users obliviously. An interesting approach is to randomize  $\mathsf{pk}$ to a new public key $\mathsf{epk}$ so that it is infeasible to decide how these keys are related without the knowledge of the secret key and the used randomness.  More specifically, when a user joins the group, the \textsf{OA}/\textsf{GM}  first randomizes $\mathsf{pk}$ to $\mathsf{epk}$ and sends  the latter to the user together with a certificate. The difference between traceable users and non-traceable ones lies in whether \textsf{OA}/\textsf{GM} knows the underlying secret key. Thanks to the obliviousness property of the randomization, the users are unaware of whether they are traceable. Then, when signing messages, the user encrypts his identity using his own randomized key $\mathsf{epk}$ (note that this ``public key'' should be kept secret) and proves the well-formedness of the ciphertext.  Several questions regarding this approach then arise. What special kind of encryption scheme should we use? How to randomize the public key in order to get the desirable obliviousness? More importantly, how could the user prove the honest execution of encryption if the underlying encryption key is secret?

To address the first two questions, Kohlweiss and Miers~\cite{KM15} proposed the notion of key-oblivious encryption ($\mathsf{KOE}$) - a public-key encryption scheme in which one can randomize public keys in an oblivious manner.  Kohlweiss and Miers showed that a $\mathsf{KOE}$ scheme can be built from a key-private homomorphic public-key encryption scheme.  They then gave an explicit construction based on the ElGamal cryptosystem~\cite{ElGamal84}, where $\mathsf{epk}$ is obtained by multiplying $\mathsf{pk}$ by a ciphertext of~$1$. When adapting this idea into the lattice setting, however, one has to be careful. In fact, we observe that an implicit condition for the underlying key-private public-key encryption scheme is that its public key and ciphertext should have \emph{the same algebraic form}\footnote{This condition is needed so that $\mathsf{epk}$ can be computed as $\mathsf{pk} ~\cdot~ \mathsf{enc}(1)$ (multiplicative homomorphic) or $\mathsf{pk} ~+~ \mathsf{enc}(0)$ (additive homomorphic).}, which is often not the case for the schemes in the lattice setting, e.g.,~\cite{Regev05,GPV08}. Furthermore,  lattice-based encryption schemes from the Learning with Errors (\textsf{LWE}) problem or its ring version \textsf{RLWE} often involve noise terms that grow quickly when one performs homomorphic operations over ciphertexts. Fortunately, we could identify a suitable candidate: the \textsf{RLWE}-based encryption scheme proposed by Lyubashevsky, Peiker and Regev (LPR)~\cite{LPR13-ACM}, for which both the public key and the ciphertext consist of a pair of ring elements. Setting the parameters carefully to control the noise growth in LPR, we are able to adapt the blueprint of~\cite{KM15} into the lattice setting and obtain a lattice-based \textsf{KOE} scheme.

To tackle the third question, we need a zero-knowledge ($\mathsf{ZK}$) protocol for proving well-formedness of the ciphertext under a hidden encryption key, which is quite challenging to build in the \textsf{RLWE} setting. Existing $\mathsf{ZK}$ protocols from lattices belong to two main families. One line of research~\cite{Lyu09,Lyu12,BCKLN14,BKLP15,LyuN17,LyuS18} designed very elegant approximate $\mathsf{ZK}$ proofs for $\mathsf{(R)LWE}$ and $\mathsf{(R)SIS}$ relations by employing rejection sampling techniques. While  these proofs are quite efficient and compact, they only handle linear relations. In other words, they can only prove knowledge of a short vector $\mathbf{x}$ satisfying $\mathbf{y}=\mathbf{A}\cdot \mathbf{x}\bmod q$, for \emph{public} $\mathbf{A}$ and  public $\mathbf{y}$. This seems insufficient for our purpose.  Another line of research~\cite{LNSW13,LNW15,CNW16,LLNW16,LLMNW16-ge,LNWX18} developed decomposition/ extension/permutation techniques that operate in Stern's framework~\cite{Ste96}. Although Stern-like protocols are less practical than those in the first family, they are much more versatile and can even deal with quadratic relations~\cite{LLMNW16-ge}. More precisely, as demonstrated by Libert et al.~\cite{LLMNW16-ge} one can employ Stern-like techniques to prove knowledge of \emph{secret-and-certified} $\mathbf{A}$ together with short secret vector $\mathbf{x}$ satisfying $\mathbf{y}=\mathbf{A}\cdot \mathbf{x}\bmod q$. Thus, Libert et al.'s work appears to be the ``right'' stepping stone for our case. However, in~\cite{LLMNW16-ge}, quadratic relations were considered only in the setting of general lattices, while here we have to deal with the ring setting, for which the multiplication operation is harder to express, capture and prove in zero-knowledge. Nevertheless we manage to adapt their techniques into the ring lattices and obtain the desired technical building block. 

As discussed so far, we have identified the necessary ingredients - the LPR encryption scheme and Stern-like $\mathsf{ZK}$ protocols - for upgrading a lattice-based ordinary group signature to a lattice-based accountable tracing signature. Next, we need to find a lattice-based ordinary group signature scheme that is compatible with the those ingredients. To this end, we work with Ling et al.'s scheme~\cite{LNWX18}, that also employs the LPR system for its tracing layer and Stern-like techniques for proving knowledge of a valid user certificate (which is a Ducas-Micciancio signature~\cite{DM14,DM14-eprint} based on the hardness of the Ring Short Integer Solution (\textsf{RSIS}) problem). We note that the scheme from~\cite{LNWX18} achieves constant-size signatures, which means that the signature size is independent of the number of users. As a by-product, our signatures are also constant-size (although our constant is larger, due to the treatment of quadratic relations).

A remaining aspect is how to enable the accountability of the \textsf{OA}/\textsf{GM}. To this end, we let the latter reveal the choice (either traceable or non-traceable) for a given user together with the randomness used to obtain the randomized public key. The user then checks whether his $\mathsf{epk}$ was computed as claimed. However, the \textsf{OA}/\textsf{GM} may claim a traceable user to be non-traceable  by giving away malicious randomness and accusing that the user had changed $\mathsf{epk}$ by himself. To ensure non-repudiation, \textsf{OA}/\textsf{GM} is required to sign $\mathsf{epk}$ and the users' identifying information when registering the user into the group. This mechanism in fact also prevents dishonest users from choosing  non-traceable $\mathsf{epk}$ by themselves.

The obtained \textsf{ATS} scheme is then proven secure in the random oracle model under the \textsf{RSIS} and \textsf{RLWE} assumptions, according to the security requirements put forward by Kohlweiss and Miers~\cite{KM15}. On the efficiency front, as all known lattice-based group signatures with advanced functionalities, our scheme is still far from being practical. We, however, hope that our result will inspire more efficient constructions in the near future. \smallskip

\noindent
{\sc Organization.} In Section~\ref{section:ATS-background}, we recall some background materials. In Section~\ref{section:ATS-KOE}, we describe our key-oblivious encryption scheme from lattice assumptions. 
Our accountable tracing signature scheme is presented in Section~\ref{section:ATS-ATS}.

\section{Background} \label{section:ATS-background}

{\sc Notations. }
For a positive integer $n$, define the set $\{1,2,\ldots,n\}$ as $[n]$,  the set $\{0,1,\ldots,n\}$ as $[0,n]$, and the set containing all the integers from $-n$ to $n$ as $[-n,n]$. Denote the set of all positive integers as $\mathbb{Z}^{+}$. If $S$ is a finite set, then $x \xleftarrow{\$} S$ means that $x$ is chosen uniformly at random from $S$.  Let $\mathbf{a}\in\mathbb{R}^{m_1}$ and $\mathbf{b}\in\mathbb{R}^{m_2}$ be two vectors for positive integers $m_1,m_2$. Denote $(\mathbf{a}\|\mathbf{b})\in\mathbb{R}^{m_1+m_2}$, instead of $(\mathbf{a}^\top,\mathbf{b}^{\top})^{\top}$, as the concatenation of these two vectors.

\subsection{Rings, RSIS and RLWE}\label{subsection:rings}
Let $q \geq 3$ be a positive integer and let $\mathbb{Z}_q = [-\frac{q-1}{2}, \frac{q-1}{2}]$. In this work, let us consider rings $R = \mathbb{Z}[X]/(X^n+1)$ and $R_q = (R/qR)$, where $n$ is a power of $2$.

Let $\tau$ be the coefficient embedding $\tau: R_q \rightarrow \mathbb{Z}_q^n$ that maps a ring element $v = v_0 + v_1 \cdot X + \ldots + v_{n-1}\cdot X^{n-1} \in R_q$ to a vector $\tau(v) = (v_0, v_1, \ldots, v_{n-1})^\top$ over $\mathbb{Z}_q^n$. Define the ring homomorphism $\mathsf{rot}: R_q \rightarrow \mathbb{Z}_q^{n \times n}$that maps a ring element $a \in R_q$ to a matrix $\mathsf{rot}(a) = \big[\tau(a) \mid \tau(a\cdot X) \mid \cdots \mid \tau(a \cdot X^{n-1})\big]$ over $\mathbb{Z}_q^{n \times n}$ (see, e.g., \cite{Micciancio07,Xagawa15}). Using these two functions, the element product $y = a\cdot v$ over $R_q$ can be interpreted as the matrix-vector multiplication $\tau(y) = \mathsf{rot}(a) \cdot \tau(v)$ over $\mathbb{Z}_q$.

When working with vectors and matrices over $R_q$, we  generalize the notations $\tau$ and $\mathsf{rot}$ in the following way. For a vector $\mathbf{v} = (v_1, \ldots, v_m)^\top \in R_q^m$, define $\tau(\mathbf{v}) = (\tau(v_1) \| \cdots \| \tau(v_m)) \in \mathbb{Z}_q^{mn}$. For a matrix $\mathbf{A} = [a_1 \mid \cdots \mid a_m] \in R_q^{1 \times m}$, define $\mathsf{rot}(A)$ to be the matrix
\[
\mathsf{rot}(\mathbf{A}) = \big[\mathsf{rot}(a_1) \mid \cdots \mid \mathsf{rot}(a_m)\big] \in \mathbb{Z}_q^{n \times mn}.
\]
Using the generalized notations, we can interpret $y = \mathbf{A} \cdot \mathbf{v}$ over $R_q$ as matrix-vector multiplication $\tau(y) = \mathsf{rot}(\mathbf{A}) \cdot \tau(\mathbf{v})$ over $\mathbb{Z}_q$.
\smallskip

For $a = a_0 + a_1 \cdot X + \ldots + a_{n-1}\cdot X^{N-1} \in R$, we define $\|a\|_\infty = \max_i(|a_i|)$. Similarly, for vector $\mathbf{b} = (b_1, \ldots, b_{\mathfrak{m}})^\top \in R^{\mathfrak{m}}$, we define $\|\mathbf{b}\|_\infty = \max_j(\|b_j\|_\infty)$.

We now recall the average-case problems $\mathsf{RSIS}$ and $\mathsf{RLWE}$ associated with the rings $R, R_q$, as well as their hardness results.

\begin{definition}[\cite{LM06,PR06,LMPR08}]
Given a uniform  matrix $\mathbf{A}=[a_1|a_2|\cdots|a_m]$ over $R_q^{1\times m}$, the $\mathsf{RSIS}_{n,m,q,\beta}^{\infty}$ problem asks to find a ring vector $\mathbf{b}=(b_1,b_2,\ldots,b_m)^\top$ over  $R^{m}$ such that $\mathbf{A}\cdot \mathbf{b}=a_1 \cdot b_1 + a_2\cdot b_2+\cdots + a_m \cdot b_m = 0$ over $R_q$ and $0< \|\mathbf{b}\|_{\infty}\leq \beta$.

\end{definition}

For polynomial bounded $m,\beta$ and   $q\geq \beta \cdot \widetilde{\mathcal{O}}(\sqrt{n})$, it was proven that the $\mathsf{RSIS}_{n,m,q,\beta}^{\infty}$ problem is no easier than the $\mathsf{SIVP}_{\gamma}$ problem in any ideal in the ring $R$,  where  $\gamma=\beta \cdot \widetilde{\mathcal{O}}(\sqrt{nm})$ (see~\cite{LM06,PR06,LS15}).

\begin{definition}[\cite{LPR10-EC,SSTX09,LPR13-ACM}]
For positive integers $n,m,q\geq 2$ and a probability distribution $\chi$ over the ring $R$, define a distribution $A_{s, \chi}$ over $R_q \times {R}_q$ for $s \xleftarrow{\$}  {R}_q$ in the following way: it first samples a uniformly random element $a\in R_q$, an error element $e\hookleftarrow\chi$, and then outputs $(a,a\cdot s+e)$. The target of the $\mathsf{RLWE}_{n,m,q,\chi}$ problem is to distinguish $m$ samples chosen from a uniform distribution over $R_q \times {R}_q$ and $m$ samples chosen from the distribution $A_{s, \chi}$ for $s \xleftarrow{\$}  {R}_q$.
\end{definition}
Let  $q\geq 2$ and $B=\widetilde{\mathcal{O}}(\sqrt{n})$ be positive integers. $\chi$ is a distribution over $R$ which efficiently outputs samples $e \in R$ with $\|e\|_\infty \leq B$ with overwhelming probability in $n$. Then there is a  quantum reduction from the $\mathsf{RLWE}_{n,m,q,\chi}$ problem to the  $\mathsf{SIVP}_{\gamma}$ problem  and the $\mathsf{SVP}_{\gamma}$ problem
in any ideal in the ring $R$, where $\gamma=\widetilde{\mathcal{O}}(\sqrt{n}\cdot q/B)$ (see~\cite{LPR10-EC,BGV12,LS15,PeikertRS17}). It is shown that the hardness of the $\mathsf{RLWE}$ problem is preserved when the secret $s$ is sampled from the error distribution $\chi$ (see~\cite{LPR10-EC,BGV12}).

\subsection{Decompositions}\label{subsection:decomposition}
We now recall the integer decomposition technique from~\cite{LNSW13}. For any positive integer $B$,  let   $\delta_B:=\lfloor \log_2 B\rfloor +1 = \lceil \log_2(B+1)\rceil$ and the sequence $B_1, \ldots, B_{\delta_B}$, where $B_j = \lfloor\frac{B + 2^{j-1}}{2^j} \rfloor$, for any $  j \in [\delta_B]$. It is then verifiable that $\sum_{j=1}^{\delta_B} B_j = B$. In addition, for any integer $a\in[0,B]$, one can decompose $a$ into a vector of the form $\mathsf{idec}_B(a)=(a^{(1)},a^{(2)}, \ldots, a^{(\delta_B)})^\top \in \{0,1\}^{\delta_B}$, satisfying $(B_1,B_2,\ldots,B_{\delta_B})\cdot \mathsf{idec}_B(a)=a$. The procedure of the decomposition is presented below in a deterministic manner.
\begin{enumerate}
\item $a': = a$
\item For $j=1$ to $\delta_B$ do:
    \begin{enumerate}[(i)]
    \item If $a' \geq B_j$ then $a^{(j)}: = 1$, else $a^{(j)}: = 0$;
    \item $a': = a' - B_j\cdot a^{(j)}$.
    \end{enumerate}
\item Output $\mathsf{idec}_B(a) = (a^{(1)}, \ldots, a^{(\delta_B)})^\top$.
\end{enumerate}

In~\cite{LNWX18}, the above decomposition procedure is also utilized to deal with polynomials in the ring $R_q$. Specifically, for $B \in [1, \frac{q-1}{2}]$, define the injective function $\rdec_B$ that maps $a \in R_q$ with $\|a\|_\infty \leq B$ to $\mathbf{a} \in  R^{\delta_B}$ with $\|\mathbf{a}\|_\infty \leq 1$, which works as follows.

\begin{enumerate}
\item Let $\tau(a) = (a_0, \ldots, a_{n-1})^\top$. For each $i$, let $\sigma(a_i) = 0$ if $a_i =0$; $\sigma(a_i) = -1$ if $a_i <0$; and $\sigma(a_i) = 1$ if $a_i >0$. \smallskip
\item $\forall \hspace*{1pt} i$, compute $\mathbf{w}_i = \sigma(a_i)\cdot \mathsf{idec}_B(|a_i|) = (w_{i,1}, \ldots, w_{i,\delta_B})^\top \in \{-1,0,1\}^{\delta_B}$. \smallskip
\item Form the vector $\mathbf{w} = (\mathbf{w}_0 \| \ldots \| \mathbf{w}_{n-1}) \in \{-1,0,1\}^{n\delta_B}$, and let $\mathbf{a} \in  R^{\delta_B}$ be the vector such that $\tau(\mathbf{a}) = \mathbf{w}$. \smallskip
\item Output $\rdec_B(a) = \mathbf{a}$.
\end{enumerate}
To deal with ring vectors of dimension $m\in\mathbb{Z}^{+}$ and of infinity bound $B\in\mathbb{Z}^{+}$, we generalize the notion $\rdec_B(\mathbf{v})$ in the following way: it maps a ring vector $\mathbf{v} = (v_1, \ldots, v_m)^\top\in R_q^{m}$ such that $\|\mathbf{v}\|_\infty \leq B$ to a vector $\rdec_B(\mathbf{v}) = \big(\rdec_B(v_1) \| \ldots \| \rdec_B(v_m)\big) \in R^{m\delta_B}$, whose coefficients are in the set $\{-1,0,1\}$. \smallskip

\noindent
Now, $\forall \hspace*{1pt} m, B \in \mathbb{Z}^+$, we define matrices $\mathbf{H}_{B} \in \mathbb{Z}^{n \times n\delta_B}$ and $\mathbf{H}_{m, B}  \in \mathbb{Z}^{nm \times nm\delta_B}$ as
\begin{eqnarray*}
\mathbf{H}_{B} = \begin{bmatrix} B_1 \ldots  B_{\delta_B} &  & & & \\
	  &   &  &  \ddots  &  \\
			  &   &  &    & B_1 \ldots  B_{\delta_B}  \\
\end{bmatrix} , \hspace*{6.8pt}\text{ and } \hspace*{6.8pt}
                        \mathbf{H}_{m, B} = \left[
                           \begin{array}{ccc}
                             \mathbf{H}_B &  &  \\
                              & \ddots &  \\
                              &  & \mathbf{H}_B \\
                           \end{array}
                         \right].
\end{eqnarray*}
Then we have
\[
\tau(a) = \mathbf{H}_{B} \cdot \tau(\rdec_B(a)) \bmod q \hspace*{6.8pt}\text{ and }\hspace*{6.8pt} \tau(\mathbf{v}) = \mathbf{H}_{m,B}\cdot \tau(\rdec_B(\mathbf{v})).
\]

For simplicity reason, when $B = \frac{q-1}{2}$, we will use the notation $\rdec$ instead of $\rdec_{\frac{q-1}{2}}$, and $\mathbf{H}$ instead of $\mathbf{H}_{\frac{q-1}{2}}$.

\subsection{A Variant of the Ducas-Micciancio Signature scheme}\label{subsection:DM-signatures}

We recall the stateful and adaptively secure version of Ducas-Micciancio signature scheme~\cite{DM14,DM14-eprint}, which is used to enroll new users in our construction.

Following~\cite{DM14,DM14-eprint}, throughout this work, for any real constants $c>1$ and $\alpha_0\geq \frac{1}{c-1}$, define a series of sets $\mathcal{T}_j=\{0,1\}^{c_j}$ of lengths $c_j=\lfloor \alpha_0 c^{j}\rfloor$ for $j\in[d]$, where $d\geq \log_c(\omega(\log n))$. For each tag $t=(t_0,t_1,\ldots,t_{c_j})^{\top}\in\mathcal{T}_j$ for $j\in[d]$, associate it with a ring element $t(X)=\sum_{k=0}^{c_j}t_k\cdot X^k\in R_q$. Let $c_0=0$ and then define $t_{[i]}(X)=\sum_{k=c_{i-1}}^{c_i-1}t_k\cdot X^k$ and $t_{[i]}=(t_{c_{i-1}},\ldots,t_{c_{i}-1})^\top$ for $i\in[j]$. Then one can check $t=(t_{[1]}\|t_{[2]}\|\cdots\|t_{[j]})$ and $t(X)=\sum_{i=1}^{j}t_{[i]}(X)$. \smallskip

This variant works with the following parameters.
\begin{itemize}
\item   Let $n,m,q,k$ be some positive integers such that $n\geq 4$ is a power of $2$, $m\geq 2\lceil\log q\rceil +2$, and $q=3^k$. Define the rings $R=\mathbb{Z}[X]/(X^n+1)$ and $R_q=R/qR$.

\item Let the message dimension be $m_s =\mathrm{poly}(n)$. Also, let $\ell=\lfloor\log \frac{q-1}{2}\rfloor +1$, and $\overline{m} = m + k$ and $\overline{m}_s=m_s\cdot \ell$.

\item Let integer $\beta=\widetilde{\mathcal{O}}(n)$ and integer $d$ and sequence $c_0,\ldots,c_d$ be as above.

\item Let $S\in\mathbb{Z}$ be a state that is $0$ initially.
\end{itemize}
The public verification key consists of  the following: \[
\mathbf{A}, \mathbf{F}_0 \in R_q^{1 \times \overline{m}}; \hspace*{6.8pt}\mathbf{A}_{[0]}, \ldots, \mathbf{A}_{[d]} \in R_q^{1 \times k};\hspace*{6.8pt}
\mathbf{F}\in R_q^{1 \times \ell};\hspace*{6.8pt} \mathbf{F}_1 \in R_q^{1 \times \overline{m}_s};  \hspace*{6.8pt}u \in R_q
\] while the secret signing key is a Micciancio-Peikert~\cite{MP12} trapdoor matrix  $\mathbf{R}\in R_q^{m\times k}$.

When  signing a message $\mathfrak{m}\in R_q^{m_s}$, the signer first computes $\overline{\mathfrak{m}}=\rdec(\mathfrak{m})\in R^{\overline{m}_s}$, whose coefficients are in the set $\{-1,0,1\}$. He  then performs the following steps.
   \begin{itemize}
     \item Set the tag $t=(t_0,t_1\ldots, t_{c_d-1})^\top\in \mathcal{T}_d$, where $S=\sum_{j=0}^{c_d-1} 2^j\cdot t_j$, and
    compute $\mathbf{A}_{t} = [\mathbf{A}|\mathbf{A}_{[0]}+\sum_{i=1}^{d}t_{[i]}\mathbf{A}_{[i]}] \in R_q^{1\times (\overline{m} + k)}$.  Update $S$ to $S+1$. \smallskip
   \item Choose  $\mathbf{r}\in R^{\overline{m}}$ with $\|\mathbf{r}\|_{\infty}\leq \beta$.
   \smallskip
   \item Let $y=\mathbf{F}_0 \cdot \mathbf{r}+\mathbf{F}_1\cdot \overline{\mathfrak{m}}\in R_q$ and  ${u}_{p}=\mathbf{F}\cdot \rdec(y)+u \in R_q$.
    \item Employing the trapdoor matrix $\mathbf{R}$, produce a ring vector $\mathbf{v}\in R^{\overline{m} + k}$  with $\mathbf{A}_t\cdot \mathbf{v}=u_p$ over the ring $R_q$ and $\|\mathbf{v}\|_{\infty}\leq \beta$.

    \item Return the tuple $(t,\mathbf{r},\mathbf{v})$ as a signature for the message $\mathfrak{m}$.

      \end{itemize}
To check the validity of the tuple $(t,\mathbf{r},\mathbf{v})$ with respect to message $\mathfrak{m}\in R_q^{m_s}$, the verifier first computes the matrix $\mathbf{A}_t$ as above and verifies the following conditions:
 \begin{eqnarray*}
  \begin{cases}
    \mathbf{A}_t\cdot\mathbf{v}=\mathbf{F}\cdot\rdec(\mathbf{F}_0\cdot\mathbf{r}+\mathbf{F}_1\cdot\rdec(\mathfrak{m}))+u,\\
    \|\mathbf{r}\|_{\infty}\leq \beta,~~\|\mathbf{v}\|_{\infty}\leq \beta.
  \end{cases}
  \end{eqnarray*}
He outputs $1$ if all these three conditions hold and $0$ otherwise.

 \begin{lemma}[\cite{DM14,DM14-eprint}]
 Given at most polynomially bounded number of signature queries, the above variant is existentially unforgeable against adaptive chosen message attacks assuming the hardness of the $\mathsf{RSIS}_{n,\overline{m},q,\widetilde{\mathcal{O}}(n^2)}$ problem.
\end{lemma}

\subsection{Zero-Knowledge Argument of Knowledge}\label{subsection:Stern}
We will work with statistical zero-knowledge argument systems, namely, interactive protocols where the \textsf{ZK} property holds against \emph{any} cheating verifier, while the soundness property only holds against \emph{computationally bounded} cheating provers. More formally, let the set of statements-witnesses $\mathrm{R} = \{(y,w)\} \in \{0,1\}^* \times \{0,1\}^*$ be an \textsf{NP} relation. A two-party game $\langle \mathcal{P},\mathcal{V} \rangle$ is called an interactive argument system for the relation $\mathrm{R}$ with soundness error $e$ if the following two conditions hold:
\begin{itemize}
    \item {\sf Completeness.} If $(y,w) \in \mathrm{R}$ then $\mathrm{Pr}\big[\langle \mathcal{P}(y,w),\mathcal{V}(y) \rangle =1\big]=1.$
    \item {\sf Soundness.} If  $(y,w) \not \in \mathrm{R}$, then $\forall$ PPT $\widehat{\mathcal{P}}$: \hspace*{2.5pt}$\mathrm{Pr}[\langle \widehat{\mathcal{P}}(y,w),\mathcal{V}(y) \rangle =1] \leq e.$
\end{itemize}
An argument system is called statistical \textsf{ZK} if for any~$\widehat{\mathcal{V}}(y)$, there exists a PPT simulator $\mathcal{S}(y)$ having oracle access to~$\widehat{\mathcal{V}}(y)$ and producing a simulated transcript that is statistically close to the one of the real interaction between $\mathcal{P}(y,w)$ and $\widehat{\mathcal{V}}(y)$. A related notion is argument of knowledge, which, for three-move protocols (commitment-challenge-response), requires the existence of a PPT extractor taking as input a set of valid transcripts with respect to all possible values of the ``challenge'' to the same ``commitment'' and outputting $w'$ such that $(y,w') \in \mathrm{R}$.
\smallskip

The statistical zero-knowledge arguments of knowledge ($\mathsf{ZKAoK}$) presented in this work are Stern-like~\cite{Ste96} protocols. In particular, they are $\Sigma$-protocols in the generalized sense defined in~\cite{JKPT12,BCKLN14}
 (where~$3$ valid transcripts are needed for extraction, instead of just~$2$). Stern's protocol was originally proposed in the context of code-based cryptography, and was later adapted into the lattice setting by Kawachi et al.~\cite{KTX08}. Subsequently, it was empowered by Ling et al.~\cite{LNSW13} to handle the matrix-vector relations where the secret vectors are of small infinity norm, and further developed to design various lattice-based schemes. Libert et al.~\cite{LLMNW16-dgs} put forward an abstraction of Stern's protocol to capture a wider range of lattice-based relations. Now let us recall it. \smallskip

\noindent{\bf{An Abstraction of Stern's Protocol.}} Let integers $q,K, L$ be positive such that $L\geq K$ and $q \geq 2$, and $\mathsf{VALID}\subset\{-1,0,1\}^L$. Given a finite set $\mathcal{S}$,  associate every $\eta \in \mathcal{S}$ with a permutation $\Gamma_\eta$ of $L$ elements such that the following conditions hold:
\begin{eqnarray}\label{eq:zk-equivalence}
\begin{cases}
\mathbf{w} \in \mathsf{VALID} \hspace*{2.5pt} \Longleftrightarrow \hspace*{2.5pt} \Gamma_\eta(\mathbf{w}) \in \mathsf{VALID}, \\
\text{If } \mathbf{w} \in \mathsf{VALID} \text{ and } \eta \text{ is uniform in } \mathcal{S}, \text{ then }  \Gamma_\eta(\mathbf{w}) \text{ is uniform in } \mathsf{VALID}.
\end{cases}
\end{eqnarray}
Our target is  to construct a statistical $\mathsf{ZKAoK}$ for the abstract relation $\mathrm{R_{abstract}}$ of the following form:
\begin{eqnarray*}
\mathrm{R_{abstract}} = \big\{(\mathbf{M}, \mathbf{u}), \mathbf{w} \in \mathbb{Z}_q^{K \times L} \times \mathbb{Z}_q^K \times \mathsf{VALID}: \mathbf{M}\cdot \mathbf{w} = \mathbf{u} \bmod q.\big\}
\end{eqnarray*}

To obtain the desired $\mathsf{ZKAoK}$ protocol, one has to prove that $\mathbf{w}\in\mathsf{VALID}$ and $\mathbf{w}$ satisfies the linear equation $\mathbf{M}\cdot \mathbf{w} = \mathbf{u} \bmod q$. To prove $\mathbf{w}\in\mathsf{VALID}$ in a zero-knowledge manner, the prover chooses $\eta \xleftarrow{\$}\mathcal{S}$ and allows the verifier to check $\Gamma_\eta(\mathbf{w}) \in \mathsf{VALID}$. According to the first condition in~(\ref{eq:zk-equivalence}), the verifier should be convinced that $\mathbf{w}$ is indeed from the set $\mathsf{VALID}$. At the same time, the verifier cannot learn any extra information about $\mathbf{w}$ due to the second condition in~(\ref{eq:zk-equivalence}). Furthermore, to prove in $\mathsf{ZK}$ that the linear equation holds, the prover first chooses $\mathbf{r}_w\xleftarrow{\$}\mathbb{Z}_q^{L}$ as a masking vector and then shows the verifier that the equation $\mathbf{M}\cdot (\mathbf{w} + \mathbf{r}_w) = \mathbf{M}\cdot \mathbf{r}_w + \mathbf{u} \bmod q$ holds.

In Figure~\ref{Figure:Interactive-Protocol}, we describe in details  the interaction between two $\mathrm{PPT}$ algorithms prover $\mathcal{P}$ and verifier $\mathcal{V}$. The system utilizes a  statistically hiding and computationally binding string commitment scheme $\mathsf{COM}$ (e.g., the $\mathsf{RSIS}$-based scheme from~\cite{KTX08}).

\begin{figure}[!htbp]

\begin{enumerate}
  \item \textbf{Commitment:} Prover chooses $\mathbf{r}_w \xleftarrow{\$} \mathbb{Z}_q^L$, $\eta \xleftarrow{\$} \mathcal{S}$ and randomness $\rho_1, \rho_2, \rho_3$ for $\mathsf{COM}$.
Then he sends $\mathrm{CMT}= \big(C_1, C_2, C_3\big)$ to the verifier, where
    \begin{gather*}
        C_1 =  \mathsf{COM}(\eta, \mathbf{M}\cdot \mathbf{r}_w \bmod q; \rho_1), \hspace*{5pt}
        C_2 =  \mathsf{COM}(\Gamma_{\eta}(\mathbf{r}_w); \rho_2), \\
        C_3 =  \mathsf{COM}(\Gamma_{\eta}(\mathbf{w} + \mathbf{r}_w \bmod q); \rho_3).
    \end{gather*}

  \item \textbf{Challenge:} $\mathcal{V}$ sends back a challenge $Ch \xleftarrow{\$} \{1,2,3\}$ to $\mathcal{P}$.
  \item \textbf{Response:} According to the choice of $Ch$, $\mathcal{P}$ sends back  $\mathrm{RSP}$ computed in the following way:
\begin{itemize}
\item $Ch = 1$: Let $\mathbf{t}_{w} = \Gamma_{\eta}(\mathbf{w})$, $\mathbf{t}_{r} = \Gamma_{\eta}(\mathbf{r}_w)$, and $\mathrm{RSP} = (\mathbf{t}_w, \mathbf{t}_r, \rho_2, \rho_3)$. \smallskip

\item $Ch = 2$: Let $\eta_2 = \eta$, $\mathbf{w}_2 = \mathbf{w} + \mathbf{r}_w \bmod q$, and
    $\mathrm{RSP} = (\eta_2, \mathbf{w}_2, \rho_1, \rho_3)$. \smallskip
\item $Ch = 3$: Let $\eta_3 = \eta$, $\mathbf{w}_3 = \mathbf{r}_w$, and
 $\mathrm{RSP} = (\eta_3, \mathbf{w}_3, \rho_1, \rho_2)$.
\end{itemize}
\end{enumerate}
\textbf{Verification:}  When receiving $\mathrm{RSP}$ from $\mathcal{P}$, $\mathcal{V}$ performs as follows:
          \begin{itemize}
            \item $Ch = 1$: Check that $\mathbf{t}_w \in \mathsf{VALID}$, $C_2 = \mathsf{COM}(\mathbf{t}_r; \rho_2)$, ${C}_3 = \mathsf{COM}(\mathbf{t}_w + \mathbf{t}_r \bmod q; \rho_3)$. \smallskip

             \item $Ch = 2$: Check that $C_1 = \mathsf{COM}(\eta_2, \mathbf{M}\cdot \mathbf{w}_2 - \mathbf{u} \bmod q; \rho_1)$, ${C}_3 = \mathsf{COM}(\Gamma_{\eta_2}(\mathbf{w}_2); \rho_3)$. \smallskip

            \item $Ch = 3$: Check that $C_1 =  \mathsf{COM}(\eta_3, \mathbf{M}\cdot \mathbf{w}_3; \rho_1), \hspace*{5pt}
        C_2 =  \mathsf{COM}(\Gamma_{\eta_3}(\mathbf{w}_3); \rho_2).$

          \end{itemize}
          In each case, $\mathcal{V}$ returns $1$ if and only if all the conditions hold.
\caption{Stern-like $\mathsf{ZKAoK}$ for the relation $\mathrm{R_{abstract}}$.}
\label{Figure:Interactive-Protocol}
\end{figure}

\begin{theorem}[\cite{LLMNW16-dgs}]\label{Theorem:zk-protocol}
Let $\mathsf{COM}$ be a statistically hiding and computationally binding string commitment scheme. Then  the interactive protocol depicted in Figure~\ref{Figure:Interactive-Protocol} is a statistical \emph{$\mathsf{ZKAoK}$} with perfect completeness, soundness error~$2/3$, and communication cost~$\mathcal{O}(L\log q)$. Specifically:
\begin{itemize}
\item There exists a polynomial-time simulator  that on input $(\mathbf{M}, \mathbf{u})$,  with probability $2/3$ it outputs an accepted transcript that is within statistical distance from the one produced by an honest prover who knows the witness. \smallskip
\item  There exists a polynomial-time algorithm that, takes as inputs $(\mathbf{M}, \mathbf{u})$ and three accepting transcripts on $(\mathbf{M}, \mathbf{u})$, $(\mathrm{CMT},1,\mathrm{RSP}_1)$, $(\mathrm{CMT},2,\mathrm{RSP}_2)$, and $(\mathrm{CMT},3,\mathrm{RSP}_3)$, outputs   $\mathbf{w}' \in \mathsf{VALID}$ such that $\mathbf{M}\cdot \mathbf{w}' = \mathbf{u} \bmod q$.
\end{itemize}
\end{theorem}
The details of the proof appeared in~\cite{LLMNW16-dgs} and are omitted here.

\subsection{The Refined Permuting Techniques by Ling et al.}\label{subsection:permutation-technique-by-lnwx}

We next recall the permuting techniques recently suggested by Ling et al.~\cite{LNWX18}, which will be used throughout this paper.  \smallskip

\noindent
{\bf Proving that $z \in \{-1,0,1\}$. } Let $b$ an integer.
Denote the integer $b' \in \{-1,0,1\}$ with $b' = b \bmod 3$ as $[b]_3$.  For any $z \in \{-1,0,1\}$, define vector $\mathsf{enc}_3(z)$ in the following manner:
\[
\mathsf{enc}_3(z) = \big([z+1]_3, [z]_3, [z-1]_3\big)^\top  \in \{-1,0,1\}^3.
\]
Namely, $\mathsf{enc}_3(-1) = (0, -1,1)^\top$, $\mathsf{enc}_3(0) = (1, 0,-1)^\top$ and $\mathsf{enc}_3(1) = (-1, 1,0)^\top$.

Let  $e \in \{-1,0,1\}$, define a permutation $\pi_e$ associated to $e$ as follows. It  transforms    vector $\mathbf{v} = (v^{(-1)}, v^{(0)}, v^{(1)})^\top \in \mathbb{Z}^3$ into vector
\[
\pi_e(\mathbf{v}) = (v^{([-e-1]_3)}, v^{([-e]_3)}, v^{([-e+1]_3)})^\top.
\]

It is then verifiable that, for any $z, e \in \{-1,0,1\}$, the equivalence below holds.
\begin{eqnarray}\label{eq:equivalence-enc-3}
\mathbf{v} = \mathsf{enc}_3(z) \hspace*{6.8pt}\Longleftrightarrow\hspace*{6.8pt} \pi_e(\mathbf{v}) = \mathsf{enc}_3([z+e]_3).
\end{eqnarray}

In the context of Stern's protocol, the above equivalence allows us to prove knowledge of $z\in \{-1,0,1\}$, where~$z$ may have other constrains. Towards it, we simply extend  $z$ to $\mathsf{enc}_3(z)$, sample a uniform $e\in\{-1,0,1\}$, and then show the verifier $\pi_{e}(\mathsf{enc}_3(z))$ is of the form $\mathsf{enc}_3([z+e]_3)$. Due to the equivalence in~(\ref{eq:equivalence-enc-3}), the verifier should be convinced that $z$ is in the set $\{-1,0,1\}$. Furthermore, the ``one time pad'' $e$ fully hides the value of  $z$. More importantly, the above technique is extendable so that we can employ the same $e$ for other positions where $z$ appears. An example of that is to prove that $z$ is involved in a product $t \cdot z$, which we now recall.

\smallskip
\noindent
{\bf Proving that $y = t \cdot z$. }
Let $b \in \{0,1\}$, denote  the bit $1-b$ as $\overline{b}$ and the addition operation modulo $2$ as~$\oplus$.

For any $t \in \{0,1\}$ and $z \in \{-1,0,1\}$, let vector $\mathsf{ext}(t,z) \in \{-1,0,1\}^6$ be of the following form:
\begin{eqnarray*}
\hspace*{-8pt}
\mathsf{ext}(t,z) = \big(\hspace*{2.8pt}
\overline{t}\cdot [z\hspace*{-1.5pt}+\hspace*{-1.5pt}1]_3, \hspace*{4.8pt}
t \cdot [z\hspace*{-1.5pt}+\hspace*{-1.5pt}1]_3, \hspace*{4.8pt}
\overline{t} \cdot [z]_3, \hspace*{4.8pt}
t \cdot [z]_3, \hspace*{4.8pt}
\overline{t} \cdot [z\hspace*{-1.5pt}-\hspace*{-1.5pt}1]_3, \hspace*{4.8pt}
t \cdot [z\hspace*{-1.5pt}-\hspace*{-1.5pt}1]_3\hspace*{2.8pt}
\big)^\top.
\end{eqnarray*}
Let $b \in \{0,1\}$ and $e \in \{-1,0,1\}$, define the permutation $\psi_{b,e}(\cdot)$ associated to $b,e$ as follows. It  transforms  vector $$\mathbf{v} =
\big(v^{(0, -1)}, v^{(1,-1)}, v^{(0,0)}, v^{(1,0)}, v^{(0,1)}, v^{(1,1)}\big)^\top \in \mathbb{Z}^6$$
into vector $\psi_{b,e}(\mathbf{v})$ of form
\[
\psi_{b,e}(\mathbf{v}) = \big(
v^{(b, [-e-1]_3)}, \hspace*{1.6pt}
v^{(\overline{b}, [-e-1]_3)}, \hspace*{1.6pt}
v^{(b, [-e]_3)}, \hspace*{1.6pt}
v^{(\overline{b}, [-e]_3)}, \hspace*{1.6pt}
v^{(b, [-e+1]_3)}, \hspace*{1.6pt}
v^{(\overline{b}, [-e+1]_3)}
\big)^\top.
\]

It can be easily checked that for any $t, b \in \{0,1\}$ and any $z,e \in \{-1,0,1\}$, the following equivalence is satisfied.
\begin{eqnarray}\label{eq:zk-product-equiv}
\mathbf{v} = \mathsf{ext}(t,z)
\hspace*{6.8pt}\Longleftrightarrow \hspace*{6.8pt}
\psi_{b,e}(\mathbf{v}) = \mathsf{ext}(\hspace*{1.6pt}t \oplus b, \hspace*{1.6pt}[z + e]_3\hspace*{1.6pt}).
\end{eqnarray}


The same as in the case $z\in\{-1,0,1\}$, the above equivalence~(\ref{eq:zk-product-equiv}) allows us to prove knowledge of $y$, where $y$ is a product of secret integers $ t\in\{0,1\}$ and $z\in\{-1,0,1\}$.



Next, we recall the generalizations of the above two core techniques to prove knowledge of vector $\mathbf{z} \in \{-1,0,1\}^{\mathfrak{m}}$ as well as vector of the form~(\ref{equation:mix-form}). 

\smallskip
\noindent
{\bf Proving that $\mathbf{z} \in \{-1,0,1\}^{\mathfrak{m}}$. } We first generalize the notion $[b]_3$ to $[\mathbf{b}]_3$ for any $\mathbf{b} \in \mathbb{Z}^{\mathfrak{m}}$, where $[\mathbf{b}]_3$ is the vector $\mathbf{b}'$ such that $\mathbf{b}'=\mathbf{b}\bmod 3$ coordinate-wise.

For $\mathbf{z} = ({z}_1, \ldots, z_{\mathfrak{m}})^\top \in \{-1,0,1\}^{\mathfrak{m}}$, define the following extension:
\[
\mathsf{enc}(\mathbf{z}) = \big(
\hspace*{2.6pt}\mathsf{enc}_3(z_1) \hspace*{2.6pt}\| \cdots \| \hspace*{2.6pt}\mathsf{enc}_3(z_{\mathfrak{m}})\hspace*{2.6pt}
\big) \in \{-1,0,1\}^{3\mathfrak{m}}.
\]

Let $\mathbf{e} = ({e}_1, \ldots, e_{\mathfrak{m}})^\top\in \{-1,0,1\}^{\mathfrak{m}}$, define the permutation $\Pi_{\mathbf{e}}$ associated to $\mathbf{e}$ as follows. It maps  vector $\mathbf{v} = (\mathbf{v}_1 \| \ldots \| \mathbf{v}_{\mathfrak{m}}) \in \mathbb{Z}^{3\mathfrak{m}}$ consisting of $\mathfrak{m}$ blocks of size $3$ to vector as follows:
\[
\Pi_{\mathbf{e}}(\mathbf{v}) = \big(
\pi_{e_1}(\mathbf{v}_1) \| \ldots \| \pi_{e_{\mathfrak{m}}}(\mathbf{v}_{\mathfrak{m}})
\big).
\]
Following~(\ref{eq:equivalence-enc-3}), for any $\mathbf{z}, \mathbf{e} \in \{-1,0,1\}^{\mathfrak{m}}$, we obtain the  following equivalence:
\begin{eqnarray}\label{eq:equivalence-enc-vector}
\mathbf{v} = \mathsf{enc}(\mathbf{z}) \hspace*{6.8pt}\Longleftrightarrow \hspace*{6.8pt} \Pi_{\mathbf{e}}(\mathbf{v}) = \mathsf{enc}([\mathbf{z} + \mathbf{e}]_3).
\end{eqnarray}

\noindent
{\bf Handling a ``mixing'' vector. } We now deal with a ``mixing'' vector of the following form:
\begin{eqnarray}\label{equation:mix-form}
\mathbf{y} = \big(\hspace*{2.6pt}\mathbf{z} \hspace*{2.6pt}\|\hspace*{2.6pt} t_0 \cdot \mathbf{z} \hspace*{2.6pt}\| \hspace*{2.6pt}\ldots \hspace*{2.6pt}\|\hspace*{2.6pt} t_{c_d-1} \cdot \mathbf{z} \hspace*{2.6pt}\big),
\end{eqnarray}
where $\mathbf{z}\in\{-1,0,1\}^{\mathfrak{m}}$ and $t=(t_0,t_1,\ldots, t_{c_d-1})^{\top}\in\{0,1\}^{c_d}$ for $\mathfrak{m}, c_d\in \mathbb{Z}^{+}$.
\smallskip

First, we define the extension vector
$
\mathsf{mix}(\mathbf{t}, \mathbf{z}) \in \{-1,0,1\}^{3\mathfrak{m} + 6\mathfrak{m}c_d}
$
of vector $\mathbf{y}$ in the following manner:
\begin{align*}
\big(\hspace*{2.6pt}
\mathsf{enc}(\mathbf{z}) \hspace*{2.6pt}\|\hspace*{2.6pt}
\mathsf{ext}(t_0, z_1) \hspace*{2.6pt}\| \ldots \|\hspace*{2.6pt} \mathsf{ext}(t_0, z_{\mathfrak{m}}) \hspace*{2.6pt}\|
\ldots \|\hspace*{2.6pt}
\mathsf{ext}(t_{c_d-1}, z_1) \hspace*{2.6pt}\| \ldots \|\hspace*{2.6pt} \mathsf{ext}(t_{c_d-1}, z_{\mathfrak{m}})
\hspace*{2.6pt}\big).
\end{align*}
Next, for $\mathbf{b}= (b_0, \cdots, b_{c_d-1})^\top \in \{0,1\}^{c_d}$ and $\mathbf{e} = (\mathbf{e}_1, \ldots, e_{\mathfrak{m}})^{\top} \in \{-1,0,1\}^{\mathfrak{m}}$, we define the permutation $\Psi_{\mathbf{b}, \mathbf{e}}$ that works as follows. It maps vector $\mathbf{v}\in \mathbb{Z}^{3\mathfrak{m} + 6\mathfrak{m}c_d}$ of form
\[\mathbf{v} =
 \big(
\mathbf{v}_{-1} \hspace*{2.6pt}\|\hspace*{2.6pt} \mathbf{v}_{0,1} \hspace*{2.6pt}\| \ldots \|\hspace*{2.6pt} \mathbf{v}_{0, \mathfrak{m}} \hspace*{2.6pt}\| \ldots \|\hspace*{2.6pt} \mathbf{v}_{c_d-1, 1}\hspace*{2.6pt} \| \ldots \|\hspace*{2.6pt} \mathbf{v}_{c_d-1, \mathfrak{m}}
\big),
\]
where block $\mathbf{v}_{-1}$ has length $3\mathfrak{m}$ and each block $\mathbf{v}_{i,j}$ has length $6$,
to vector $\Psi_{\mathbf{b}, \mathbf{e}}(\mathbf{v})$ of form
\begin{eqnarray*}
\Psi_{\mathbf{b}, \mathbf{e}}(\mathbf{v})=
\big(
\Pi_{\mathbf{e}}(\mathbf{v}_{-1}) \| &&\psi_{b_0, e_1}(\mathbf{v}_{0,1}) \| \ldots \| \psi_{b_0, e_{\mathfrak{m}}}(\mathbf{v}_{0, \mathfrak{m}}) \| \ldots \|\\
&&\psi_{b_{c_d-1}, e_1}(\mathbf{v}_{c_d-1, 1}) \| \ldots \| \psi_{b_{c_d-1}, e_{\mathfrak{m}}}(\mathbf{v}_{c_d-1, \mathfrak{m}})
\big).
\end{eqnarray*}
Then, for all $\mathbf{t}, \mathbf{b} \in \{0,1\}^{c_d}$ and $\mathbf{z}, \mathbf{e} \in \{-1,0,1\}^{\mathfrak{m}}$, one can check the following equivalence holds:
\begin{eqnarray}\label{eq:equivalence-mix}
\mathbf{v} = \mathsf{mix}(\mathbf{t},\mathbf{z})
\hspace*{1.8pt}\Longleftrightarrow \hspace*{1.8pt}
\Psi_{\mathbf{b},\mathbf{e}}(\mathbf{v}) = \mathsf{mix}(\hspace*{1.6pt}\mathbf{t} \oplus \mathbf{b}, \hspace*{1.6pt}[\mathbf{z} + \mathbf{e}]_3\hspace*{1.6pt}).
\end{eqnarray}

\subsection{Zero-Knowledge Protocol for the Ducas-Micciancio Signature}\label{subsection:zk-for-DM}
We now recall the statistical zero-knowledge argument of knowledge of a valid message-signature pair for the Ducas-Micciancio signature, as presented in~\cite{LNWX18}. Let $n,q, m, k, \overline{m},\overline{m}_s, \ell, \beta, d, c_0, \ldots, c_d$ as specified in Section~\ref{subsection:DM-signatures}.
The protocol is summarized below.
\begin{itemize}
\item The public input consists of
\begin{eqnarray*}
\mathbf{A}, \mathbf{F}_0 \in R_q^{1 \times \overline{m}}; \hspace*{6.8pt}\mathbf{A}_{[0]}, \ldots, \mathbf{A}_{[d]} \in R_q^{1 \times k};
\mathbf{F} \in R_q^{1 \times \ell}; \hspace*{6.8pt} \mathbf{F}_1 \in R_q^{1\times \overline{m}_s};\hspace*{6.8pt}u \in R_q.
\end{eqnarray*}

\item The secret input of the prover consists of message $\mathfrak{m} \in R_q^{m_s}$ and signature $(t, \mathbf{r}, \mathbf{v})$, where
\begin{eqnarray*}\hspace*{-10pt}
    \begin{cases}
    t = (t_0, \ldots, t_{c_1 -1}, \ldots , t_{c_{d-1}}, \ldots, t_{c_d -1})^\top \in \{0,1\}^{c_d}; \\[2.6pt]
    \mathbf{r} \in R^{\overline{m}}; \hspace*{6.8pt}\mathbf{v} = (\mathbf{s} \| \mathbf{z})\in R^{\overline{m} + k};
    \hspace*{6.8pt} \mathbf{s} \in R^{\overline{m}}; \hspace*{6.8pt}
    \mathbf{z} \in R^k;
    \end{cases}
    \end{eqnarray*}
\item The goal of the prover is to prove in $\mathsf{ZK}$ that $\|\mathbf{r}\|_\infty \leq \beta$, $\|\mathbf{v}\|_\infty \leq \beta$, and that the following equation
    \begin{eqnarray}\label{eq:protocol-initial}
\mathbf{A} \cdot \mathbf{s} + \mathbf{A}_{[0]} \cdot \mathbf{z} + \sum_{i=1}^d \mathbf{A}_{[i]} \cdot t_{[i]} \cdot \mathbf{z}  =  \mathbf{F}\cdot \mathbf{y} + u
\end{eqnarray}
holds for $\big\{t_{[i]} = \sum_{j = c_{i-1}}^{{c_i -1}} t_j \cdot X^j\big\}_{i=1}^d$ and
\begin{eqnarray}\label{eq:protocol-cham-initial}
\mathbf{y} = \rdec\left(\mathbf{F}_0 \cdot \mathbf{r} + \mathbf{F}_1 \cdot \rdec(\mathfrak{m})\right) \in R^\ell.
\end{eqnarray}
\end{itemize}
The next step is to transform the secret input into a vector $\mathbf{w}$ that belongs to a specific set $\mathsf{VALID}$ and reduce the considered statements~(\ref{eq:protocol-initial}) and~(\ref{eq:protocol-cham-initial}) into $\mathbf{M}\cdot \mathbf{w}=\mathbf{u}\bmod q$ for some public input $\mathbf{M},\mathbf{u}$, in the form of the abstract protocol from Section~\ref{subsection:Stern}. To realize this, we employ the following two steps.

\smallskip

\noindent
{\sc Decomposing-Unifying.}
To begin with, we utilize  the notations $\mathsf{rot}$ and~$\tau$ from Section~\ref{subsection:rings} and the decomposition techniques from Section~\ref{subsection:decomposition}. 
\smallskip

Let $\mathbf{s}^\star = \tau(\rdec_\beta(\mathbf{s})) \in \{-1,0,1\}^{n \overline{m}\delta_\beta}$, $\mathbf{z}^\star = \tau(\rdec_\beta(\mathbf{z})) \in \{-1,0,1\}^{nk\delta_\beta}$ and $\mathbf{r}^\star = \tau(\rdec_\beta(\mathbf{r})) \in \{-1,0,1\}^{n\overline{m}\delta_\beta}$. Then, one can check that, equation~(\ref{eq:protocol-initial}) is equivalent to,
\begin{eqnarray*}
\nonumber&& [\mathsf{rot}(\mathbf{A}_{[0]})\cdot \mathbf{H}_{{k}, \beta}]\cdot \mathbf{z}^\star +
\sum_{i=1}^d \sum_{j=c_{i-1}}^{c_i -1} [\mathsf{rot}(\mathbf{A}_{[i]}\cdot X^{j}) \cdot \mathbf{H}_{k, \beta}]\cdot t_j \cdot \mathbf{z}^\star + \\[2.6pt]
&&[\mathsf{rot}(\mathbf{A})\cdot \mathbf{H}_{\overline{m}, \beta}]\cdot \mathbf{s}^\star - [\mathsf{rot}(\mathbf{F})]\cdot \tau(\mathbf{y}) = \tau(u) \bmod q,
\end{eqnarray*}
and equation~(\ref{eq:protocol-cham-initial}) is equivalent to
\begin{eqnarray*}
[\mathsf{rot}(\mathbf{F}_{0})\cdot \mathbf{H}_{\overline{m}, \beta}] \cdot \mathbf{r}^\star + [\mathsf{rot}(\mathbf{F}_1)]\cdot \tau(\rdec(\mathfrak{m})) - [\mathbf{H}]\cdot \tau(\mathbf{y}) = \mathbf{0} \bmod q.
\end{eqnarray*}

Rearrange the two derived equations using some basic algebra, we are able to obtain the following unifying equation:
\[
\mathbf{M}_0 \cdot \mathbf{w}_0 = \mathbf{u} \bmod q,
\]
where $\mathbf{u} = (\tau(u) \hspace*{2.6pt}\|\hspace*{2.6pt} \mathbf{0}) \in \mathbb{Z}_q^{2n}$ and $\mathbf{M}_0$ are built from public input, and  $\mathbf{w}_0 = (\mathbf{w}_1 \hspace*{2.6pt}\|\hspace*{2.6pt} \mathbf{w}_2)$ is built from secret input with $\mathbf{w}_1 \in \{-1,0,1\}^{(k\delta_\beta + c_d k\delta_\beta)n}$ and $\mathbf{w}_2\in \{-1,0,1\}^{2n\overline{m}\delta_\beta+ n\ell+n\overline{m}_s}$ and
\begin{eqnarray*}
\begin{cases}
\mathbf{w}_1 = (\mathbf{z}^\star \hspace*{3.6pt}\|\hspace*{3.6pt} t_0 \cdot \mathbf{z}^\star
\hspace*{3.6pt}\|\hspace*{3.6pt} \ldots \hspace*{3.6pt}\|\hspace*{3.6pt} t_{c_d-1}\cdot \mathbf{z}^\star); \\[1.6pt]
\mathbf{w}_2 = (\mathbf{s}^\star \hspace*{3.6pt}\|\hspace*{3.6pt} \mathbf{r}^\star \hspace*{3.6pt}\|\hspace*{3.6pt}  \tau(\mathbf{y}) \hspace*{3.6pt}\|\hspace*{3.6pt} \tau(\rdec(\mathfrak{m}))).
\end{cases}
\end{eqnarray*}
Until now, we have transformed the secret input into a vector $\mathbf{w}_0$ whose coefficients are in the set $\{-1,0,1\}$ and reduced statements~(\ref{eq:protocol-initial}) and~(\ref{eq:protocol-cham-initial}) into $\mathbf{M}_0\cdot \mathbf{w}_0 = \mathbf{u} \bmod q$, where $\mathbf{M}_0, \mathbf{u}$ are public.

\noindent
{\sc Extending-Permuting. }
Now the target is to transform the secret vector $\mathbf{w}_0$ to a vector $\mathbf{w}$ such that the conditions in~(\ref{eq:zk-equivalence}) hold. Towards this goal, the extension and permutation techniques described in Section~\ref{subsection:permutation-technique-by-lnwx} is employed.

We first extend $\mathbf{w}_0 = (\mathbf{w}_1 \| \mathbf{w}_2)$ as follows.
\begin{align}\label{eq:zk-w'_1}
\mathbf{w}_1 &\mapsto \mathbf{w}'_1 = \mathsf{mix}\big(t, \mathbf{z}^\star\big) \in \{-1,0,1\}^{L_1}; \\
\nonumber\mathbf{w}_2 &\mapsto \mathbf{w}'_2 = \mathsf{enc}(\mathbf{w}_2) \in \{-1,0,1\}^{L_2}.
\end{align}
Then form a new vector $\mathbf{w} = (\mathbf{w}'_1 \| \mathbf{w}'_2) \in \{-1,0,1\}^{L}$, where
$ L = L_1 + L_2$  and
$$L_1 = (k\delta_\beta + 2c_d k \delta_\beta)3n;  \hspace*{2.8pt}
L_2 = 6n\overline{m}\delta_\beta+ 3n\ell+3n\overline{m}_s.$$
According to the extension, adding suitable zero-columns to $\mathbf{M}_0$ to obtain a new matrix $\mathbf{M} \in \mathbb{Z}_q^{2n \times L}$ such that $\mathbf{M} \cdot \mathbf{w} = \mathbf{M}_0 \cdot \mathbf{w}_0$.
\smallskip

We are ready to define the set $\mathsf{VALID}$ that consists of our transformed secret vector $\mathbf{w}$, the set $\mathcal{S}$, and the associated permutations $\{\Gamma_\eta: \eta \in \mathcal{S}\}$, such
that the conditions in~(\ref{eq:zk-equivalence}) are all satisfied.
\smallskip

\noindent
Let $\mathsf{VALID}$ be the set of all vectors $\mathbf{v}' = (\mathbf{v}'_1 \| \mathbf{v}'_2) \in \{-1,0,1\}^{L}$ such that the following conditions hold:
\begin{itemize}
\item $\mathbf{v}'_1 = \mathsf{mix}(t, \mathbf{z}^\star)$ for some vectors $t \in \{0,1\}^{c_d}$ and $\mathbf{z}^\star \in \{-1,0,1\}^{nk\delta_\beta}$. 

\item $\mathbf{v}'_2 = \mathsf{enc}(\mathbf{w}_2)$ for vector $\mathbf{w}_2 \in \{-1,0,1\}^{L_2/3}$.
\end{itemize}

It is easy to see that $\mathbf{w}$ belongs to this special set $\mathsf{VALID}$. \smallskip

Now, define $\mathcal{S} = \{0,1\}^{c_d} \times \{-1,0,1\}^{nk\delta_\beta} \times \{-1,0,1\}^{L_2/3}$. For each element $\eta= (\mathbf{b}, \mathbf{e}, \mathbf{f}) \in \mathcal{S}$, define an associated permutation $\Gamma_\eta$ as follows. It permutes vector $\mathbf{v}^\star = (\mathbf{v}_1^\star \| \mathbf{v}_2^\star) \in \mathbb{Z}^L$, where $\mathbf{v}_1^\star \in \mathbb{Z}^{L_1}$ and $\mathbf{v}_2^\star \in \mathbb{Z}^{L_2}$, into vector of the following form:
\[
\Gamma_\eta(\mathbf{v}^\star) = \big(\hspace*{2.6pt}\Psi_{\mathbf{b}, \mathbf{e}}(\mathbf{v}_1^\star) \hspace*{2.6pt}\|\hspace*{2.6pt} \Pi_{\mathbf{f}}(\mathbf{v}_2^\star)\hspace*{2.6pt}\big).
\]
It then follows from the equivalences in~(\ref{eq:equivalence-enc-vector}) and~(\ref{eq:equivalence-mix}) that $\mathsf{VALID}$, $\mathcal{S}$, and $\Gamma_\eta$ satisfy the conditions in~(\ref{eq:zk-equivalence}). Therefore, we have obtained an instance of the abstract protocol from Section~\ref{subsection:Stern}. Up to this point, running the protocol of Figure~\ref{Figure:Interactive-Protocol} results in the desired statistical $\mathsf{ZKAoK}$ protocol. The protocol has perfect completeness, soundness error $2/3$, and communication cost $\mathcal{O}(L \cdot \log q)$, which is of order ${\mathcal{O}}(n \cdot \log^4 n) = \widetilde{\mathcal{O}}(\lambda)$. 

 \subsection{Key-Oblivious Encryption}\label{subsection:KOE-security-model}
  We next recall the definitions of key-oblivious encryption (\textsf{KOE}), as introduced in~\cite{KM15}. A $\mathsf{KOE}$ scheme consists of the following polynomial-time algorithms.

  \begin{description}
  \item[$\mathsf{Setup}(\lambda)$:] On input the security parameter $\lambda$, it outputs public parameter $\mathsf{pp}$. $\mathsf{pp}$ is implicit for all algorithms below if not explicitly mentioned. \smallskip
  \item[$\mathsf{KeyGen}(\mathsf{pp})$:] On input $\mathsf{pp}$, it generates a key pair $(\mathsf{pk},\mathsf{sk})$. \smallskip
  \item[$\mathsf{KeyRand}(\mathsf{pk})$:] On input the public key $\mathsf{pk}$, it outputs a new public key $\mathsf{pk'}$ for the same secret key.\smallskip
  \item[$\mathsf{Enc}(\mathsf{pk},\mathfrak{m})$:] On inputs $\mathsf{pk}$ and a message $\mathfrak{m}$, it outputs a ciphertext $\mathsf{ct}$ on this message.\smallskip
  \item[$\mathsf{Dec}(\mathsf{sk},\mathsf{ct})$:] On inputs $\mathsf{sk}$ and $\mathsf{ct}$, it outputs the decrypted message $\mathfrak{m'}$.

  \end{description}

\noindent{\sc{Correctness.}}
The above scheme must satisfy the following correctness requirement: For all $\lambda$, all $\mathsf{pp}\leftarrow\mathsf{Setup}(\lambda)$, all $(\mathsf{pk}, \mathsf{sk})\leftarrow \mathsf{KeyGen}(\mathsf{pp})$, all $\mathsf{pk'}\leftarrow\mathsf{KeyRand}(\mathsf{pk})$, all $\mathfrak{m}$,
\[\mathsf{Dec}(\mathsf{sk},\mathsf{Enc}(\mathsf{pk'},\mathfrak{m}))=\mathfrak{m}.\]

\smallskip

\noindent
{\sc Security. }
The security requirements of a $\mathsf{KOE}$ scheme consist of \emph{key randomizability} ($\mathsf{KR}$), \emph{plaintext indistinguishability under key randomization} ($\mathsf{INDr}$), and \emph{key privacy under key randomization} ($\mathsf{KPr}$). 

\noindent{\sc{Key Randomizability.}} $\mathsf{KR}$ requires that any 
adversary cannot determine how public keys are related to each other without possession of secret keys. Details are modelled in the experiment $\mathbf{Exp}_{\mathsf{KOE},\mathcal{A}}^{\mathsf{KR}}(\lambda)$ in Fig~\ref{Figure:KOE-security-definition}.

Define the advantage $\mathbf{Adv}_{\mathsf{KOE},\mathcal{A}}^{\mathsf{KR}}(\lambda)$ of adversary $\mathcal{A}$ against $\mathsf{KR}$ of the $\mathsf{KOE}$ scheme as $\vert2\text{Pr}[\mathbf{Exp}_{\mathsf{KOE},\mathcal{A}}^{\mathsf{KR}}(\lambda)=1]-1\vert$.  A $\mathsf{KOE}$ scheme is key randomizable if the advantage of any $\mathrm{PPT}$ adversary $\mathcal{A}$ is negligible. \smallskip

\noindent {\sc{Plaintext indistinguishability under key randomization.}} $\mathsf{INDr}$ requires that any adversary cannot distinguish ciphertext of one message from ciphertext of another one even though the adversary is allowed to choose the two messages and to randomize the public key. Details are modelled in the experiment $\mathbf{Exp}_{\mathsf{KOE},\mathcal{A}}^{\mathsf{INDr}}(\lambda)$ in Fig~\ref{Figure:KOE-security-definition}.

Define the advantage $\mathbf{Adv}_{\mathsf{KOE},\mathcal{A}}^{\mathsf{INDr}}(\lambda)$ of adversary $\mathcal{A}$ against $\mathsf{INDr}$ of the $\mathsf{KOE}$ scheme as $\vert2\text{Pr}[\mathbf{Exp}_{\mathsf{KOE},\mathcal{A}}^{\mathsf{INDr}}(\lambda)=1]-1\vert$.  A $\mathsf{KOE}$ scheme is plaintext indistinguishable under key randomization if the advantage of any $\mathrm{PPT}$ adversary $\mathcal{A}$ is negligible.

\begin{figure}[!htb]
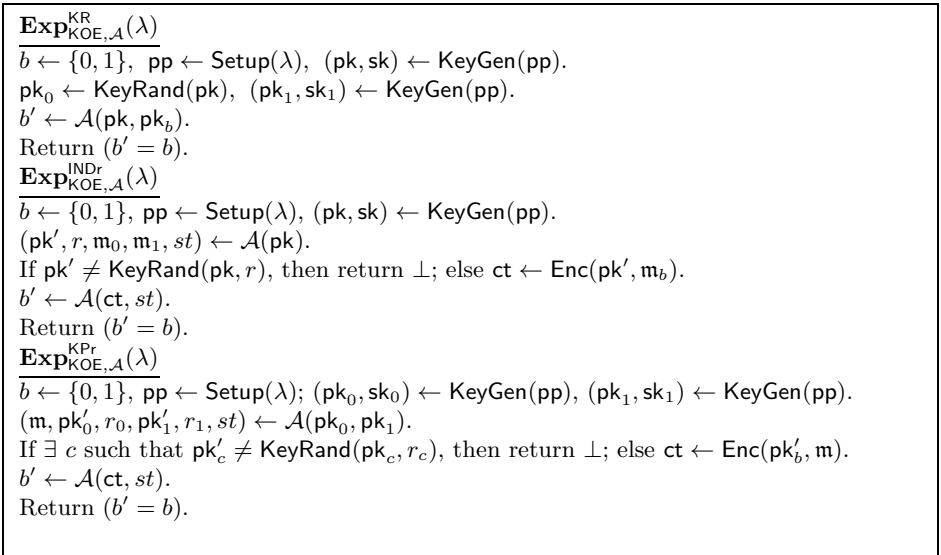

\begin{center}
\begin{tabular}{c}
\begin{minipage}{12cm}
\begin{small}
    \underline{$\mathbf{Exp}_{\mathsf{KOE},\mathcal{A}}^{\mathsf{KR}}(\lambda)$}\\
        $\> b\leftarrow\{0,1\}$,
        $\> \mathsf{pp}\leftarrow\mathsf{Setup}(\lambda)$,
        $\> (\mathsf{pk},\mathsf{sk})\leftarrow \mathsf{KeyGen}(\mathsf{pp})$.\\
        $\>\mathsf{pk}_0\leftarrow\mathsf{KeyRand}(\mathsf{pk})$,
        $\>(\mathsf{pk}_1,\mathsf{sk}_1)\leftarrow \mathsf{KeyGen}(\mathsf{pp})$.\\
        $\>b'\leftarrow\mathcal{A}(\mathsf{pk},\mathsf{pk}_b)$.\\
        $ \>{\mbox{Return}}~(b'=b).$

$\underline{\mathbf{Exp}_{\mathsf{KOE},\mathcal{A}}^{\mathsf{INDr}}(\lambda)}$ \\
        $b\leftarrow\{0,1\}$,
        $\mathsf{pp}\leftarrow\mathsf{Setup}(\lambda)$,
        $(\mathsf{pk},\mathsf{sk})\leftarrow \mathsf{KeyGen}(\mathsf{pp})$.\\
        $(\mathsf{pk}',r,\mathfrak{m}_0,\mathfrak{m}_1,st)\leftarrow\mathcal{A}(\mathsf{pk})$.\\
        If $\mathsf{pk}'\neq \mathsf{KeyRand}(\mathsf{pk},r)$, then return~$\bot$;
        else
        $\mathsf{ct}\leftarrow \mathsf{Enc}(\mathsf{pk'},\mathfrak{m}_b)$.\\
        $b'\leftarrow\mathcal{A}(\mathsf{ct},st)$.\\
        {\mbox{Return}}~$(b'=b).$

$\underline{\mathbf{Exp}_{\mathsf{KOE},\mathcal{A}}^{\mathsf{KPr}}(\lambda)}$ \\
    $b\leftarrow\{0,1\}$,
    $\mathsf{pp}\leftarrow\mathsf{Setup}(\lambda)$;
    $(\mathsf{pk}_0,\mathsf{sk}_0)\leftarrow \mathsf{KeyGen}(\mathsf{pp})$,
    $(\mathsf{pk}_1,\mathsf{sk}_1)\leftarrow \mathsf{KeyGen}(\mathsf{pp})$.\\
    $(\mathfrak{m},\mathsf{pk}'_0,r_0,\mathsf{pk}'_1,r_1,st)\leftarrow\mathcal{A}(\mathsf{pk}_0,\mathsf{pk}_1)$.\\
    If $\exists~c$ such that $\mathsf{pk}'_c\neq \mathsf{KeyRand}(\mathsf{pk}_c,r_c)$, then return~$\bot$;
    else
    $\mathsf{ct}\leftarrow \mathsf{Enc}(\mathsf{pk}'_b,\mathfrak{m})$.\\
    $b'\leftarrow\mathcal{A}(\mathsf{ct},st)$.\\
    {\mbox{Return}}~$(b'=b).$
\end{small}
\end{minipage}
\end{tabular}
\end{center}
\caption{Experiment to define security requirements of a $\mathsf{KOE}$ scheme.}\label{Figure:KOE-security-definition}
\end{figure}

\noindent{\sc{Key privacy under key randomization.}} $\mathsf{KPr}$ requires that  any adversary cannot distinguish ciphertext of a message under one public key from ciphertext of the same message under another public key even though the adversary is allowed to choose the message and to randomize the two public keys. Details are modelled in the experiment $\mathbf{Exp}_{\mathsf{KOE},\mathcal{A}}^{\mathsf{KPr}}(\lambda)$ in Fig~\ref{Figure:KOE-security-definition}.

Define the advantage $\mathbf{Adv}_{\mathsf{KOE},\mathcal{A}}^{\mathsf{KPr}}(\lambda)$ of adversary $\mathcal{A}$ against $\mathsf{INDr}$ of the $\mathsf{KOE}$ scheme as $\vert2\text{Pr}[\mathbf{Exp}_{\mathsf{KOE},\mathcal{A}}^{\mathsf{KPr}}(\lambda)=1]-1\vert$.  A $\mathsf{KOE}$ scheme is key private under key randomization if the advantage of any $\mathrm{PPT}$ adversary $\mathcal{A}$ is negligible.

\subsection{Accountable Tracing Signatures}\label{subsection:ATS-security-model}
We then recall the definition of accountable tracing signature ($\mathsf{ATS}$), as introduced in~\cite{KM15}.
An $\mathsf{ATS}$ scheme involves  a group manager ($\mathsf{GM}$) who also serves as the opening authority (\textsf{OA}), a set of users, who are potential group members. As a standard group signature scheme (e.g.~\cite{BMW03,BSZ05}), $\mathsf{GM}$ is able to identify the signer of a given signature. However, if $\mathsf{GM}$ is able to do so, there is an additional \emph{accounting} mechanism that later reveals which {user he chose to trace} (traceable user).  Specifically, if a user suspects that he was traceable by group manager who had claimed non-traceability of this user, then the user can resort to this mechanism to check whether group manager is honest/accountable or not.
An $\mathsf{ATS}$ scheme consists of the following polynomial-time algorithms.
\begin{description}
\item[$\mathsf{Setup}(\lambda)$:] On input the security parameter $\lambda$, it outputs public parameter $\mathsf{pp}$. $\mathsf{pp}$ is implicit for all algorithms below if not explicitly mentioned. \smallskip

  \item[$\mathsf{GKeyGen}(\mathsf{pp})$:] This algorithm is run by $\mathsf{GM}$. On input $\mathsf{pp}$, $\mathsf{GM}$ generates  group public key $\mathsf{gpk}$ and group secret keys: issue key $\mathsf{ik}$ and opening key $\mathsf{ok}$. \smallskip 

  \item[$\mathsf{UKeyGen}(\mathsf{pp})$:] Given input $\mathsf{pp}$, it outputs a user key pair $(\mathsf{upk},\mathsf{usk})$. 

  \item[$\mathsf{Enroll}(\mathsf{gpk},\mathsf{ik},\mathsf{upk},\mathsf{tr})$:] This algorithm is run by $\mathsf{GM}$. Upon receiving a user public key $\mathsf{upk}$ from a user, $\mathsf{GM}$ determines the value of the bit $\mathsf{tr}\in\{0,1\}$, indicating whether  the user is traceable ($\mathsf{tr}=1$) or not. He then produces a certificate $\mathsf{cert}$ for this user according to his choice of $\mathsf{tr}$.  $\mathsf{GM}$ then registers this user to the group and stores the registration information and the witness $w^{\mathsf{escrw}}$ to the bit $\mathsf{tr}$, and sends $\mathsf{cert}$ to the user. \smallskip

  \item[$\mathsf{Sign}(\mathsf{gpk},\mathsf{cert},\mathsf{usk},M)$:] Given the inputs $\mathsf{gpk}$, $\mathsf{cert}$, $\mathsf{usk}$ and message $M$, this algorithm outputs  a signature $\Sigma$ on this message $M$.
      \smallskip

  \item[$\mathsf{Verify}(\mathsf{gpk},M,\Sigma)$:]  Given the inputs $\mathsf{gpk}$ and the message-signature pair $(M,\Sigma)$, this algorithm outputs $1/0$ indicating whether the signature is valid or not. \smallskip

  \item[$\mathsf{Open}(\mathsf{gpk},\mathsf{ok},M,\Sigma)$:] Given the inputs $\mathsf{gpk}$, $\mathsf{ok}$ and the pair $(M,\Sigma)$, this algorithm returns a user public key $\mathsf{upk}'$ and a proof $\Pi_{\mathsf{open}}$ demonstrating  that user $\mathsf{upk}'$ indeed generated the signature $\Sigma$. 
      In case of $\mathsf{upk}'=\bot$, $\Pi_{\mathsf{open}}=\bot$. \smallskip
  \item[$\mathsf{Judge}(\mathsf{gpk},M,\Sigma,\mathsf{upk}',\Pi_{\mathsf{open}})$:] Given all the inputs, this algorithm outputs $1/0$ indicating whether it accepts the opening result or not. \smallskip

  \item[$\mathsf{Account}(\mathsf{gpk},\mathsf{cert},w^{\mathsf{escrw}},\mathsf{tr})$:] Given all the inputs, this algorithm returns $1$ confirming the choice of $\mathsf{tr}$ and $0$ otherwise.

\end{description}

\noindent{\sc{Correctness.}} The above $\mathsf{ATS}$ scheme requires that: for any honestly generated signature, the $\mathsf{Verify}$ algorithm always outputs $1$. Furthermore, if the user is traceable, then $\mathsf{Account}$ algorithm outputs~$1$ when $\mathsf{tr}=1$, and the $\mathsf{Open}$ algorithm can identify the signer and generate a proof $\Pi_{\mathsf{open}}$ that will be accepted by the $\mathsf{Judge}$ algorithm. On the other hand, if the user is non-traceable, then the $\mathsf{Account}$ algorithm outputs~$1$ when $\mathsf{tr}=0$, and  the $\mathsf{Open}$ algorithm outputs $\bot$.  
\begin{remark}
There is a minor difference between the syntax we describe here and that presented by Kohlweiss and Miers~\cite{KM15}. Specifically, we omit the time epoch when the user joins the group, since we do not consider forward and backward tracing scenarios as   in~\cite{KM15}. 
\end{remark}

\noindent
{\sc Security. }
The security requirements of an $\mathsf{ATS}$ scheme consist of \emph{anonymity under tracing} ($\mathsf{AuT}$), \emph{traceability} ($\mathsf{Trace}$), and \emph{non-frameability} ($\mathsf{NF}$), \emph{anonymity with accountability} ($\mathsf{AwA}$) and \emph{trace-obliviousness} ($\mathsf{TO}$). 

\noindent{\sc{Anonymity under tracing.}} $\mathsf{AuT}$ is the standard anonymity requirement of group signatures (e.g. \cite{BMW03,BSZ05}). It guarantees that even when being traced, users are anonymous to the adversary who does not hold the opening key. 
Details are modelled in the experiment in Figure~\ref{Figure:ATS-AuT}.

\begin{figure}
\begin{tabular}{cc}
\begin{minipage}{4.35cm}
\begin{small}
 \underline{$\mb{Exp}_{\ms{ATS},\mc{A}}^{\ms{AuT}-b}(\lambda)$} 

     \smallskip
        $ \ms{pp}\leftarrow\ms{Setup}(\lambda)$.\\
        $(\ms{gpk},\ms{ik},\mathsf{ok})\leftarrow\ms{GKeyGen}(\ms{pp}).$\\
        $b'\leftarrow \mc{A}^{\ms{Ch},\ms{Open}}(\ms{gpk},\ms{ik})$\\
        $ \text{Return}~b'$.

  $\underline{\mathbf{Oracle}~\mathsf{Open}(M,\Sigma)}$

    \smallskip
    If $\Sigma\in Q$, then return $\bot$,\\
    Else~return \\$(\mathsf{upk},\Pi)\leftarrow \mathsf{Open}(\mathsf{ok},M,\Sigma)$.
\end{small}
\end{minipage}
&
\begin{minipage}{7.65cm}
\begin{small}

 $\underline{\mathbf{Oracle}~\ms{Ch}(\ms{cert}_0,\ms{cert}_1,\ms{usk}_0,\ms{usk}_1,M,w_{0}^{\ms{escrw}},w_{1}^{\ms{escrw}},1)}$

     \smallskip
        $\Sigma_0\leftarrow\ms{Sign}(\ms{gpk},\ms{cert}_0,\ms{usk}_0,M)$.\\
        $\Sigma_1\leftarrow\ms{Sign}(\ms{gpk},\ms{cert}_1,\ms{usk}_1,M)$.\\
        $\text{If}~(\Sigma_0\neq\bot\wedge\Sigma_1\neq\bot~\wedge$\\
        $~~~~\ms{Account}(\ms{gpk},\ms{cert}_0,w_0^{\ms{escrw}},1)~\wedge$\\
        $~~~~\ms{Account}(\ms{gpk},\ms{cert}_1,w_1^{\ms{escrw}},1))$\\
        $~~~~Q\leftarrow Q\cup \{\Sigma_b\}$\\
        $~~~~\text{return}~\Sigma_b$,\\
        $\text{Else~return}~\bot.$
\end{small}
\end{minipage}
\end{tabular}
\caption{Experiment to define anonymity under tracing }\label{Figure:ATS-AuT}
\end{figure}

Define the advantage $\mathbf{Adv}_{\mathsf{ATS},\mathcal{A}}^{\mathsf{AuT}}(\lambda)$ of adversary $\mathcal{A}$ against anonymity under tracing of the $\mathsf{ATS}$ scheme as $\vert\text{Pr}[\mathbf{Exp}_{\mathsf{ATS},\mathcal{A}}^{\mathsf{AuT}-1}(\lambda)=1]-\text{Pr}[\mathbf{Exp}_{\mathsf{ATS},\mathcal{A}}^{\mathsf{AuT}-0}(\lambda)=1]\vert$. An $\mathsf{ATS}$ scheme is anonymous under tracing if the advantage of any $\mathrm{PPT}$ adversary $\mathcal{A}$ is negligible.
\smallskip

\noindent{\sc{Traceability.}} Traceability requires that every valid signature will trace to someone as long as the adversary does not hold both the certificate and user secret key of {a user who is not traceable} (non-traceable user). As pointed out by Kohlweiss and Miers~\cite{KM15}, this is slightly different from the standard traceability game (e.g. \cite{BMW03,BSZ05}), where all users are being traced by $\mathsf{GM}$. In an $\mathsf{ATS}$ scheme, when adversary queries certificate of a user of his choice, challenger will always generate a certificate according to $\mathsf{tr}=1$. In other words, the user of the adversary's choice is a traceable user. This ensures that the adversary does not  hold both certificate and user secret key for a non-traceable user. Details are modelled in the experiment in Figure~\ref{Figure:ATS-Trace}.

\begin{figure}
\begin{tabular}{cc}
\begin{minipage}{5cm}
\begin{small}

 \smallskip
    \underline{$\mb{Exp}_{\ms{ATS},\mc{A}}^{\ms{Trace}}(\lambda)$}

     \smallskip
        $ \ms{pp}\leftarrow\ms{Setup}(\lambda)$.\\
        $(\ms{gpk},\ms{ik},\mathsf{ok})\leftarrow\ms{GKeyGen}(\ms{pp}).$\\
        $(M,\Sigma)\leftarrow \mc{A}^{\ms{UKG},\ms{Enroll},\ms{Sign},\ms{Open}}(\ms{gpk}).$\\
        Return $0$ if $(M,\Sigma)\in Q$ or \\ $~~~\ms{Verify}(\ms{gpk},M,\Sigma)=0$.\\
        Else $(\ms{upk},\Pi)\leftarrow\ms{Open}(\ms{ok},m,\Sigma).$\\
        \hspace*{5pt}Return $1$ if $\ms{upk}=\bot$ or \\$~~~~~~\ms{Judge}(\ms{gpk},M,\Sigma,\ms{upk},\Pi)=0$.\\
        \hspace*{5pt}Else return $0$.

         \underline{$\mathbf{Oracle}~\ms{UKG}(\ms{pp})$}

     \smallskip
        $(\ms{upk},\ms{usk})\leftarrow\ms{UKeyGen}(\ms{pp})$.\\
        $ S[\ms{upk}]=\ms{usk}$.\\
        $\text{Return}~\ms{upk}$.
\end{small}
\end{minipage}
&
\begin{minipage}{7cm}
\begin{small}
 $\underline{\mb{Oracle}~\ms{Enroll}(\ms{upk},\ms{tr})}$

     \smallskip
        Let $\ms{tr'}=(\ms{upk}\notin \text{dom}~S)\in\{0,1\}$.\\
        $(\ms{cert},w^{\ms{escrw}})\leftarrow\ms{Enroll}(\ms{ik},\ms{upk},\ms{tr}\vee\ms{tr'}).$\\ 
        $\text{Return}~\ms{cert}$.\\
    $\underline{\mb{Oracle}~\ms{Sign}(\ms{cert},M)}$

     \smallskip
        $\ms{usk}=S[\ms{cert.upk}]$.\\
        $\text{If}~(\ms{usk}=\bot),~\text{return}~\bot.$\\
        $\text{Else}~~\Sigma\leftarrow\ms{Sign}(\ms{gpk},\ms{cert},\ms{usk},M)$.\\
       $\hspace*{22pt}Q=Q\cup\{(m,\Sigma)\}$. \\
       $\hspace*{22pt}\text{return}~\Sigma$.

       $\underline{\mathbf{Oracle}~\mathsf{Open}(M,\Sigma)}$

    \smallskip
    $(\mathsf{upk},\Pi)\leftarrow \mathsf{Open}(\mathsf{ok},M,\Sigma)$\\
    Return $(\mathsf{upk},\Pi)$.
\end{small}
\end{minipage}\\

\end{tabular}
\caption{Experiment to define traceability. }\label{Figure:ATS-Trace}
\end{figure}

Define the advantage $\mathbf{Adv}_{\mathsf{ATS},\mathcal{A}}^{\mathsf{Trace}}(\lambda)$ of adversary $\mathcal{A}$ against traceability of the $\mathsf{ATS}$ scheme as $\text{Pr}[\mathbf{Exp}_{\mathsf{ATS},\mathcal{A}}^{\mathsf{Trace}}(\lambda)=1]$. An $\mathsf{ATS}$ scheme is traceable if the advantage of any $\mathrm{PPT}$ adversary $\mathcal{A}$ is negligible. \smallskip

\noindent{\sc{Non-frameability.}} It requires that the adversary cannot sign messages on behalf of honest users, even though the adversary can corrupt $\mathsf{GM}$ and all other users. This ensures that {signatures signed by a traceable user} (traceable signatures) are non-repudiated. Details are modelled in the experiment in Figure~\ref{Figure:ATS-NF}.
\begin{figure}
\begin{tabular}{cc}
\begin{minipage}{5cm}
\begin{small}

    \smallskip
    \underline{$\mb{Exp}_{\ms{ATS},\mc{A}}^{\ms{NF}}(\lambda)$}

     \smallskip
        $ \ms{pp}\leftarrow\ms{Setup}(\lambda)$.\\
        $(\ms{gpk},\ms{st})\leftarrow\mc{A}(\ms{pp}).$\\
        $\text{If}~\ms{gpk.pp}\neq\ms{pp},~\text{return}~\bot.$\\
        $(M,\Sigma,\ms{upk},\Pi)\leftarrow \mc{A}^{\ms{UKG},\ms{Sign}}(\ms{st})$.\\
         Return $1$ if $((M,\Sigma)\notin Q~\wedge~$  \\ $~~~~\ms{Verify}(\ms{gpk},M,\Sigma)=1~\wedge~$\\
          $~~~~\ms{upk}\in \mathsf{dom}(S)~\wedge~$  \\ $~~~~\ms{Judge}(\ms{gpk},M,\Sigma,\ms{upk},\Pi)=1)$.
\end{small}
\end{minipage}
&
\begin{minipage}{7.35cm}
\begin{small}

 $\underline{\mathbf{Oracle}~\ms{UKG}(\ms{pp})}$

     \smallskip
        $(\ms{upk},\ms{usk})\leftarrow\ms{UKeyGen}(\ms{pp})$,\\
        $ S[\ms{upk}]=\ms{usk}$.\\
        $\text{Return}~\ms{upk}$.\\
    $\underline{\mb{Oracle}~\ms{Sign}(\ms{cert},M)}$

     \smallskip
        $\ms{usk}=S[\ms{cert.upk}]$.\\
        $\text{If}~(\ms{usk}=\bot)~\text{return}~\bot.$\\
        $\Sigma\leftarrow\ms{Sign}(\ms{gpk},\ms{cert},\ms{usk},M)$.\\
        $Q=Q\cup\{(M,\Sigma)\}$. Return $\Sigma$.
\end{small}
\end{minipage}
\end{tabular}
\caption{Experiment to define non-frameability. }\label{Figure:ATS-NF}
\end{figure}

Define the advantage $\mathbf{Adv}_{\mathsf{ATS},\mathcal{A}}^{\mathsf{NF}}(\lambda)$ of adversary $\mathcal{A}$ against non-frameability of the $\mathsf{ATS}$ scheme as $\text{Pr}[\mathbf{Exp}_{\mathsf{ATS},\mathcal{A}}^{\mathsf{NF}}(\lambda)=1]$. An $\mathsf{ATS}$ scheme is non-frameable if the advantage of any $\mathrm{PPT}$ adversary $\mathcal{A}$ is negligible.

\noindent{\sc{Anonymity with accountability.}} $\mathsf{AwA}$ requires that a user is anonymous even from a corrupted group manager that has full control over the system as long as this user is non-traceable. In other words, the certificate is generated according to $\mathsf{tr}=0$. 
Details are modelled in the experiment in Figure~\ref{Figure:ATS-AwA}.

\begin{figure}
\begin{tabular}{cc}
\begin{minipage}{4.65cm}
\begin{small}

\smallskip
    \underline{$\mb{Exp}_{\ms{ATS},\mc{A}}^{\ms{AwA}-b}(\lambda)$}

        \smallskip
        $ \ms{pp}\leftarrow\ms{Setup}(\lambda)$.\\
        $(\ms{gpk},\ms{st})\leftarrow\mc{A}(\ms{pp}).$\\
        $\text{If}~\ms{gpk.pp}\neq\ms{pp},~\text{return}~\bot.$\\
        $b'\leftarrow \mc{A}^{\ms{Ch}}(\ms{st})$\\
        $ \text{Return}~b'$.
\end{small}
\end{minipage}
&
\begin{minipage}{7.35cm}
\begin{small}

\smallskip
    $\underline{\mathbf{Oracle}~\ms{Ch}(\ms{cert}_0,\ms{cert}_1,\ms{usk}_0,\ms{usk}_1,M,w_{0}^{\ms{escrw}},w_{1}^{\ms{escrw}},0)}$

\smallskip
        $\Sigma_0\leftarrow\ms{Sign}(\ms{gpk},\ms{cert}_0,\ms{usk}_0,M)$.\\
        $\Sigma_1\leftarrow\ms{Sign}(\ms{gpk},\ms{cert}_1,\ms{usk}_1,M)$.\\
        $\text{If}~(\Sigma_0\neq\bot\wedge\Sigma_1\neq\bot~\wedge$\\
        $~~~~\ms{Account}(\ms{gpk},\ms{cert}_0,w_0^{\ms{escrw}},0)~\wedge$\\
        $~~~~\ms{Account}(\ms{gpk},\ms{cert}_1,w_1^{\ms{escrw}},0)), $\\
        $~~~~\text{return}~\Sigma_b.$\\
        $\text{Else~return}~\bot$.
\end{small}
\end{minipage}
\end{tabular}
\caption{Experiment to define anonymity with accountability. }\label{Figure:ATS-AwA}
\end{figure}

Define the advantage $\mathbf{Adv}_{\mathsf{ATS},\mathcal{A}}^{\mathsf{AwA}}(\lambda)$ of $\mathcal{A}$ against anonymity with accountability of the $\mathsf{ATS}$ scheme as $\vert\text{Pr}[\mathbf{Exp}_{\mathsf{ATS},\mathcal{A}}^{\mathsf{AwA}-1}(\lambda)=1]-\text{Pr}[\mathbf{Exp}_{\mathsf{ATS},\mathcal{A}}^{\mathsf{AwA}-0}(\lambda)=1]\vert$. An $\mathsf{ATS}$ scheme is anonymous with accountability if the advantage of any $\mathrm{PPT}$ adversary $\mathcal{A}$ is negligible. \smallskip

\noindent{\sc{Trace-obliviousness.}} Trace-obliviousness requires that each user cannot determine whether they are being traced or not. Details are modelled in the experiment in Figure~\ref{Figure:ATS-TO}.

\begin{figure}
\begin{tabular}{cc}
\begin{minipage}{4.65cm}
\begin{small}

\smallskip
    \underline{$\mb{Exp}_{\ms{ATS},\mc{A}}^{\ms{TO}-b}(\lambda)$}

     \smallskip
        $\ms{pp}\leftarrow\ms{Setup}(\lambda)$.\\
        $(\ms{gpk},\ms{ik},\mathsf{ok})\leftarrow\ms{GKeyGen}(\ms{pp}).$\\
        $b'\leftarrow \mc{A}^{\ms{Ch},\ms{Enroll},\mathsf{Open}}(\ms{gpk})$\\
        $ \text{Return}~b'$.

 \end{small}
\end{minipage}

&
\begin{minipage}{7.35cm}
\begin{small}

    $\underline{\mb{Oracle}~\ms{Enroll}(\ms{upk},\ms{tr})}$

    \smallskip
        $(\ms{cert},w^{\ms{escrw}})\leftarrow\ms{Enroll}(\ms{ik},\ms{upk},\mathsf{tr}).$\\
        $\text{Return}~\ms{cert}$.\\
          $\underline{\mathbf{Oracle}~\ms{Ch}(\ms{upk})}$

          \smallskip
        $(\ms{cert},w^{\ms{escrw}})\leftarrow\ms{Enroll}(\ms{ik},\ms{upk}, b).$\\
        $U=U\cup\{\ms{upk}\},~\text{Return}~\ms{cert}$.

 $\underline{\mathbf{Oracle}~\mathsf{Open}(M,\Sigma)}$

    \smallskip
    $(\mathsf{upk},\Pi)\leftarrow \mathsf{Open}(\mathsf{ok},M,\Sigma)$\\
    If $\mathsf{upk}\in U$, then return $\bot$;
    Else return $(\mathsf{upk},\Pi)$.

\end{small}
\end{minipage}
\end{tabular}
\caption{Experiment to define trace-obliviousness.}\label{Figure:ATS-TO}
\end{figure}

Define the advantage $\mathbf{Adv}_{\mathsf{ATS},\mathcal{A}}^{\mathsf{TO}}(\lambda)$ of adversary $\mathcal{A}$ against trace-obliviousness of the $\mathsf{ATS}$ scheme as $\vert\text{Pr}[\mathbf{Exp}_{\mathsf{ATS},\mathcal{A}}^{\mathsf{TO}-1}(\lambda)=1]-\text{Pr}[\mathbf{Exp}_{\mathsf{ATS},\mathcal{A}}^{\mathsf{TO}-0}(\lambda)=1]\vert$. An $\mathsf{ATS}$ scheme is trace-oblivious if the advantage of any $\mathrm{PPT}$ adversary $\mathcal{A}$ is negligible.

\section{Key-Oblivious Encryption from Lattices}\label{section:ATS-KOE}
In~\cite{KM15}, Kohlweiss and Miers constructed a $\mathsf{KOE}$ scheme based on ElGamal cryptosystem~\cite{ElGamal84}. To adapt their blueprint into the lattice setting, we would need a key-private homomorphic encryption scheme whose public keys and ciphertexts should have the same algebraic form (e.g., each of them is a pair of ring elements). We observe that, the LPR \textsf{RLWE}-based encryption scheme, under appropriate setting of parameters, does satisfy these conditions. We thus obtain an instantiation of \textsf{KOE} which will then serve as a building block for our \textsf{ATS} construction in Section~\ref{section:ATS-ATS}.

\subsection{Description of Our KOE Scheme}\label{subsection:KOE-construction}
Our \textsf{KOE} scheme works as follows.
\begin{description}
\item[$\mathsf{Setup}(\lambda)$:] Given the security parameter $\lambda$, let $n=\mathcal{O}(\lambda)$ be a power of $2$ and $q=\widetilde{\mathcal{O}}(n^4)$. Also let $\ell=\lfloor\log \frac{q-1}{2}\rfloor+1$. Define the rings $R=\mathbb{Z}[X]/(X^n+1)$ and $R_q=R/qR$. Let the integer bound $B$ be of order $\widetilde{\mathcal{O}}(\sqrt{n})$ and $\chi$ be a $B$-bounded distribution over the ring $R$. This algorithm then outputs public parameter $\mathsf{pp}=\{n, q, \ell, R, R_q, B, \chi\}$. \smallskip

\item[$\mathsf{KeyGen}(\mathsf{pp})$:] Given the input $\mathsf{pp}$, this algorithm samples $s\hookleftarrow \chi$, $\mathbf{e}\hookleftarrow\chi^{\ell}$ and $\mathbf{a}\xleftarrow{\$} R_q^{\ell}$. Set $\mathsf{pk}=(\mathbf{a},\mathbf{b})=(\mathbf{a},\mathbf{a}\cdot s+\mathbf{e})\in R_q^{\ell}\times R_q^{\ell}$ and $\mathsf{sk}=s$. It then returns $(\mathsf{pk},\mathsf{sk})$. \smallskip

\item[$\mathsf{KeyRand}(\mathsf{pk})$:] Given the public key $\mathsf{pk}=(\mathbf{a},\mathbf{b})$, it samples $g\hookleftarrow\chi$, $\mathbf{e}_1\hookleftarrow\chi^{\ell}$ and $\mathbf{e}_2\hookleftarrow \chi^{\ell}$. Compute $$(\mathbf{a}',\mathbf{b}')=(\mathbf{a}\cdot g+\mathbf{e}_1,\hspace*{6.8pt}\mathbf{b}\cdot g+\mathbf{e}_2)\in R_q^{\ell}\times R_q^{\ell}.$$ This algorithm then  outputs randomized public key as $\mathsf{pk}'=(\mathbf{a}',\mathbf{b}')$. \smallskip

\item[$\mathsf{Enc}(\mathsf{pk}',p)$:]Given the public key $\mathsf{pk}'=(\mathbf{a}',\mathbf{b}')$ and a message $p\in R_q$, it samples $g'\in\chi$, $\mathbf{e}'_1\in\chi^{\ell}$ and $\mathbf{e}'_2\in \chi^{\ell}$. Compute $$(\mathbf{c}_1,\mathbf{c}_2)=(\mathbf{a}'\cdot g'+\mathbf{e}'_1,\hspace*{6.8pt}\mathbf{b}'\cdot g'+\mathbf{e}'_2+\lfloor q/4\rfloor\cdot\rdec(p))\in R_q^{\ell}\times R_q^{\ell}.$$
    This algorithm returns ciphertext as $\mathsf{ct}=(\mathbf{c}_1,\mathbf{c}_2)$.\smallskip

\item[$\mathsf{Dec}(\mathsf{sk},\mathsf{ct})$:] Given $\mathsf{sk}=s$ and $\mathsf{ct}=(\mathbf{c}_1,\mathbf{c}_2)$, the algorithm proceeds as follows.
\begin{enumerate}
        \item It  computes \[\mathbf{p''}=\frac{\mathbf{c}_{2}-\mathbf{c}_{1}\cdot s}{\lfloor q/4\rfloor}.\]
        \item For each coefficient of $\mathbf{p''}$,

             \begin{itemize}
                    \item if it is closer to $0$ than to $-1$ and $1$, then round it to $0$; \smallskip

                    \item   if it is closer to $-1$ than to $0$ and $1$, then round it to $-1$;\smallskip

                    \item if it is closer to $1$ than to $0$ and $-1$, then round it to $1$.
             \end{itemize}

        \item Denote the rounded $\mathbf{p''}$ as $\mathbf{p'}\in R_q^{\ell}$ with coefficients in $\{-1,0,1\}$.
        \item Let $p'\in R_q$ such that $\tau(p')=\mathbf{H}\cdot \tau(\mathbf{p'})$. Here,  $\mathbf{H}\in\mathbb{Z}_q^{n\times n\ell}$ is the decomposition matrix for elements of $R_q$ (see Appendix~\ref{subsection:decomposition}).
        \end{enumerate}
\end{description}

\subsection{Analysis of Our KOE Scheme}

\noindent{\sc{Correctness.}}
Note that \begin{align*}
\mathbf{c}_{2}-\mathbf{c}_{1}\cdot s&=\mathbf{b}'\cdot g'+\mathbf{e}'_2+\lfloor q/4\rfloor\cdot \rdec(p) -(\mathbf{a}'\cdot g'+\mathbf{e}'_1)\cdot s \\
&=\mathbf{e}\cdot g\cdot g' + \mathbf{e}_2\cdot g'-\mathbf{e}_1\cdot s\cdot g'+\mathbf{e}'_{2}- \mathbf{e}'_{1} \cdot s+\lfloor q/4\rfloor\cdot \rdec(p)
\end{align*}
where $s, g, g', \mathbf{e},\mathbf{e}_1,\mathbf{e}_2,\mathbf{e}'_1,\mathbf{e}'_2$ are $B$-bounded.  Hence we have:
\[\|\mathbf{e}\cdot g\cdot g' + \mathbf{e}_2\cdot g'-\mathbf{e}_1\cdot s\cdot g'+\mathbf{e}'_{2}- \mathbf{e}'_{1} \cdot s\|_{\infty}\leq 3n^2\cdot B^3=\widetilde{\mathcal{O}}(n^{3.5})\leq \big\lceil \frac{q}{10}\big\rceil=\widetilde{\mathcal{O}}(n^4).\]
With overwhelming probability, the rounding procedure described in the $\mathsf{Dec}$ algorithm  recovers $\rdec(p)$ and hence outputs $p$. Therefore, our $\mathsf{KOE}$ scheme is correct.

\smallskip
\noindent
{\sc Security.} The security of our \textsf{KOE} scheme is stated in the following theorem.

\begin{theorem}\label{theorem:KOE-sec}
Under the $\mathsf{RLWE}$ assumption, the described key-oblivious encryption scheme satisfies: (i) key randomizability;  (ii) plaintext indistinguishability under key randomization; and (iii) key privacy under key randomization.
\end{theorem}

The proof of Theorem~\ref{theorem:KOE-sec} is established by Lemma~\ref{lemma:KR}-\ref{lemma:KPr}.

\begin{lemma}\label{lemma:KR}
The key-oblivious encryption scheme described in Section~\ref{subsection:KOE-construction} is key randomizable defined in Section~\ref{subsection:KOE-security-model} under $\mathsf{RLWE}$ assumption.

\end{lemma}

\begin{proof}
Notice that the samples chosen according to $\mathcal{A}_{s,\chi}$ for some $s\hookleftarrow \chi$ are indistinguishable from random under the $\mathsf{RLWE}$ assumption. Therefore, the honestly generated public key $\mathsf{pk}=(\mathbf{a},\mathbf{b})\in R_q^{\ell}\times R_q^{\ell}$ is indistinguishable from truly random pair $\widetilde{\mathsf{pk}}=(\widetilde{\mathbf{a}},\widetilde{\mathbf{b}})\in R_q^{\ell}\times R_q^{\ell}$.  Hence, we may replace $\mathsf{pk}$ with $\widetilde{\mathsf{pk}}$ and this modification is negligible to the adversary.

Let $\mathsf{pk}_0=(\widetilde{\mathbf{a}}\cdot g+\mathbf{e}_1,\widetilde{\mathbf{b}}\cdot g+\mathbf{e}_2)$ and $\mathsf{pk}_1=(\mathbf{a}',\mathbf{a}'\cdot s'+\mathbf{e}')$, where $\mathsf{pk}_1$ is independent of $\widetilde{\mathsf{pk}}$. When $b=0$,  adversary is given $(\widetilde{\mathbf{a}},\widetilde{\mathbf{b}},\widetilde{\mathbf{a}}\cdot g+\mathbf{e}_1,\widetilde{\mathbf{b}}\cdot g+\mathbf{e}_2)$, which are $2\ell$  samples chosen according to $\mathcal{A}_{g,\chi}$. Therefore, $(\widetilde{\mathsf{pk}},\mathsf{pk}_0)$ is indistinguishable from $2\ell$ samples chosen according to $U(R_q\times R_q)$. When $b=1$,  adversary is given $(\widetilde{\mathbf{a}},\widetilde{\mathbf{b}},\mathbf{a'},\mathbf{a'}\cdot s'+\mathbf{e'})$. Since $\mathsf{pk}_1$ is independent of $\widetilde{\mathsf{pk}}$, so we can replace $\mathsf{pk}_1$ with a truly random pair. Hence, $(\widetilde{\mathsf{pk}},\mathsf{pk}_1)$ is also indistinguishable from $2\ell$ samples chosen according to $U(R_q\times R_q)$. Therefore, the adversary cannot distinguish the case $b=0$ from the case $b=1$.

It then follows that the advantage of any $\mathrm{PPT}$ adversary in the experiment $\mathbf{Exp}_{\mathsf{KOE},\mathcal{A}}^{\mathsf{KR}}(\lambda)$ is negligible and hence our $\mathsf{KOE}$ scheme is key randomizable.
\end{proof}

\begin{lemma}\label{lemma:INDr}

The key-oblivious encryption scheme described in Section~\ref{subsection:KOE-construction} is plaintext indistinguishable under key randomization defined in Section~\ref{subsection:KOE-security-model} under $\mathsf{RLWE}$ assumption.

\end{lemma}
\begin{proof}
Let $\mathcal{A}$ be any PPT adversary attacking the plaintext indistinguishability under key randomization with advantage $\epsilon$, we will show $\epsilon=\mathrm{negl}(\lambda)$ assuming the hardness of the $\mathsf{RLWE}$ problem.  Specifically, we construct a sequence of indistinguishable games $G_0,G_1, G_2, G_3,G_4$, such that, $\mathbf{Adv}_{\mathcal{A}}(G_0)=\epsilon$ and $\mathbf{Adv}_{\mathcal{A}}(G_4)=0$.
\begin{description}
\item[Game~$G_0$:] This is the real experiment $\mathbf{Exp}_{\mathsf{KOE},\mathcal{A}}^{\mathsf{INDr}}(\lambda)$. The challenger generates a public key $\mathsf{pk}=(\mathbf{a},\mathbf{b})=(\mathbf{a},\mathbf{a}\cdot s+\mathbf{e})$ honestly, sends it to the adversary $\mathcal{A}$, receives back a randomized key pair $\mathsf{pk}'=(\mathbf{a}\cdot g+\mathbf{e}_1,\mathbf{b}\cdot g+\mathbf{e}_2)$, the randomness used to generate $\mathsf{pk}'$, and two messages $p_0,p_1\in R_q$. The challenger first checks whether  $\mathsf{pk}'$ is generated from the randomness or not. If not, the challenger returns $\bot$. Otherwise, he samples $b\xleftarrow{\$}\{0,1\}$ and encrypts the message $p_b$ to  ciphertext $(\mathbf{c}_1,\mathbf{c}_2)=(\mathbf{a}'\cdot g'+\mathbf{e}'_1,\hspace*{6.8pt}\mathbf{b}'\cdot g'+\mathbf{e}'_2+\lfloor q/4\rfloor\cdot\rdec(p_b))$ and sends $(\mathbf{c}_1,\mathbf{c}_2)$ to the adversary $\mathcal{A}$, who then outputs $b'\in\{0,1\}$.  This game outputs $1$ if $b'=b$ or~$0$ otherwise. By assumption, $\mathcal{A}$ has advantage $\epsilon$ in this game.
    \smallskip

\item[Game~$G_1$:] In this game,we make a slight modification to the Game $G_0$: the public key $\mathsf{pk}$ is replaced with a truly random pair $\widetilde{\mathsf{pk}}=(\widetilde{\mathbf{a}},\widetilde{\mathbf{b}})$. By the $\mathsf{RLWE}_{n,q,\ell,\chi}$ assumption, the adversary cannot distinguish  $\mathsf{pk}=(\mathbf{a},\mathbf{b})$ from uniform. It then follows that $G_0$ is indistinguishable from $G_1$. We additionally remark that $\mathsf{pk}'$ obtained from randomizing $\widetilde{\mathsf{pk}}$ is indistinguishable from random by the same assumption.
\smallskip
\item[Game~$G_2$:] In this game, we modify $G_1$ as follows: instead of generating $(\mathbf{c}_1,\mathbf{c}_2)$ faithfully using the randomized public key $\mathsf{pk}'$, we generate ciphertext $(\mathbf{c}_1,\mathbf{c}_2)$ as $(\widetilde{\mathbf{a}'}\cdot g'+\mathbf{e}'_1,\hspace*{6.8pt}\widetilde{\mathbf{b}'}\cdot g'+\mathbf{e}'_2+\lfloor q/4\rfloor\cdot\rdec(p_b))$, where $\widetilde{\mathsf{pk}'}=(\widetilde{\mathbf{a}'},\widetilde{\mathbf{b}'})$ is uniformly chosen over $R_q^{\ell}\times R_q^{\ell}$. Since $\mathsf{pk}'$ obtained from randomizing $\widetilde{\mathsf{pk}}$ is indistinguishable from random,  this modification is indistinguishable to adversary $\mathcal{A}$. 

\smallskip
\item[Game~$G_3$:]  In this game, we generate $(\mathbf{c}_1,\mathbf{c}_2)$ as $(\mathbf{z}_1,\mathbf{z}_2+\lfloor q/4\rfloor\cdot\rdec(p_b))$, where $(\mathbf{z}_1,\mathbf{z}_2)\in R_q^{\ell}\times R_q^{\ell}$ are uniformly random. The assumed hardness of the $\mathsf{RLWE}_{n,q,\ell,\chi}$ problem implies that $G_2$ and $G_3$ are computationally indistinguishable.
\smallskip
\item[Game~$G_4$:] In the game, we make a conceptual modification to $G_3$. Namely, we sample uniformly random $\mathbf{z}'_1\in R_q^{\ell}$ and $\mathbf{z}'_2\in R_q^{\ell}$ and let $(\mathbf{c}_1,\mathbf{c}_2)=(\mathbf{z}'_1,\mathbf{z}'_2)$. It is clear that $G_3$ and $G_4$ are statistically indistinguishable. Moreover, since $G_4$ is no longer dependent on the challenger's bit $b$, the advantage of $\mathcal{A}$ in this game is $0$.

 \end{description}
It follows from the above construction that the advantage $\epsilon$ of the adversary $\mathcal{A}$ is negligible. This concludes the proof.
\end{proof}
\begin{lemma}\label{lemma:KPr}

The key-oblivious encryption scheme described in Section~\ref{subsection:KOE-construction} is key private under key randomization defined in Section~\ref{subsection:KOE-security-model} under $\mathsf{RLWE}_{n,q,\chi}$ assumption.

\end{lemma}

\begin{proof}
The proof of Lemma~\ref{lemma:KPr} is similar to that of Lemma~\ref{lemma:INDr}, we briefly describe it here.  As in Lemma~\ref{lemma:INDr}, we construct a sequence of indistinguishable games $G_0,G_1, G_2, G_3$, such that, $\mathbf{Adv}_{\mathcal{A}}(G_0)=\mathbf{Adv}_{\mathsf{KOE},\mathcal{A}}^{\mathsf{KPr}}(\lambda)$ and $\mathbf{Adv}_{\mathcal{A}}(G_3)=0$.

Game $G_0$ is the experiment $\mathbf{Exp}_{\mathsf{KOE},\mathcal{A}}^{\mathsf{KPr}}(\lambda)$, Game $G_1$ modifies Game $G_0$ by replacing public key $\mathsf{pk}_0$  with truly random pair $\widetilde{\mathsf{pk}}_0$ while Game $G_2$ modifies Game $G_1$ by replacing public key $\mathsf{pk}_1$  with another independent and  random pair $\widetilde{\mathsf{pk}}_1$. By the hardness of the $\mathsf{RLWE}_{n,q,\ell,\chi}$ problem, these two modifications are indistinguishable to any PPT adversary. In Game $G_3$, we further modify Game $G_2$ by generating the ciphertext $(\mathbf{c}_1,\mathbf{c}_2)$ using $\widetilde{\mathsf{pk}'}$ chosen uniformly over $R_q^{\ell}\times R_q^{\ell}$ as in Lemma~\ref{lemma:INDr}. By the same argument, this change is negligible to any PPT adversary. Furthermore, since $G_3$ is no longer dependent on the challenger's bit $b$, the advantage of adversary in this game is $0$. This ends the brief description.

\end{proof}

\section{Handling Quadratically Hidden RLWE Relations}\label{section:Stern-quad-RLWE}
In Section~\ref{subsection:extended-permutation-technique}, we extend the refined permuting technique recalled in Section~\ref{subsection:permutation-technique-by-lnwx} to prove that a secret integer $y$ is multiplication of two secret integers $a\in\{-1,0,1\}$ and $g\in\{-1,0,1\}$. We then describe our zero-knowledge protocol for handling quadratic relations in the $\mathsf{RLWE}$ setting in Section~\ref{subsection:zk-for-rlwe}. Specifically, we demonstrate how to prove in zero-knowledge that a give vector $\mathbf{c}$ is a correct $\mathsf{RLWE}$ evaluation, i.e.,~$\mathbf{c}=\mathbf{a}\cdot g+\mathbf{e}$, where the hidden vectors $\mathbf{a},\mathbf{e}$ and element $g$ may satisfy additional conditions. The protocol is developed based on Libert et al.'s work~\cite{LLMNW16-ge} on quadratic relations in the general lattice setting.
\subsection{Our Extended Permuting Technique}\label{subsection:extended-permutation-technique}
{\bf Proving that $y = a \cdot g$. }
For any $a, g\in \{-1,0,1\}$, define vector $\mathsf{mult}_3(a,g) \in \{-1,0,1\}^9$ of the following form:
\begin{align*}
&\mathsf{mult}_3(a,g)=
\big(
[a+1]_3\cdot [g+1]_3, \hspace*{3.6pt}
[a]_3 \cdot [g+1]_3, \hspace*{3.6pt}
[a-1]_3 \cdot [g+1]_3, \hspace*{3.6pt}
[a+1]_3 \cdot [g]_3, \hspace*{3.6pt}\\
&[a]_3\cdot [g]_3, \hspace*{3.6pt}
[a-1]_3\cdot [g]_3\hspace*{3.6pt},
[a+1]_3 \cdot [g-1]_3, \hspace*{3.6pt}
[a]_3\cdot [g-1]_3, \hspace*{3.6pt}
[a-1]_3\cdot [g-1]_3
\big)^\top.
\end{align*}
Then for any $b, e \in \{-1,0,1\}$, we define the permutation $\phi_{b,e}(\cdot)$ that acts in the following way. It maps vector $\mathbf{v}$ of the following form
\begin{align*}
\mathbf{v}=
\big(v^{(-1, -1)}, v^{(0,-1)},v^{(1,-1)}, v^{(-1,0)}, v^{(0,0)},
v^{(1,0)},v^{(-1,1)},v^{(0,1)}, v^{(1,1)}\big)^\top \in \mathbb{Z}^9
\end{align*}
into vector $\phi_{b,e}(\mathbf{v})$ of the following form
\begin{align*}
\phi_{b,e}(\mathbf{v})=
\big( &
v^{([-b-1]_3, [-e-1]_3)},
v^{([-b]_3, [-e-1]_3)},
v^{([-b+1]_3, [-e-1]_3)}, \\
&v^{([-b-1]_3, [-e]_3)},
v^{([-b]_3, [-e]_3)},
v^{([-b+1]_3, [-e]_3)}, \\
&v^{([-b-1]_3, [-e+1]_3)},
v^{([-b]_3, [-e+1]_3)},
v^{([-b+1]_3, [-e+1]_3)}
\big)^\top.
\end{align*}
Then for any $a, b, g, e \in \{-1,0,1\}$, one is able to check that the following equivalence is satisfied.
\begin{eqnarray}\label{eq:equivalence-ext-3}
\mathbf{v} =  \mathsf{mult}_3(a,g)
 \Longleftrightarrow
\phi_{b,e}(\mathbf{v})  =  \mathsf{mult}_3( [a + b]_3, \hspace*{-1pt}[g + e]_3).
\end{eqnarray}
Note that the above equivalence in~(\ref{eq:equivalence-ext-3}) is essential to prove knowledge of such secret integer $y$ in the framework of Stern's protocol. We first extend $y$ to vector $\mathbf{v}=\mathsf{mult}_3(a,g)$, sample uniform $b\in\{0,1\}$ and $e\in\{-1,0,1\}$, and then demonstrate to the verifier $\phi_{b,e}(\mathbf{v})= \mathsf{mult}_3(\hspace*{1.6pt}[a+b]_3, \hspace*{1.6pt}[g + e]_3\hspace*{1.6pt})$. Due to the equivalence in~(\ref{eq:equivalence-ext-3}), the verifier should be convinced of the well-formedness of $y$ and no extra information is revealed to him. Furthermore, the technique is extendable so that we can use the same ``one time pads'' $b$ and  $e$ at the places where $a$ and $g$ appear, respectively.

\smallskip

Now we generalize the above technique to prove knowledge of vector of the following expansion form. We aim to obtain equivalence similar to~(\ref{eq:equivalence-ext-3}), which is useful in Stern's framework.
\smallskip

\noindent
{\bf Handling an expansion vector. } 
%
We now tackle an expansion vector $\mathbf{y}=\mathsf{expd}(\mathbf{a},\mathbf{g})$ of the form $\mathbf{y}=(\mathbf{y}_0\|\ldots\|\mathbf{y}_{n-1})\in\{-1,0,1\}^{n^2\ell\delta_B}$, where $\mathbf{y}_i$ is of the following form
\begin{align*}
\mathbf{y}_i=(a_{1}\cdot g_{i,1},\hspace*{2pt}\ldots,\hspace*{2pt}a_{1}\cdot g_{i,\delta_B},\hspace*{2pt}\ldots,\hspace*{2pt}
a_{n\ell}\cdot g_{i,1}\hspace*{2pt},\ldots,\hspace*{2pt}a_{n\ell}\cdot g_{i,\delta_B}),
\end{align*}
$\mathbf{g}\in\{-1,0,1\}^{n\delta_B}$ is of the form
\begin{align*}
\mathbf{g}=(g_{0,1},g_{0,2},\ldots,g_{0,\delta_B},\ldots,g_{n-1,1},g_{n-1,2},\ldots,g_{n-1,\delta_B})^\top,
\end{align*}  and $\mathbf{a}=(a_1,\ldots,a_{n\ell})^\top\in\{-1,0,1\}^{n\ell}$ for some positive integers $n,\ell,\delta_B$. \smallskip

Denote $\mathbf{y}=(a_i\cdot g_{j,k})_{i\in[n\ell],j\in[0,n-1],k\in[\delta_B]}$, we then
define an extension of the expansion vector $\mathbf{y}$ as $\mathsf{mult}(\mathbf{a},\mathbf{g})=(\mathsf{mult}_{3}(a_i,g_{j,k}))_{i\in[n\ell],j\in[0,n-1],k\in[\delta_B]}\in\{-1,0,1\}^{9n^2\ell\delta_B}$.

For  $\mathbf{e}=(e_{0,1},e_{0,2},\ldots,e_{0,\delta_B},\ldots,e_{n-1,1},e_{n-1,2},\ldots,e_{n-1,\delta_B})^\top\in\{-1,0,1\}^{n\delta_B}$ and $\mathbf{b}=(b_1,\ldots,b_{n\ell})^\top\in\{-1,0,1\}^{n\ell}$, we define the permutation $\Phi_{\mathbf{b},\mathbf{e}}(\cdot)$ that behaves as follows. It maps   vector $\mathbf{v}\in\mathbb{Z}^{9n^2\ell\delta_B}$ of the following form: \begin{align*}
\big(&\mathbf{v}_{1,0,1}\|\cdots\|\mathbf{v}_{1,0,\delta_B}\|\cdots\|\mathbf{v}_{n\ell,0,1}\|\cdots\|\mathbf{v}_{n\ell,0,\delta_B}\|\\
&\mathbf{v}_{1,1,1}\|\cdots\|\mathbf{v}_{1,1,\delta_B}\|\cdots\|\mathbf{v}_{n\ell,1,1}\|\cdots\|\mathbf{v}_{n\ell,1,\delta_B}\|\\
&\cdots\cdot\\
&\mathbf{v}_{1,n-1,1}\|\cdots\|\mathbf{v}_{1,n-1,\delta_B}\|\cdots\|\mathbf{v}_{n\ell,n-1,1}\|\cdots\|\mathbf{v}_{n\ell,n-1,\delta_B}\big)
\end{align*}
which consists of blocks of size $9$, to vector $\Phi_{\mathbf{b},\mathbf{e}}(\mathbf{v})$ of the following form:
\begin{align*}
\big(
&\phi_{b_1,e_{0,1}}(\mathbf{v}_{1,0,1})\|\cdots\|\phi_{b_1,e_{0,\delta_B}}(\mathbf{v}_{1,0,\delta_B})\|\cdots\|\\
&\phi_{b_{n\ell},e_{0,1}}(\mathbf{v}_{n\ell,0,1})\|\cdots\|\phi_{b_{n\ell},e_{0,\delta_B}}(\mathbf{v}_{n\ell,0,\delta_B})\|\\
&\phi_{b_1,e_{1,1}}(\mathbf{v}_{1,1,1})\|\cdots\|\phi_{b_1,e_{1,\delta_B}}(\mathbf{v}_{1,1,\delta_B})\|\cdots\|\\
&\phi_{b_{n\ell},e_{1,1}}(\mathbf{v}_{n\ell,1,1})\|\cdots\|\phi_{b_{n\ell},e_{1,\delta_B}}(\mathbf{v}_{n\ell,1,\delta_B})\|\\
&\cdots\cdot\\
&\phi_{b_1,e_{n-1,1}}(\mathbf{v}_{1,n-1,1})\|\cdots\|\phi_{b_1,e_{n-1,\delta_B}}(\mathbf{v}_{1,n-1,\delta_B})\|\cdots\|\\
&\phi_{b_{n\ell},e_{n-1,1}}(\mathbf{v}_{n\ell,n-1,1})\|\cdots\|\phi_{b_{n\ell},e_{n-1,\delta_B}}(\mathbf{v}_{n\ell,n-1,\delta_B})\big)
\end{align*}
For any $\mathbf{a},\mathbf{b}\in\{-1,0,1\}^{n\ell}$ and any $\mathbf{g},\mathbf{e}\in\{-1,0,1\}^{n\delta_B}$, it then follows from~(\ref{eq:equivalence-ext-3}) that the following equivalence holds.
\begin{eqnarray}\label{eq:equivalence-ext-vector}
\mathbf{v} = \mathsf{mult}(\mathbf{a},\mathbf{g}) \hspace*{6.8pt}\Longleftrightarrow \hspace*{6.8pt} \Phi_{\mathbf{b},\mathbf{e}}(\mathbf{v}) = \mathsf{mult}([\mathbf{a} + \mathbf{b}]_3, [\mathbf{g} + \mathbf{e}]_3).
\end{eqnarray}

\subsection{Proving the RLWE Relation with Hidden Vector}\label{subsection:zk-for-rlwe}
We are going to describe our statistical $\mathsf{ZKAoK}$ protocol for the $\mathsf{RLWE}$ relation with hidden vector.
Let $q, \ell, B$ be some integers and   $R, R_q$ be two rings, which are specified as in Section~\ref{subsection:KOE-construction}. Our goal is to design a $\mathsf{ZK}$ argument system that allows a prover $\mathcal{P}$ to convince a verifier $\mathcal{V}$ on input $\mathbf{c}\in R_q^{\ell}$ that $\mathcal{P}$ knows secrets $\mathbf{a}\in R_q^{\ell}$,  $g\in R_q$   and $\mathbf{e}\in R_q^{\ell}$   such that $g$ and $\mathbf{e}$ are $B$-bounded and
\begin{eqnarray}\label{equation:rlwe-initial-1}
\mathbf{c}=\mathbf{a}\cdot g+\mathbf{e}.
 \end{eqnarray}
 Furthermore, this protocol should be extendable such that we are able to prove that the secrets $\mathbf{a},g,\mathbf{e}$ satisfy other relations.

As in Section~\ref{subsection:zk-for-DM}, we aim to obtain an instance of the abstract protocol from  Section~\ref{subsection:Stern}. 
\smallskip

\noindent
{\sc Decomposing-Unifying.}
To start with, we also employ the notations $\mathsf{rot}$ and~$\tau$ from Section~\ref{subsection:rings}  and the decomposition techniques from Section~\ref{subsection:decomposition} to transform equation~(\ref{equation:rlwe-initial-1}) into
$\mathbf{M}_0\cdot \mathbf{w}_0 = \mathbf{u} \bmod q$, where $\mathbf{M}_0, \mathbf{u}$ are built from public input, and vector $\mathbf{w}_0$ is built from secret input  and  coefficients of which are in the set $\{-1,0,1\}$.
\smallskip

\noindent Let $\mathbf{a}=(a_1,a_2,\cdots,a_\ell)^\top$, $\tau(g)=(g_0,\cdots,g_{n-1})^\top$, $\mathbf{a}_i^{\star}=\tau(\rdec({a}_i))\in\{-1,0,1\}^{n\ell}$  $\forall~ i\in[\ell]$, $\mathbf{g}^{\star}=\tau(\rdec_B(g))\in\{-1,0,1\}^{n\delta_B}$. Let  $\mathbf{a}_i^{\star}=(a_{i,1},a_{i,2},\cdots, a_{i,n\ell})^\top$  $\forall~i\in[\ell]$, $\mathbf{g}^{\star}=(g_{0,1},\cdots g_{0,\delta_B},\cdots,g_{n-1,1},\cdots,g_{n-1,\delta_B})^\top$. We then have the following:
\begin{align*}
\tau(a_i\cdot g)&=\mathsf{rot}(a_i)\cdot \tau(g)
                 =[\tau(a_i)|\tau(a_i\cdot X)|\ldots|\tau(a_i\cdot X^{n-1})]\cdot \tau(g)   \\
                &= \sum_{j=0}^{n-1} \tau(a_i\cdot X^j)\cdot g_j
                 =\sum_{j=0}^{n-1} \mathsf{rot}(X^j)\cdot\tau(a_i)\cdot g_j
                =\sum_{j=0}^{n-1} \mathsf{rot}(X^j)\cdot \mathbf{H}\cdot \mathbf{a}_i^{\star}\cdot g_j \\
                &=\sum_{j=0}^{n-1} \mathsf{rot}(X^j)\cdot \mathbf{H}\cdot (a_{i,1}\cdot g_j,\ldots,a_{i,n\ell}\cdot g_j)^\top\bmod q\\
\end{align*}
\vspace*{-24pt}

Observe that, for each $k\in[n\ell]$, we have
\begin{align*}
a_{i,k}\cdot g_j&=a_{i,k}\cdot (B_1,\ldots,B_{\delta_B})\cdot (g_{j,1},\dots,g_{j,\delta_B})^\top \\
&=(B_1,\ldots,B_{\delta_B}) \cdot (a_{i,k}\cdot g_{j,1},\dots,a_{i,k}\cdot g_{j,\delta_B} )^\top
\end{align*}
Denote $\mathbf{y}_{i,j}\in\{-1,0,1\}^{n\ell\delta_B}$ of the following form: \[\mathbf{y}_{i,j}=(a_{i,1}\cdot g_{j,1},\dots,a_{i,1}\cdot g_{j,\delta_B},\ldots, a_{i,n\ell}\cdot g_{j,1},\dots,a_{i,n\ell}\cdot g_{j,\delta_B} )^\top,\]
we then obtain
$$(a_{i,1}\cdot g_j,\ldots,a_{i,n\ell}\cdot g_j)^\top = \mathbf{H}_{\ell,B}\cdot \mathbf{y}_{i,j}\mod q. $$
Define $\mathbf{Q}_0\in\mathbb{Z}_q^{n\times n^2\ell\delta_B}$ of the following form: \[\mathbf{Q}_0=[\mathsf{rot}(X^0)\cdot\mathbf{H}\cdot\mathbf{H}_{\ell,B}|\cdots|\mathsf{rot}(X^{n-1})\cdot\mathbf{H}\cdot\mathbf{H}_{\ell,B}].\] Let $\mathbf{y}_i=(\mathbf{y}_{i,0}\|\cdots\|\mathbf{y}_{i,n-1})=\mathsf{expd}(\mathbf{a}_i^{\star},\mathbf{g}^\star)\in\{-1,0,1\}^{n^2\ell\delta_B}$,
we then obtain:
$$\tau(a_i\cdot g) =\mathbf{Q}_0\cdot \mathbf{y}_i \mod q.$$
Let $\mathbf{e}^{\star}=\tau(\rdec_B(\mathbf{e}))\in\{-1,0,1\}^{n\ell\delta_B}$, $\mathbf{Q}=\left(
                                                                                       \begin{array}{cccc}
                                                                                         \mathbf{Q}_0 & ~ & ~ & ~ \\
                                                                                         ~ &  \mathbf{Q}_0 &~ & ~ \\
                                                                                         ~ & ~ & \ddots & ~ \\
                                                                                         ~ & ~ & ~ &  \mathbf{Q}_0  \\
                                                                                       \end{array}
                                                                                     \right)\in\mathbb{Z}_q^{n\ell\times n^2\ell^2\delta_B}.$
Now equation~(\ref{equation:rlwe-initial-1}) is equivalent to
\begin{align*}
\tau(\mathbf{c})&=(\tau(a_1\cdot g),\ldots,\tau(a_\ell\cdot g))^\top+\tau(\mathbf{e})\\
& =\mathbf{Q}\cdot(\mathbf{y}_1\|\cdots\|\mathbf{y}_{\ell})+\mathbf{H}_{\ell,B}\cdot \mathbf{e}^{\star} \bmod q
\end{align*}

Rearrange the above equivalent form using some basic algebra, we are able to obtain an unifying equation of the following form:
\[\mathbf{M}_0\cdot \mathbf{w}_0=\mathbf{u} \bmod q, \]
where  $\mathbf{M}_0$ is built from the public matrices $\mathbf{Q}$ and $\mathbf{H}_{\ell,B}$, 
$\mathbf{u}$ is the vector $\tau(\mathbf{c})$, while $\mathbf{w}_0=(\mathbf{y}_1\|\cdots\|\mathbf{y}_{\ell}\|\mathbf{e}^{\star})\in\{-1,0,1\}^{n^2\ell^2\delta_{B}+n\ell\delta_B}$.
\medskip

\noindent
{\sc Extending-Permuting. } In this second step, we aim to  transform the secret $\mathbf{w}_0$ to a vector $\mathbf{w}$ such that it satisfies the requirements specified by the abstract protocol from section~\ref{subsection:Stern}. In the process, the techniques introduced in Section~\ref{subsection:permutation-technique-by-lnwx} and~\ref{subsection:extended-permutation-technique} are utilized. \smallskip

We first extend $\mathbf{w}_0 = (\mathbf{y}_1 \| \cdots\|\mathbf{y}_\ell\| \mathbf{e}^{\star})$ as follows.
\begin{align*}
\mathbf{y}_i &\mapsto \mathbf{y}'_i = \mathsf{mult}\big(\mathbf{a}_i^{\star}, \mathbf{g}^\star\big) \in \{-1,0,1\}^{9n^2\ell\delta_B},~i\in[\ell]; \\
\mathbf{e}^{\star} &\mapsto \mathbf{e}'^{\star} = \mathsf{enc}(\mathbf{e}^{\star}) \in \{-1,0,1\}^{L_2}.
\end{align*}
Notice that for each $i\in[\ell]$, we have $\mathbf{y}_i=\mathsf{expd}(\mathbf{a}_i^\star,\mathbf{g}^\star)$. We then form vector $\mathbf{w}=(\mathbf{y}'_1\|\cdots\|\mathbf{y}'_{\ell}\|\mathbf{e}'^{\star})\in\{-1,0,1\}^{L}$,
where $$ L = L_1 + L_2; \hspace*{6.8pt}
L_1 =9n^2\ell^2\delta_B;  \hspace*{6.8pt}
L_2 = 3n\ell\delta_B.$$
According to the extension, we insert appropriate zero-columns to matrix $\mathbf{M}_0$, obtaining a new matrix $\mathbf{M} \in \mathbb{Z}_q^{n\ell \times L}$ such that the equation $\mathbf{M} \cdot \mathbf{w} = \mathbf{M}_0 \cdot \mathbf{w}_0$ holds.
\smallskip

We now define the set $\mathsf{VALID}$ that includes our secret vector $\mathbf{w}$, the set $\mathcal{S}$, and the associated  permutations $\{\Gamma_\eta: \eta \in \mathcal{S}\}$, such that the conditions in~(\ref{eq:zk-equivalence}) are satisfied.

\noindent
Let $\mathsf{VALID}$ be the set of all vectors $\mathbf{v}' = (\mathbf{v}'_1 \| \cdots\|\mathbf{v}'_\ell\|\mathbf{v}'_{\ell+1}) \in \{-1,0,1\}^{L}$ such that the following conditions hold:
\begin{itemize}
\item There exist $\mathbf{a}_i^{\star} \in \{-1,0,1\}^{n\ell}$ for each $i\in[\ell]$ and $\mathbf{g}^\star \in \{-1,0,1\}^{n\delta_B}$ such that $\mathbf{v}'_i = \mathsf{mult}(\mathbf{a}_i^{\star}, \mathbf{g}^\star)$. 
\item There exists $\mathbf{e}^{\star} \in \{-1,0,1\}^{n\ell\delta_B}$ such that $\mathbf{v}'_{\ell+1} = \mathsf{enc}(\mathbf{e}^{\star})$.
\end{itemize}
It is easy to see that the obtained vector $\mathbf{w}$ belongs to the set $\mathsf{VALID}$. \smallskip

Now let $\mathcal{S}=(\{-1,0,1\}^{n\ell})^{\ell}\times \{-1,0,1\}^{n\delta_B}\times\{-1,0,1\}^{n\ell\delta_B}$, and associate every element $\eta=(\mathbf{b}_1,\ldots,\mathbf{b}_{\ell},\mathbf{f}_1,\mathbf{f}_2)\in\mathcal{S}$ with permutation $\Gamma_{\eta}$ that behaves as follows. For a vector of the form $\mathbf{v}=(\mathbf{v}_1\|\cdots\|\mathbf{v}_{\ell}\|\mathbf{v}_{\ell+1})\in \mathbb{Z}^{L}$, where $\mathbf{v}_i\in \mathbb{Z}^{9n^2\ell\delta_B}$ for each $i\in[\ell]$ and $\mathbf{v}_{\ell+1}\in \mathbb{Z}^{L_2}$, it transforms $\mathbf{v}$ into vector
\[
\Gamma_\eta(\mathbf{v}) = \big(\hspace*{2.6pt}\Phi_{\mathbf{b}_1, \mathbf{f}_1}(\mathbf{v}_1) \hspace*{2.6pt} \|\cdots\|\hspace*{2.6pt}\Phi_{\mathbf{b}_{\ell}, \mathbf{f}_1}(\mathbf{v}_{\ell}) \hspace*{2.6pt}\|\hspace*{2.6pt} \Pi_{\mathbf{f}_2}(\mathbf{v}_{\ell+1})\hspace*{2.6pt}\big).
\]
It then follows from the equivalences in~(\ref{eq:equivalence-enc-vector}) and~(\ref{eq:equivalence-ext-vector}) that $\mathsf{VALID}$, $\mathcal{S}$, and $\Gamma_\eta$ fulfill the requirements specified in~(\ref{eq:zk-equivalence}). Therefore, we have transformed the considered statement to a case of the abstract protocol from Section~\ref{subsection:Stern}. To obtain the desired statistical $\mathsf{ZKAoK}$ protocol, it suffices for the prover and verifier to run the interactive protocol described in Figure~\ref{Figure:Interactive-Protocol}. The protocol has perfect completeness, soundness error $2/3$ and communication cost $\mathcal{O}(L \cdot \log q)$, which is of order ${\mathcal{O}}(n^2 \cdot \log^4 n) = \widetilde{\mathcal{O}}(\lambda^2)$.

\section{Accountable Tracing Signatures from Lattices}\label{section:ATS-ATS}
In this section, we construct our $\mathsf{ATS}$ scheme based on: (i) The Ducas-Micciancio signature scheme (as recalled in Section~\ref{subsection:DM-signatures}); (ii) The \textsf{KOE} scheme described in Section~\ref{section:ATS-KOE}; and
(iii) Stern-like zero-knowledge argument system that underlies our $\mathsf{ATS}$ construction, which is obtained by smoothly combining previous techniques as recalled in Section~\ref{subsection:zk-for-DM} and ours as described in Section~\ref{subsection:zk-for-rlwe}.
\subsection{The Zero-Knowledge Argument System Underlying the ATS Scheme}\label{subsection:main-zk-protocol}
Before describing our accountable tracing signature scheme in Section~\ref{subsection:ATS-construction}, let us first present the statistical $\mathsf{ZKAoK}$ that will be invoked by the signer when generating group signatures.
Let $n,q,k,\ell,m,\overline{m},\overline{m}_s,d,c_0,\cdots, c_d,\beta,B$ be parameters as specified in Section~\ref{subsection:ATS-construction}.
The protocol is summarized as follows.
\begin{itemize}
\item  The public input consists of
\begin{eqnarray*}
&\mathbf{A}, \mathbf{F}_0 \in R_q^{1 \times \overline{m}}; \hspace*{1pt}\mathbf{A}_{[0]}, \ldots, \mathbf{A}_{[d]} \in R_q^{1 \times k};
\mathbf{F} \in R_q^{1 \times \ell}; \hspace*{1pt} \\
&\mathbf{F}_1 \in R_q^{1\times \overline{m}_s};
u \in R_q;  \hspace*{1pt} \mathbf{B}\in R_q^{m};
\mathbf{c}_{1,1},\mathbf{c}_{1,2}\in R_q^{\ell}, \hspace*{1pt} \mathbf{c}_{2,1},\mathbf{c}_{2,2}\in R_q^{\ell}.
\end{eqnarray*}
\item The secret input of the prover  consists of message $\mathfrak{m}=(p\|\mathbf{a}'_1\|\mathbf{b}'_1\|\mathbf{a}'_2\|\mathbf{b}'_2)$ and the corresponding Ducas-Micciancio signature $(t, \mathbf{r}, \mathbf{v})$, a user secret key  $ \mathbf{x}$ that corresponds to the public key $p$, and encryption randomness   $g'_1, g'_2 , \mathbf{e}'_{1,1},\mathbf{e}'_{1,2},\mathbf{e}'_{2,1},\mathbf{e}'_{2,2} $,  where
\begin{eqnarray*}\hspace*{-10pt}
    \begin{cases}
    p\in R_q;  \hspace*{6.8pt} \mathbf{a}'_1\in R_q^{\ell};\hspace*{6.8pt}  \mathbf{b}'_1\in R_q^{\ell};\hspace*{6.8pt} \mathbf{a}'_2\in R_q^{\ell}; \hspace*{6.8pt} \mathbf{b}'_2\in R_q^{\ell};\\
    t = (t_0, \ldots, t_{c_1 -1}, \ldots , t_{c_{d-1}}, \ldots, t_{c_d -1})^\top \in \{0,1\}^{c_d}; \\[2.6pt]
    \mathbf{r} \in R^{\overline{m}}; \hspace*{6.8pt}\mathbf{v} = (\mathbf{s} \| \mathbf{z})\in R^{\overline{m} + k};
    \hspace*{6.8pt} \mathbf{s} \in R^{\overline{m}}; \hspace*{6.8pt}
    \mathbf{z} \in R^k;\\
   \mathbf{x}\in R^m; \hspace*{6.8pt} g'_1, g'_2  \in R; \hspace*{6.8pt} \mathbf{e}'_{1,1},\mathbf{e}'_{1,2},\mathbf{e}'_{2,1},\mathbf{e}'_{2,2}\in R^{\ell}.
    \end{cases}
    \end{eqnarray*}
\item The goal of the prover  is to prove in $\mathsf{ZK}$ that $\|\mathbf{r}\|_\infty \leq \beta$, $\|\mathbf{v}\|_\infty \leq \beta$, $\|\mathbf{x}\|_{\infty}\leq 1$, $\|{g}'_i\|_{\infty}\leq B$, $\|\mathbf{e}_{i,1}\|_{\infty}\leq B$, $\|\mathbf{e}_{i,2}\|_{\infty}\leq B$   and that the following conditions hold:
\[
\mathbf{A}_t \cdot \mathbf{v} =  \mathbf{F}\cdot \rdec\left(\mathbf{F}_0 \cdot \mathbf{r} + \mathbf{F}_1 \cdot \rdec(\mathfrak{m})\right) + u,
\]
\[
\mathbf{B}\cdot \mathbf{x} = p,
\]
\begin{eqnarray}\label{equation:for-ct-step-0}
\hspace*{-16pt}   \text{for}~ i\in\{1,2\},~
\mathbf{c}_{i,1}=\mathbf{a}'_i\cdot g'_i+\mathbf{e}'_{i,1}, \hspace{6.8pt}
\mathbf{c}_{i,2}=\mathbf{b}'_{i}\cdot g'_i+\mathbf{e}'_{i,2}+\lfloor q/4\rfloor \cdot \rdec(p).
\end{eqnarray}

\end{itemize}
Since we already established the transformations for  the Ducas-Micciancio signature in Section~\ref{subsection:zk-for-DM}, we now focus on the transformations for other relations.

Let $\mathbf{a}'_i=(a'_{i,1},\ldots, a'_{i,\ell})^\top$, $\mathbf{b}'_i=(b'_{i,1},\ldots, b'_{i,\ell})\top$ for each $i\in\{1,2\}$.  First, we employ the decomposition techniques in Section~\ref{subsection:decomposition} to the following secrets.
\begin{itemize}
\item Let $\mathbf{x}^\star=\tau(\mathbf{x})\in\{-1,0,1\}^{nm}$.

\item For each $i\in\{1,2\}$, each $j\in[\ell]$, compute $\mathbf{a}_{i,j}^{\star}=\tau(\rdec(a'_{i,j}))\in\{-1,0,1\}^{n\ell}$, $\mathbf{b}_{i,j}^{\star}=\tau(\rdec(b'_{i,j}))\in\{-1,0,1\}^{n\ell}$. \smallskip

\item For $i\in\{1,2\}$, compute $\mathbf{g}_i^\star=\tau(\rdec_B(g'_i))\in\{-1,0,1\}^{n\delta_B}$.\smallskip

\item For $i\in\{1,2\}$, compute $\mathbf{e}_{i,1}^\star=\tau(\rdec_B(\mathbf{e}'_{i,1}))\in\{-1,0,1\}^{n\ell\delta_B}$  and $\mathbf{e}_{i,2}^\star=\tau(\rdec_B(\mathbf{e}'_{i,2}))\in\{-1,0,1\}^{n\ell\delta_B}$. 

\end{itemize}
Then the equation $\mathbf{B} \cdot \mathbf{x} = p$ over $R_q$ is equivalent to
\begin{eqnarray}\label{equation:for-x-step-1}
[\mathsf{rot}(\mathbf{B})]\cdot \mathbf{x}^\star-[\mathbf{H}]\cdot \tau(\rdec(p))=\mathbf{0}^{n} \bmod q.
\end{eqnarray}
For each $i\in\{1,2\}$, each $j\in[\ell]$, let
\begin{eqnarray}\label{equation:for-expand-a-b}
\begin{cases}
\mathbf{y}_{i,j}=\mathsf{expd}\hspace*{1.8pt}(\mathbf{a}_{i,j}^{\star},~\mathbf{g}_i^{\star})\in\{-1,0,1\}^{n^2\ell\delta_B}, \\
\mathbf{z}_{i,j}=\mathsf{expd}\hspace*{1.8pt}(\mathbf{b}_{i,j}^{\star},~\mathbf{g}_i^{\star})\in\{-1,0,1\}^{n^2\ell\delta_B}.
\end{cases}
\end{eqnarray}
From Section~\ref{subsection:zk-for-rlwe}, we know that equations in~(\ref{equation:for-ct-step-0}) can be written as, for $i\in\{1,2\}$,
\begin{eqnarray}\label{equation:for-ct-step-1}
\begin{cases}
\tau(\mathbf{c}_{i,1})=\left[\mathbf{Q}\right]\cdot (\mathbf{y}_{i,1}\|\cdots\|\mathbf{y}_{i,\ell})+\left[\mathbf{H}_{\ell,B}\right]\cdot \mathbf{e}_{i,1}^{\star}; \\
\tau(\mathbf{c}_{i,2})=\left[\mathbf{Q}\right]\cdot (\mathbf{z}_{i,1}\|\cdots\|\mathbf{z}_{i,\ell})+\left[\mathbf{H}_{\ell,B}\right]\cdot \mathbf{e}_{i,2}^{\star}+
\lfloor q/4\rfloor \cdot \tau(\rdec(p)).
\end{cases}
\end{eqnarray}
Following the procedure in Section~\ref{subsection:zk-for-DM}, we form secret vectors $\mathbf{w}_1\in \{-1,0,1\}^{(k\delta_\beta + c_d k\delta_\beta)n}$, $\mathbf{w}_2\in \{-1,0,1\}^{2n\overline{m}\delta_\beta+ n\ell+n\overline{m}_s}$ of the form:
\begin{eqnarray*}
\begin{cases}
\mathbf{w}_1 = (\mathbf{z}^\star \hspace*{3.6pt}\|\hspace*{3.6pt} t_0 \cdot \mathbf{z}^\star
\hspace*{3.6pt}\|\hspace*{3.6pt} \ldots \hspace*{3.6pt}\|\hspace*{3.6pt} t_{c_d-1}\cdot \mathbf{z}^\star); \\[1.6pt]
\mathbf{w}_2 = (\mathbf{s}^\star \hspace*{3.6pt}\|\hspace*{3.6pt} \mathbf{r}^\star \hspace*{3.6pt}\|\hspace*{3.6pt}  \tau(\mathbf{y}) \hspace*{3.6pt}\|\hspace*{3.6pt} \tau(\rdec(\mathfrak{m}))),
\end{cases}
\end{eqnarray*}
where $\tau(\rdec(\mathfrak{m}))$
\begin{align*} =&~(\tau(\rdec(p))\|\tau(\rdec(\mathbf{a}'_{1}))\|\tau(\rdec(\mathbf{b}'_{1}))\|
\hspace*{0pt}\tau(\rdec(\mathbf{a}'_{2}))\|\tau(\rdec(\mathbf{b}'_{1})))\\[2.6pt]
=&~(\tau(\rdec(p))\|\mathbf{a}_{1,1}^{\star}\|\cdots\|\mathbf{a}_{1,\ell}^{\star}\|\mathbf{b}_{1,1}^{\star}\|\cdots\|\mathbf{b}_{1,\ell}^{\star}\|
\hspace*{0pt}\mathbf{a}_{2,1}^{\star}\|\cdots\|\mathbf{a}_{2,\ell}^{\star}\|\mathbf{b}_{2,1}^{\star}\|\cdots\|\mathbf{b}_{2,\ell}^{\star}). \end{align*}
Since $\tau(\rdec(p))$ has been included in $\mathbf{w}_2$, we now combine the remaining secret vectors appearing in equations~(\ref{equation:for-x-step-1}), (\ref{equation:for-ct-step-1}) into $\mathbf{w}_3\in \{-1,0,1\}^{nm  + 4n\ell\delta_B} $ of the form \[\mathbf{w}_3=\big(\hspace*{2.6pt}
\mathbf{x}^\star \hspace*{2.6pt}\|\hspace*{2.6pt}\mathbf{e}_{1,1}^\star \hspace*{2.6pt}\|\hspace*{2.6pt} \mathbf{e}_{1,2}^\star \hspace*{2.6pt}\|\hspace*{2.6pt} \mathbf{e}_{2,1}^\star \hspace*{2.6pt}\|\hspace*{2.6pt} \mathbf{e}_{2,2}^\star\hspace*{2.6pt} \big)\]
and $\mathbf{w}_4\in\{-1,0,1\}^{4n^2\ell^2\delta_B}$ of the form
\begin{align*}
\mathbf{w}_4=\big(&\mathbf{y}_{1,1}\|\cdots\|\mathbf{y}_{1,\ell}\|\mathbf{z}_{1,1}\|\cdots\|\mathbf{z}_{1,\ell}\|
\mathbf{y}_{2,1}\|\cdots\|\mathbf{y}_{2,\ell}\|\mathbf{z}_{2,1}\|\cdots\|\mathbf{z}_{2,\ell}\big)
\end{align*}
such that for $i\in\{1,2\}$, and $j\in[\ell]$, $\mathbf{y}_{i,j},\mathbf{z}_{i,j}$ satisfy the  equations in~(\ref{equation:for-expand-a-b}).

For the sake of simplicity when defining our tailored set $ \mathsf{VALID}$ and permutation $\Gamma_{\eta}$, we rearrange our secret vectors $\mathbf{w}_2, \mathbf{w}_3$ into vector $\overline{\mathbf{w}}_2\in\{-1,0,1\}^{L'_2}$ of the form \[\overline{\mathbf{w}}_2=\big(\mathbf{s}^\star \hspace*{1pt}\|\hspace*{1pt} \mathbf{r}^\star \hspace*{1pt}\|\hspace*{1pt}  \tau(\mathbf{y}) \hspace*{1pt}\|\hspace*{1pt} \tau(\rdec(p))\hspace*{1pt}\| \hspace*{1pt}\mathbf{x}^\star  \hspace*{1pt}\|\hspace*{1pt}\mathbf{e}_{1,1}^\star \hspace*{1pt}\|\hspace*{1pt} \mathbf{e}_{1,2}^\star \hspace*{1pt}\|\hspace*{1pt} \mathbf{e}_{2,1}^\star \hspace*{1pt}\|\hspace*{1pt} \mathbf{e}_{2,2}^\star\hspace*{1pt}\big).\]
and $\overline{\mathbf{w}}_3\in\{-1,0,1\}^{4n\ell^2}$ of the form
\begin{align*}
\overline{\mathbf{w}}_3=\big(\mathbf{a}_{1,1}^{\star}\hspace*{2.6pt}\|\cdots\|\hspace*{2.6pt}\mathbf{a}_{1,\ell}^{\star}\hspace*{2.6pt}\|\hspace*{2.6pt}\mathbf{b}_{1,1}^{\star}\hspace*{2.6pt}\|\cdots\|\mathbf{b}_{1,\ell}^{\star}\hspace*{2.6pt}\|\hspace*{2.6pt}
\mathbf{a}_{2,1}^{\star}\hspace*{2.6pt}\|\cdots\|\hspace*{2.6pt}\mathbf{a}_{2,\ell}^{\star}\hspace*{2.6pt}\|\hspace*{2.6pt}\mathbf{b}_{2,1}^{\star}\hspace*{2.6pt}\|\cdots\|\hspace*{2.6pt}\mathbf{b}_{2,\ell}^{\star}\big)
\end{align*}
with $L'_2=2n\overline{m}\delta_{\beta}+2n\ell+nm+4n\ell\delta_B$.
Now we form our secret vector as  $\mathbf{w}_0=(\mathbf{w}_1\|\overline{\mathbf{w}}_2\|\overline{\mathbf{w}}_3\|\mathbf{w}_4)$. \smallskip

Second, we apply the extension and permutation techniques from Section~\ref{subsection:permutation-technique-by-lnwx} and Section~\ref{subsection:extended-permutation-technique} to our secret vectors $\mathbf{w}_0$.  Let $\mathbf{w}'_1=\mathsf{mix}(t,\mathbf{z}^{\star})\in\{-1,0,1\}^{L_1}$ be the ``mixing'' vector obtained in equation~(\ref{eq:zk-w'_1}), $\mathbf{w}'_2=\mathsf{enc}(\overline{\mathbf{w}}_2)\in\{-1,0,1\}^{L_2}$
$\mathbf{w}'_3=\mathsf{enc}(\overline{\mathbf{w}}_3)\in\{-1,0,1\}^{L_3}$, 
and $\mathbf{w}'_4=\mathsf{Mult}(\mathbf{w}_4)\in\{-1,0,1\}^{L_4}$ be of the following form:
\begin{align*}
\big(\hspace{1.6pt}&
\mathsf{mult}(\mathbf{a}_{1,1}^{\star},\mathbf{g}_1^{\star})\hspace{1.6pt}\|\cdots\|\hspace{1.6pt}\mathsf{mult}(\mathbf{a}_{1,\ell}^{\star},\mathbf{g}_1^{\star}) \hspace{1.6pt}\|
\mathsf{mult}(\mathbf{b}_{1,1}^{\star},\mathbf{g}_1^{\star})\hspace{1.6pt}\|\cdots\|\hspace{1.6pt}\mathsf{mult}(\mathbf{b}_{1,\ell}^{\star},\mathbf{g}_1^{\star}) \hspace{1.6pt}\| \\
&\mathsf{mult}(\mathbf{a}_{2,1}^{\star},\mathbf{g}_2^{\star})\hspace{1.6pt}\| \cdots \|\hspace{1.6pt}\mathsf{mult}(\mathbf{a}_{2,\ell}^{\star},\mathbf{g}_2^{\star})\hspace{1.6pt}\| \hspace{1.6pt}
\mathsf{mult}(\mathbf{b}_{2,1}^{\star},\mathbf{g}_2^{\star})\hspace{1.6pt}\| \cdots\|\hspace{1.6pt}\mathsf{mult}(\mathbf{b}_{2,\ell}^{\star},\mathbf{g}_2^{\star})
\hspace{1.6pt}\big),
\end{align*}
Where $L_1=3k\delta_\beta+6nk\delta_{\beta}c_d$, $L_2=3L'_2$, $L_3=12n\ell^2$, and $L_4=36n^2\ell^2\delta_B$.
Denote $L=L_1+L_2+L_3+L_4$. Form our extended vector $\mathbf{w}=(\mathbf{w}'_1\|\mathbf{w}'_2\|\mathbf{w}'_3\|\mathbf{w}'_4)\in\{-1,0,1\}^{L}$. \smallskip

Following the process in Section~\ref{subsection:zk-for-DM} and Section~\ref{subsection:zk-for-rlwe}, we are able to obtain public matrix/vector $\mathbf{M}$ and $\mathbf{u}$ such that the considered statement is reduced to $\mathbf{M}\cdot \mathbf{w} = \mathbf{u} \bmod q$. Therefore, we are prepared to define the set $\mathsf{VALID}$ that includes our secret vector $\mathbf{w}$, the set $\mathcal{S}$, and the associated  permutations $\{\Gamma_\eta: \eta \in \mathcal{S}\}$, such that the conditions in~(\ref{eq:zk-equivalence}) are satisfied. \smallskip


\noindent Let $ \mathsf{VALID} $ be the set of all vectors $\mathbf{v}' = (\mathbf{v}'_1 \| \mathbf{v}'_2\|\mathbf{v}'_3\|\mathbf{v}'_4) \in \{-1,0,1\}^{L}$ such that the following requirements hold:
\begin{itemize}
\item $\mathbf{v}'_1 = \mathsf{mix}(t, \mathbf{z}^\star)$ for some $t \in \{0,1\}^{c_d}$ and $\mathbf{z}^\star \in \{-1,0,1\}^{nk\delta_\beta}$.  \smallskip
\item $\mathbf{v}'_2 = \mathsf{enc}(\overline{\mathbf{w}}_2)$ for some $\overline{\mathbf{w}}_2 \in \{-1,0,1\}^{L'_2}$. \smallskip
\item For $ j\in[4\ell]$, there exists $\overline{\mathbf{w}}_{3,j}\in\{-1,0,1\}^{n\ell}$ and $\overline{\mathbf{w}}_3=(\overline{\mathbf{w}}_{3,1}\cdots\|\overline{\mathbf{w}}_{3,4\ell}) \in \{-1,0,1\}^{4n\ell^2}$   such that $\mathbf{v}'_3 = (\mathsf{enc}(\overline{\mathbf{w}}_{3,1})\|\cdots\|\mathsf{enc}(\overline{\mathbf{w}}_{3,4\ell}))=\mathsf{enc}(\overline{\mathbf{w}}_3)$. \smallskip
\item  There exists $\mathbf{g}_1^{\star},\mathbf{g}_2^{\star}\in\{-1,0,1\}^{n\delta_B}$ and  $\mathbf{w}_4\in \{-1,0,1\}^{4n^2\ell^2\delta_B}$ be of the form:
\begin{align*}
(\mathsf{expd}(\overline{\mathbf{w}}_{3,1},\mathbf{g}_1^{\star})\|\cdots\|\mathsf{expd}(\overline{\mathbf{w}}_{3,2\ell},\mathbf{g}_1^{\star})\|
\mathsf{expd}(\overline{\mathbf{w}}_{3,2\ell+1},\mathbf{g}_2^{\star})\|\cdots\|\mathsf{expd}(\overline{\mathbf{w}}_{3,4\ell},\mathbf{g}_2^{\star})) \end{align*}  such that  $\mathbf{v}'_4 = \mathsf{Mult}( \mathbf{w}_4)$.
\end{itemize}
It is verifiable  that our secret vector $ \mathbf{w} $ belongs to $ \mathsf{VALID}$. \smallskip

Now  let $\mathcal{S}=  \{0,1\}^{c_d}\times \{-1,0,1\}^{nk\delta_\beta}\times\{-1,0,1\}^{L'_2}\times(\{-1,0,1\}^{n\ell})^{4\ell}\times(\{-1,0,1\}^{n\delta_B})^2$, and associate every element $$\eta=(\mathbf{f}_1,\mathbf{f}_2,\mathbf{f}_3,\mathbf{f}_{4,1},\ldots,\mathbf{f}_{4,2\ell},\mathbf{f}_{5,1},\ldots,\mathbf{f}_{5,2\ell},\mathbf{f}_6,\mathbf{f}_7)\in\mathcal{S}$$
 with $\Gamma_{\eta}$ that works as follows. For a vector of form $\mathbf{v}^{\star}=(\mathbf{v}_1^{\star}\|\mathbf{v}_2^{\star}\|\mathbf{v}_{3}^{\star}\|\mathbf{v}_4^{\star})\in\mathbb{Z}^{L}$, where
 $ \mathbf{v}_i^{\star}\in\mathbb{Z}^{L_i}$ for $i\in\{1,2\}$,  $\mathbf{v}_{3}^{\star}=(\mathbf{v}_{3,1}^{\star}\|\cdots\|\mathbf{v}_{3,4\ell}^{\star})$ with $\mathbf{v}_{3,j}^{\star}\in\mathbb{Z}^{3n\ell}$, and $\mathbf{v}_{4}^{\star}=(\mathbf{v}_{4,1}^{\star}\cdots\|\mathbf{v}_{4,4\ell}^{\star})$ with $\mathbf{v}_{4,j}^{\star}\in\mathbb{Z}^{9n^2\ell\delta_B}$, it transforms $\mathbf{v}^{\star}$ into vector $\Gamma_{\eta}(\mathbf{v}^{\star})$
 \begin{align*}
 ( \hspace*{1.6pt}
 &\Psi_{\mathbf{f}_1,\mathbf{f}_2}(\mathbf{v}_1^{\star})\hspace*{1.6pt}\|\hspace*{1.6pt}\Pi_{\mathbf{f}_3}(\mathbf{v}_2^{\star})\hspace*{1.6pt}\| \\
 &\Pi_{\mathbf{f}_{4,1}}(\mathbf{v}_{3,1}^{\star})  \hspace*{1.6pt} \| \hspace*{1.6pt} \cdots \hspace*{1.6pt} \| \hspace*{1.6pt} 
\Pi_{\mathbf{f}_{4,2\ell}}(\mathbf{v}_{3,2\ell}^{\star})  \hspace*{1.6pt} \|    \hspace*{1.6pt}
 \Pi_{\mathbf{f}_{5,1}}(\mathbf{v}_{3,2\ell+1}^{\star}) \hspace*{1.6pt} \| \cdots \| \hspace*{1.6pt}
 \Pi_{\mathbf{f}_{5,2\ell}}(\mathbf{v}_{3,4\ell}^{\star})\|
 \\
 &\Phi_{\mathbf{f}_{4,1},\mathbf{f}_6}(\mathbf{v}_{4,1}^{\star}) \hspace*{1.6pt} \| \cdots\| \hspace*{1.6pt} 
  \Phi_{\mathbf{f}_{4,2\ell},\mathbf{f}_6}(\mathbf{v}_{4,2\ell}^{\star})  \hspace*{1.6pt} \|   \hspace*{1.6pt}
  \Phi_{\mathbf{f}_{5,1},\mathbf{f}_7}(\mathbf{v}_{4,2\ell+1}^{\star}) \hspace*{1.6pt}\| \cdots\|  \hspace*{1.6pt} 
 \Phi_{\mathbf{f}_{5,2\ell},\mathbf{f}_7}(\mathbf{v}_{4,4\ell}^{\star})
 \hspace*{1.6pt})
 \end{align*}
It then follows from the equivalences in~(\ref{eq:equivalence-enc-vector}),~(\ref{eq:equivalence-mix}), and~(\ref{eq:equivalence-ext-vector}) that $ \mathsf{VALID}$, $\mathcal{S}$, and $\Gamma_\eta$ satisfy the conditions in~(\ref{eq:zk-equivalence}). Therefore, we have transformed the considered statement to a case of the abstract protocol from Section~\ref{subsection:Stern}. To obtain the desired statistical $\mathsf{ZKAoK}$ protocol, it suffices for the prover and verifier to run the interactive protocol described in Figure~\ref{Figure:Interactive-Protocol}. The protocol has perfect completeness, soundness error $2/3$ and communication cost $\mathcal{O}(L \cdot \log q)$, which is of the order $\mathcal{O}(n^2\cdot\log^3n)=\widetilde{\mathcal{O}}(\lambda^2)$.

 \subsection{Description of Our ATS Scheme}\label{subsection:ATS-construction}
We assume there is a trusted setup such that it generates parameters of the scheme. Specifically, it generates a public matrix $\mathbf{B}$ for generating users' key pairs, and two secret-public key pairs of our $\mathsf{KOE}$ scheme such that the secret keys are discarded and not known by any party. The group public key then consists of three parts: (i) the parameters from the trusted setup, (ii) a verification key of the Ducas-Micciancio signature, (iii) two public keys of our $\mathsf{KOE}$ scheme such that the group manager knows both secret keys. The issue key is the Ducas-Micciancio signing key,  while the opening key is any one of the corresponding secret keys of the two public keys. Note that both the issue key and the opening key are generated by the group manager.

When a user joins the group,  it first generates a secret-public key pair $(\mathbf{x},p)$ such that  $\mathbf{B}\cdot \mathbf{x}=p$. It then interacts with the group manager, who will determine whether user $p$ is traceable or not. If the user is traceable, group manager sets a bit $\mathsf{tr}=1$, randomizes the two public key generated by himself, and then generates a Ducas-Micciancio signature $\sigma_{\mathsf{cert}}$ on  user public key $p$ and the two randomized public keys ($\mathsf{epk}_1,\mathsf{epk}_2$). If the user is non-traceable, group manager sets a bit $\mathsf{tr}=0$, randomizes the two public key generated from the trusted setup, and then generates a signature on  $p$ and $\mathsf{epk}_1,\mathsf{epk}_2$. If it completes successfully, the group manager sends certificate $\mathsf{cert}=(p,\mathsf{epk}_1,\mathsf{epk}_2,\sigma_{\mathsf{cert}})$ to user $p$, registers this user to the group, and keeps himself the witness $w^{\mathsf{escrw}}$ that was ever used for randomization.

Once registered as a group member, the user can sign messages on behalf of the group. To this end, the user first encrypts his public key $p$ twice using his two randomized public keys, and obtains ciphertexts $\mathbf{c}_1,\mathbf{c}_2$. The user then generates a $\mathsf{ZKAoK}$ such that (i) he has a valid secret key $\mathbf{x}$ corresponding to $p$; (ii) he possesses a Ducas-Micciancio signature on $p$ and $\mathsf{epk}_1,\mathsf{epk}_2$; and (iii) $\mathbf{c}_1,\mathbf{c}_2$ are   correct ciphertexts of~$p$ under the randomized keys $\mathsf{epk}_1,\mathsf{epk}_2$, respectively.  Since the $\mathsf{ZKAoK}$ protocol the user employs has soundness error $2/3$ in each execution, it is repeated $\kappa=\omega(\log \lambda)$ times to make the error negligibly small. Then, it is made non-interactive via the Fiat-Shamir heuristic~\cite{FS86}. The signature then consists of the non-interactive zero-knowledge argument of knowledge (\textsf{NIZKAoK}) $\Pi_{\mathsf{gs}}$ and the two ciphertexts. Note that the $\mathsf{ZK}$ argument together with double encryption enables CCA-security of the underlying encryption scheme, which is known as the Naor-Yung transformation~\cite{NY90}.

To verify the validity of a signature, it suffices to  verify the validity of the argument $\Pi_{\mathsf{gs}}$. Should the need arises, the group manager can decrypt  using his opening key. If a user is traceable,  the opening key group manager possesses can be used to correctly identify the signer.  However, if a user is non-traceable,  then his anonymity is preserved against the manager.

To prevent corrupted opening, group manager is required to generate a $\mathsf{NIZKAoK}$ of correct opening $\Pi_{\mathsf{open}}$. Only when $\Pi_{\mathsf{open}}$ is a valid argument, we then accept the opening result.
Furthermore, there is an additional accounting mechanism for group manager to reveal which users he had chosen to be traceable. This is done by checking the consistency of $\mathsf{tr}$ and the randomized public keys in user's certificate with the help of the witness $w^{\mathsf{escrw}}$.

We describe the details of our scheme below.


\begin{description}
\item[$\mathsf{Setup}(\lambda)$:] Given the security parameter $\lambda$, it generates the following public parameter.
 \begin{itemize}
    \item  Let $n=\mathcal{O}(\lambda)$ be a power of $2$, and modulus $q=\widetilde{\mathcal{O}}(n^4)$, where $q=3^k$ for $k\in\mathbb{Z}^{+}$. Let  $R=\mathbb{Z}[X]/(X^n+1)$ and $R_q=R/qR$.

        Also, let $m\geq 2\lceil\log q\rceil +2$,  $\ell=\lfloor\log \frac{q-1}{2}\rfloor +1$,  $m_s=4\ell+1$, and $\overline{m} = m + k$ and $\overline{m}_s=m_s\cdot \ell$.

    \item Let integer $d$ and sequence $c_0,\ldots,c_d$ be described in Section~\ref{subsection:DM-signatures}.

    \item Let  $\beta=\widetilde{\mathcal{O}}(n)$ and $B=\widetilde{\mathcal{O}}(\sqrt{n})$ be two integer bounds, and $\chi$ be a $B$-bounded distribution over the ring $R$.

    \item Choose a collision-resistant hash function $\mathcal{H}_{\mathsf{FS}}:\{0,1\}^*\rightarrow \{1,2,3\}^{\kappa}$, where $\kappa=\omega(\log\lambda)$,  which will act as a random oracle in the Fiat-Shamir heuristic~\cite{FS86}.

    \item Choose a statistically hiding and computationally binding commitment scheme from~\cite{KTX08}, denoted as $\mathsf{COM}$,  which will be employed in our $\mathsf{ZK}$ argument systems. 

     \item Let $\mathbf{B} \xleftarrow{\$} R_q^{1 \times m}$, $\mathbf{a}_1^{(0)}\xleftarrow{\$} R_q^{\ell}$, $\mathbf{a}_2^{(0)}\xleftarrow{\$} R_q^{\ell}$, $s_{-1},s_{-2}\hookleftarrow\chi$, $\mathbf{e}_{-1},\mathbf{e}_{-2}\hookleftarrow \chi^{\ell}$. Compute

     $$\mathbf{b}_1^{(0)}=\mathbf{a}_1^{(0)}\cdot s_{-1}+\mathbf{e}_{-1}\in R_q^{\ell}; \hspace*{12pt} \mathbf{b}_2^{(0)}=\mathbf{a}_2^{(0)}\cdot s_{-2}+\mathbf{e}_{-2}\in R_q^{\ell}.$$

    \end{itemize}
    This algorithm outputs the
    public parameter $\mathsf{pp}$: \begin{eqnarray*} &\{\hspace*{6.6pt}n,q,k,R,R_q,\ell,m,m_s,\overline{m},\overline{m}_s,d,c_0,\cdots,c_d,\\
    &\beta, B,\chi, \mathcal{H}_{\mathsf{FS}},\kappa,\mathsf{COM},\mathbf{B},\{\mathbf{a}_i^{(0)}, \mathbf{b}_i^{(0)}\}_{i\in\{1,2\}}\hspace*{6.6pt}\}.
     \end{eqnarray*}$\mathsf{pp}$ is implicit for all algorithms below if not explicitly mentioned.
    \smallskip
  \item[$\mathsf{GKeyGen}(\mathsf{pp})$:] On input $\mathsf{pp}$, $\mathsf{GM}$ proceeds as follows. 
      \begin{itemize}
            \item Generate verification key \begin{eqnarray*}
            \mathbf{A}, \mathbf{F}_0 \in R_q^{1 \times \overline{m}}; \hspace*{2.8pt}\mathbf{A}_{[0]}, \ldots, \mathbf{A}_{[d]} \in R_q^{1 \times k};\hspace*{2.8pt}
             \mathbf{F}\in R_q^{1 \times \ell};\hspace*{2.8pt} \mathbf{F}_1 \in R_q^{1 \times \overline{m}_s};  \hspace*{2.8pt}u \in R_q
                \end{eqnarray*} and signing key $\mathbf{R}\in R_q^{m\times k}$ for the Ducas-Micciancio signature from Section~\ref{subsection:DM-signatures}. 

        \item Initialize the Naor-Yung double-encryption mechanism~\cite{NY90} with the key-oblivious encryption scheme described in Section~\ref{subsection:KOE-construction}.   Specifically, sample $s_1,s_2 \hookleftarrow \chi$, $\mathbf{e}_1,\mathbf{e}_2 \hookleftarrow \chi^{\ell}$,
         $\mathbf{a}_1^{(1)}\xleftarrow{\$} R_q^{\ell}$, $\mathbf{a}_2^{(1)}\xleftarrow{\$} R_q^{\ell}$ and compute
        \[
        \mathbf{b}_1^{(1)}=\mathbf{a}_1^{(1)}\cdot s_1+\mathbf{e}_1\in R_q^{\ell};\hspace*{12pt}
        \mathbf{b}_2^{(1)}=\mathbf{a}_2^{(1)}\cdot s_2+\mathbf{e}_2\in R_q^{\ell}.
        \]

    \end{itemize}

  Set the group public key $\mathsf{gpk}$, the issue key $\mathsf{ik}$ and the opening key $\mathsf{ok}$ as follows: \begin{eqnarray*}
  \mathsf{gpk}=\{\mathsf{pp},\mathbf{A},\{\mathbf{A}_{[j]}\}_{j=0}^{d}, \mathbf{F},\mathbf{F}_0,\mathbf{F}_1,u,  \mathbf{a}_1^{(1)},\mathbf{b}_1^{(1)},\mathbf{a}_2^{(1)},\mathbf{b}_2^{(1)}~~\},\end{eqnarray*}
\[\mathsf{ik}=\mathbf{R}, ~~~\mathsf{ok}=(s_1,\mathbf{e}_1).\]

$\mathsf{GM}$ then makes $\mathsf{gpk}$ public, sets the registration table $\mathbf{reg}=\emptyset$ and his internal state $S=0$.
\smallskip
  \item[$\mathsf{UKeyGen}(\mathsf{pp})$:] Given the public parameter, the user first chooses $\mathbf{x} \in R^m$ such that the coefficients are uniformly chosen from the set $\{-1,0,1\}$. He then calculates $p=\mathbf{B}\cdot\mathbf{x}\in R_q$. Set $\mathsf{upk} = p$ and $\mathsf{usk}=\mathbf{x}$. 

  \item[$\mathsf{Enroll}(\mathsf{gpk},\mathsf{ik},\mathsf{upk},\mathsf{tr})$:] Upon receiving a user public key $\mathsf{upk}$ from a user, $\mathsf{GM}$ determines the value of the bit $\mathsf{tr}\in\{0,1\}$, indicating whether the  user is traceable. He then does the following: 
      \begin{itemize}
      \item   Randomize two pairs of public keys $(\mathbf{a}_1^{(\mathsf{tr})},\mathbf{b}_1^{(\mathsf{tr})})$ and $(\mathbf{a}_2^{(\mathsf{tr})},\mathbf{b}_2^{(\mathsf{tr})})$ as described in Section~\ref{subsection:KOE-construction}.  Specifically, sample $g_1,g_2 \hookleftarrow \chi$, $\mathbf{e}_{1,1},\mathbf{e}_{1,2}\hookleftarrow \chi^{\ell}$, $\mathbf{e}_{2,1},\mathbf{e}_{2,2} \hookleftarrow \chi^{\ell}$. For each $i\in\{1,2\}$, compute
          \begin{eqnarray}\label{equation:randomization-of-public-key}
          \mathsf{epk}_i=(\mathbf{a}'_i,\mathbf{b}_i')=(\mathbf{a}_{i}^{(\mathsf{tr})}\cdot g_i+\mathbf{e}_{i,1},\hspace*{6.8pt}\mathbf{b}_{i}^{(\mathsf{tr})}\cdot g_i+\mathbf{e}_{i,2})\in R_q^{\ell}\times R_q^{\ell}.
          \end{eqnarray}
      \item Set the tag $t=(t_0,t_1\ldots, t_{c_d-1})^\top\in \mathcal{T}_d$, where $S=\sum_{j=0}^{c_d-1} 2^j\cdot t_j$, and
        compute $\mathbf{A}_{t} = [\mathbf{A}|\mathbf{A}_{[0]}+\sum_{i=1}^{d}t_{[i]}\mathbf{A}_{[i]}] \in R_q^{1\times (\overline{m} + k)}$.  \smallskip
     \item Let $\mathfrak{m}=(p\|\mathbf{a}'_1\|\mathbf{b}'_1\|\mathbf{a}'_2\|\mathbf{b}'_2)\in R_q^{m_s}$. \smallskip

    \item  Generate a signature $\sigma_{\mathsf{cert}}=(t,\mathbf{r},\mathbf{v})$
         on message $\rdec(\mathfrak{m}) \in R^{\overline{m}_s}$ - whose coefficients are in $\{-1,0,1\}$ - using his issue key $\mathsf{ik}=\mathbf{R}$. As in Section~\ref{subsection:DM-signatures}, we have $\mathbf{r} \in R^{\overline{m}}$, $\mathbf{v} \in R^{\overline{m} + k}$ and
        \begin{eqnarray}\label{equation:DM-certificate-condition}
         \begin{cases}
            \mathbf{A}_t\cdot\mathbf{v}=\mathbf{F}\cdot\rdec(\mathbf{F}_0\cdot\mathbf{r}+\mathbf{F}_1\cdot\rdec(\mathfrak{m}))+u,\\
            \|\mathbf{r}\|_{\infty}\leq \beta,~~\|\mathbf{v}\|_{\infty}\leq \beta.
         \end{cases}
        \end{eqnarray}


      \end{itemize}
    Set certificate $\mathsf{cert}$ and $w^{\mathsf{escrw}}$   as follows:
 \[\mathsf{cert}=(p,\mathbf{a}'_1,\mathbf{b}'_1,\mathbf{a}'_2,\mathbf{b}'_2, t,\mathbf{r}, \mathbf{v}), \hspace*{12pt}
w^{\mathsf{escrw}}=(g_1,\mathbf{e}_{1,1},\mathbf{e}_{1,2},g_2,\mathbf{e}_{2,1},\mathbf{e}_{2,2}).\]  $\mathsf{GM}$ sends $\mathsf{cert}$ to the user $p$, stores $\mathbf{reg}[S]=(p, \mathsf{tr}, w^{\mathsf{escrw}})$,  
 and  updates the state to $S+1$. 
  \item[$\mathsf{Sign}(\mathsf{gpk},\mathsf{cert},\mathsf{usk},M)$:]  To sign a message $M \in \{0,1\}^*$ using the certificate $\mathsf{cert}=(p,\mathbf{a}'_1,\mathbf{b}'_1,\mathbf{a}'_2,\mathbf{b}'_2, t,\mathbf{r}, \mathbf{v})$ and $\mathsf{usk}=\mathbf{x}$, the user proceeds as follows. 
      \begin{itemize}
        \item Encrypt the ring vector $\rdec(p)\in R_q^{\ell}$ whose coefficients are in $\{-1,0,1\}$ twice.  Namely,  sample $g'_1,g'_2 \hookleftarrow \chi$, $\mathbf{e}'_{1,1},\mathbf{e}'_{1,2}\hookleftarrow\chi ^{\ell}$, and $\mathbf{e}'_{2,1},\mathbf{e}'_{2,2}\hookleftarrow\chi ^{\ell}$. For each $i\in\{1,2\}$,  compute $\mathbf{c}_i=(\mathbf{c}_{i,1},\mathbf{c}_{i,2})\in R_q^{\ell}\times R_q^{\ell}$ as follows:
            \begin{align*}
            \mathbf{c}_{i,1}=\mathbf{a}'_i\cdot g'_i+\mathbf{e}'_{i,1}; \hspace*{12pt}\mathbf{c}_{i,2}=\mathbf{b}'_{i}\cdot g'_i+\mathbf{e}'_{i,2}+\lfloor q/4\rfloor\cdot \rdec(p).
            \end{align*}
        \item Generate a $\mathsf{NIZKAoK}$ $\Pi_{\mathsf{gs}}$ to demonstrate the possession of a valid tuple $\zeta$ of the following form
                \begin{align}\label{equation:ATS-witness-zeta}
               \zeta= (p,\mathbf{a}'_1,\mathbf{b}'_1,\mathbf{a}'_2,\mathbf{b}'_2, t,\mathbf{r}, \mathbf{v},\mathbf{x}, g'_1,\mathbf{e}'_{1,1},\mathbf{e}'_{1,2},g'_2,\mathbf{e}'_{2,1},\mathbf{e}'_{2,2})
                \end{align}
                such that
                    \begin{enumerate}[(i)]
                    \item\label{condition-1} The conditions in~(\ref{equation:DM-certificate-condition}) are satisfied. \smallskip
                    \item\label{condition-2} $\mathbf{c}_1$ and $\mathbf{c}_2$ are correct encryptions of $\rdec(p)$ with $B$-bounded randomness $g'_1,\mathbf{e}'_{1,1},\mathbf{e}'_{1,2}$ and $g'_2,\mathbf{e}'_{2,1},\mathbf{e}'_{2,2}$, respectively. \smallskip
                    \item\label{condition-3} $\|\mathbf{x}\|_\infty \leq 1$ and $\mathbf{B}\cdot\mathbf{x}=p$. \smallskip
                    \end{enumerate}
                    This is achieved by running the protocol from Section~\ref{subsection:main-zk-protocol},  
                which is repeated $\kappa=\omega(\log \lambda)$ times and made non-interactive via Fiat-Shamir heuristic~\cite{FS86} as a triple $\Pi_{\mathsf{gs}}=(\{\mathrm{CMT}_i\}_{i=1}^{\kappa},\mathrm{CH},\{\mathrm{RSP}_i\}_{i=1}^{\kappa})$ where the challenge $\mathrm{CH}$ is generated as $\mathrm{CH}=\mathcal{H}_{\mathsf{FS}}(M,\{\mathrm{CMT}_i\}_{i=1}^{\kappa},\xi)$ with  $\xi$ of the following form
                \begin{equation}\label{equation:public-input}
               \hspace*{-6.6pt}\xi=(\mathbf{A},\mathbf{A}_{[0]},\ldots,\mathbf{A}_{[d]},\mathbf{F},\mathbf{F}_0,\mathbf{F}_1,u,\mathbf{B},\mathbf{c}_1,\mathbf{c}_2)
                \end{equation}
            \item Output the group signature $\Sigma=(\Pi_{\mathsf{gs}},\mathbf{c}_1,\mathbf{c}_2)$.
        \end{itemize}
      \smallskip
\smallskip
  \item[$\mathsf{Verify}(\mathsf{gpk},M,\Sigma)$:]  Given the inputs, the verifier performs in the following manner. 
        \begin{itemize}
    \item Parse $\Sigma$ as $\Sigma = \big(\{\mathrm{CMT}_i\}_{i=1}^\kappa, (Ch_1, \ldots, Ch_\kappa), \{\mathrm{RSP}\}_{i=1}^\kappa, \mathbf{c}_1, \mathbf{c}_2\big)$.\\
      If $(Ch_1, \ldots, Ch_\kappa) \neq \mathcal{H}_{\mathsf{FS}}\big(M, \{\mathrm{CMT}_i\}_{i=1}^\kappa, \xi \big)$, output~$0$, where $\xi$ is as in~(\ref{equation:public-input}). 
  \item For each $i \in [\kappa]$, run the verification phase of the protocol in Section~\ref{subsection:main-zk-protocol} to verify the validity of $\mathrm{RSP}_i$ corresponding to $\mathrm{CMT}_i$ and $Ch_i$. If any of the verification process fails, output~$0$.
      \item Output~$1$.
      \smallskip
  \end{itemize}
\smallskip

  \item[$\mathsf{Open}(\mathsf{gpk},\mathsf{ok},M,\Sigma)$:] Let $\mathsf{ok}=(s_1,\mathbf{e}_1)$ and $\Sigma=(\Pi_{\mathsf{gs}},\mathbf{c}_1,\mathbf{c}_2)$. The group manager proceeds as follows.
      \begin{itemize}
        \item Use $s_1$ to decrypt $\mathbf{c}_1=(\mathbf{c}_{1,1},\mathbf{c}_{1,2})$ as in the decryption algorithm from Section~\ref{subsection:KOE-construction}. The result is $p'\in R_q$. 
        \item He then searches the registration information. If $\mathbf{reg}$ does not include an element $p'$, then return $\bot$. 
        \item Otherwise, he produces a $\mathsf{NIZKAoK}$ $\Pi_{\mathsf{open}}$ to show the knowledge of a tuple $(s_1,\mathbf{e}_1,\mathbf{y})\in R_q\times R_q^{\ell}\times R_q^{\ell}$ such that the following conditions hold.
            \begin{eqnarray}\label{equation:pi-open}
            \begin{cases}
                \| s_1 \|_\infty \leq B; \hspace*{2.8pt}  \| \mathbf{e}_1 \|_\infty \leq B; \hspace*{2.8pt}
                \| \mathbf{y}\|_\infty \leq  \lceil q/10 \rceil; \\
                \mathbf{a}_1^{(1)} \cdot s_1 + \mathbf{e}_1 = \mathbf{b}_1^{(1)}; \\
                \mathbf{c}_{1,2} - \mathbf{c}_{1,1}\cdot s_1 = \mathbf{y} + \lfloor q/4 \rfloor\cdot \rdec(p').
            \end{cases}
            \end{eqnarray}
       Since the conditions in~(\ref{equation:pi-open}) only encounter linear secret objects with bounded norm, we can easily handled them using the Stern-like techniques from Sections~\ref{subsection:zk-for-rlwe} and \ref{subsection:main-zk-protocol}. Therefore, we are able to have a   statistical $\mathsf{ZKAoK}$ for the above statement. Furthermore, the protocol is repeated $\kappa = \omega(\log \lambda)$ times and made non-interactive via the Fiat-Shamir heuristic, resulting in a triple $\Pi_{\mathsf{Open}}= (\{\mathrm{CMT}_i\}_{i=1}^\kappa, \mathrm{CH}, \{\mathrm{RSP}\}_{i=1}^\kappa)$, where $\mathrm{CH} \in \{1,2,3\}^\kappa$ is computed as
          \begin{eqnarray}\label{equation:pi-open-FS}
            \mathrm{CH}=\mathcal{H}_{\mathsf{FS}}\big(\{\mathrm{CMT}_i\}_{i=1}^\kappa, \mathbf{a}_{1}^{(1)},\mathbf{b}_1^{(1)}, M, \Sigma, p'\big).
          \end{eqnarray}
        \item Output $(p', \Pi_{\mathsf{Open}})$. \smallskip
\smallskip
     \end{itemize}
  \item[$\mathsf{Judge}(\mathsf{gpk},M,\Sigma, p', \Pi_{\mathsf{open}})$:] Given all the inputs, this algorithm does the following. \smallskip
  \begin{itemize}
 \item  If $\mathsf{Verify}$ algorithm outputs $0$ or $p'=\bot$,  return~$0$.  
 \item This algorithm then verifies the argument $\Pi_{\mathsf{Open}}$ with respect to common input $(\mathbf{a}_{1}^{(1)}, \mathbf{b}_1^{(1)}, M, \Sigma, p')$, in the same way as in the algorithm $\mathsf{Verify}$. If verification of the argument $\Pi_{\mathsf{open}}$ fails, output~$0$. 
 \item Else output~$1$. \smallskip \smallskip
\end{itemize}
  \item[$\mathsf{Account}(\mathsf{gpk},\mathsf{cert},w^{\mathsf{escrw}},\mathsf{tr})$:] Let the certificate be $\mathsf{cert}=(p,\mathbf{a}'_1,\mathbf{b}'_1,\mathbf{a}'_2,\mathbf{b}'_2,t,\mathbf{r},\mathbf{v})$ and witness be $w^{\mathsf{escrw}}=(g_1,\mathbf{e}_{1,1},\mathbf{e}_{1,2},g_2,\mathbf{e}_{2,1},\mathbf{e}_{2,2})$ and the bit $\mathsf{tr}$, this algorithm proceeds as follows. 
       \begin{itemize}

     \item  It checks whether $(t,\mathbf{r},\mathbf{v})$ is  a valid  Ducas-Micciancio signature on the message $(p,\mathbf{a}'_1,\mathbf{b}'_1,\mathbf{a}'_2,\mathbf{b}'_2)$. Specifically, it verifies whether $\mathsf{cert}$ satisfies the conditions in~(\ref{equation:DM-certificate-condition}). If not, output~$0$. 

     \item Otherwise, it then checks if   $(\mathbf{a}'_1,\mathbf{b}'_1)$ and $(\mathbf{a}'_2,\mathbf{b}'_2)$ are randomization of $(\mathbf{a}_{1}^{(\mathsf{tr})},(\mathbf{b}_{1}^{(\mathsf{tr})})$ and $(\mathbf{a}_{2}^{(\mathsf{tr})},(\mathbf{b}_{2}^{(\mathsf{tr})})$ with respect to randomness $(g_1,\mathbf{e}_{1,1},\mathbf{e}_{1,2})$ and $(g_2,\mathbf{e}_{2,1},\mathbf{e}_{2,2})$, respectively.
      Specifically, it verifies whether  the conditions in~(\ref{equation:randomization-of-public-key}) hold. If not, output~$0$. 
      \item Else output~$1$.

    \end{itemize}

\end{description}

\subsection{Analysis of Our ATS Scheme}\label{subsection:ATS-security-analysis}

\noindent
{\sc Efficiency. }
We first analyze the efficiency of our scheme from Section~\ref{subsection:ATS-construction} in terms of the security parameter $\lambda$.
\begin{itemize}
\item The bit-size of the public key $\mathsf{gpk}$ is of order $\mathcal{O}(\lambda\cdot\log^3 \lambda)=\widetilde{\mathcal{O}}(\lambda)$. \smallskip
\item The bit-size of the membership certificate $\mathsf{cert}$ is of order  $\mathcal{O}(\lambda\cdot\log^2 \lambda)=\widetilde{\mathcal{O}}(\lambda)$. \smallskip
\item The bit-size of  a signature $\Sigma$ is determined by that of the Stern-like  $\mathsf{NIZKAoK}$ ~ $\Pi_{\mathsf{ gs}}$, which is of order $\mathcal{O}(L\cdot \log q)\cdot \omega(\log \lambda)$, where $L$ is the bit-size  of a vector $\mathbf{w}\in\mathsf{VALID}$ from Section~\ref{subsection:main-zk-protocol}. Recall $\mathcal{O}(L\cdot \log q)=\mathcal{O}(\lambda^2\cdot\log^3\lambda)$. Therefore, the bit-size of $\Sigma$ is of order $\mathcal{O}(\lambda^2\cdot \log^3 \lambda) \cdot\omega(\log \lambda)=\widetilde{\mathcal{O}}(\lambda^2)$.\smallskip

\item The bit-size of the Stern-like $\mathsf{NIZKAoK}$ ~ $\Pi_{\mathsf{open}}$ is of order $\mathcal{O}(\lambda\cdot\log^3 \lambda)\cdot \omega(\log \lambda)=\widetilde{\mathcal{O}}(\lambda)$.
\end{itemize}

\smallskip
\noindent
{\sc Correctness.} For an honestly generated signature $\Sigma$ for message $M$, we first show that the $\mathsf{Verify}$ algorithm always outputs~$1$.  Due to the honest behavior of the user, when signing a message in the name of the group, this user possesses a valid tuple~$\zeta$ of the form~(\ref{equation:ATS-witness-zeta}). Therefore,  $\Pi_{\mathsf{gs}}$ will be accepted by the $\mathsf{Verify}$ algorithm with probability~$1$ due to the perfect completeness of our argument system.

If an honest user is traceable, then $\mathsf{Account}(\mathsf{gpk},\mathsf{cert},w^{\mathsf{escrw}},1)$ will output~$1$, implied by the correctness of Ducas-Micciancio signature scheme and honest behaviour of group manager. In terms of the correctness of the $\mathsf{Open}$ algorithm, we observe that $\mathbf{c}_{1,2}-\mathbf{c}_{1,1}\cdot s_1=$
\begin{align*}
(\mathbf{b}_1^{(\mathsf{tr})}-\mathbf{a}_1^{(\mathsf{tr})}\hspace*{-1pt}\cdot\hspace*{-1pt} s_1)\hspace*{-1pt}\cdot\hspace*{-1pt} g_1\hspace*{-1pt}\cdot \hspace*{-1pt} g'_1 + \mathbf{e}_{1,2}\hspace*{-1pt}\cdot \hspace*{-1pt}g'_1-\mathbf{e}_{1,1}\hspace*{-1pt}\cdot\hspace*{-1pt} s_1\hspace*{-1pt}\cdot \hspace*{-1pt} g'_1+\mathbf{e}'_{1,2}- \mathbf{e}'_{1,1} \hspace*{-1pt}\cdot\hspace*{-1pt} s_1+\lfloor q/4\rfloor\cdot \rdec(p),
\end{align*}
denoted as $\widetilde{\mathbf{e}}+\lfloor q/4\rfloor\cdot \rdec(p)$.
In this case, $\mathsf{tr}=1$,  $\mathbf{b}_1^{(\mathsf{tr})}-\mathbf{a}_1^{(\mathsf{tr})}\cdot s_1 =\mathbf{e}_1$, and $\|\widetilde{\mathbf{e}}\|_{\infty}\leq \big\lceil \frac{q}{10}\big\rceil$. The decryption can recover $\rdec(p)$ and hence  the real signer due to the correctness of our key-oblivious encryption from Section~\ref{subsection:KOE-construction}. Thus, correctness of the $\mathsf{Open}$ algorithm follows. What is more,  $\Pi_{\mathsf{open}}$ will be accepted by the $\mathsf{Judge}$ algorithm with probability~$1$ due to the perfect completeness of our argument system.

If an honest user is non-traceable, then again $\mathsf{Account}(\mathsf{gpk},\mathsf{cert},w^{\mathsf{escrw}},1)$ will output~$1$.  For the $\mathsf{Open}$ algorithm, since $\mathbf{b}_1^{(0)}-\mathbf{a}_1^{(0)}\cdot s_1=\mathbf{a}_{1}^{(0)}\cdot (s_{-1}-s_1)+\mathbf{e}_{-1}$, then we obtain $$\mathbf{c}_{1,2}-\mathbf{c}_{1,1}\cdot s_1=\mathbf{a}_{1}^{(0)}\cdot (s_{-1}-s_1)\cdot g_1\cdot g'_1+\widetilde{\mathbf{e}}+\lfloor q/4\rfloor\cdot \rdec(p),$$
where $\|\widetilde{\mathbf{e}}\|_{\infty}\leq \big\lceil \frac{q}{10}\big\rceil$. Observe that $\mathbf{a}_1^{(0)}\xleftarrow{\$} R_q^{\ell}$,  and $s_{-1}\neq s_1$ with overwhelming probability.
 Over the randomness of $g_1,g'_1$, the decryption algorithm described in Section~\ref{subsection:KOE-construction} will output a random element $p'\in R_q$.
Then, with overwhelming probability, $p'$ is not in the registration table and the $\mathsf{Open}$ algorithm outputs~$\bot$. 
It then follows that our scheme is correct. \smallskip

\smallskip
\noindent
{\sc Security. } In Theorem~\ref{theorem:ATS-security-theorem}, we prove that our scheme satisfies the security requirements of accountable tracing signatures, as specified by Kohlweiss and Miers.

\begin{theorem}\label{theorem:ATS-security-theorem}
Under the $\mathsf{RLWE}$ and $\mathsf{RSIS}$ assumptions, the accountable tracing signature scheme described in Section~\ref{subsection:ATS-construction} satisfies the following requirements in the random oracle model: (i) anonymity under tracing; (ii) traceability; (iii) non-frameability; (iv) anonymity with accountability; and (v) trace-obliviousness.
\end{theorem}

For the proofs of  traceability and non-frameability, the lemma below  from~\cite{LNWX18} is needed.
\begin{lemma}[\cite{LNWX18}]\label{lemma:counting-argument}
Let $\mathbf{B}\in R_q^{1\times m}$, where $m\geq 2 \lceil \log q\rceil +2$. If $\mathbf{x}$ is a uniform element over $R^m$ with $\|\mathbf{x}\|_{\infty}\leq 1$, then with probability at least $1-2^{-n}$, there exists a different $\mathbf{x'}\in R^m$ with $\|\mathbf{x'}\|_{\infty}\leq 1$ and  $\mathbf{B}\cdot \mathbf{x}'=\mathbf{B}\cdot \mathbf{x} \in R_q$.
\end{lemma}

The proof of the Theorem~\ref{theorem:ATS-security-theorem} follows from  Lemma~\ref{lemma:ATS-AuT}-\ref{lemma:ATS-TO} given below.

\begin{lemma}\label{lemma:ATS-AuT}
Assuming  the hardness of the $\mathsf{RLWE}$ problem, in the random oracle model, the given accountable tracing signature scheme is anonymous under tracing.

\end{lemma}
\begin{proof}
We prove this lemma using a series of indistinguishable games. In the initial game, the challenger runs the experiment $\mathbf{Exp}_{\mathsf{ATS},\mathcal{A}}^{\mathsf{AuT}-0}(\lambda)$ while in the last game, the challenger runs the experiment $\mathbf{Exp}_{\mathsf{ATS},\mathcal{A}}^{\mathsf{AuT}-1}(\lambda)$. 
Let $W_i$ be the event that the adversary outputs~$1$ in Game~$i$.

\begin{description}
\item[Game~$0$:] This is exactly the experiment $\mathbf{Exp}_{\mathsf{ATS},\mathcal{A}}^{\mathsf{AuT}-0}(\lambda)$, where the adversary receives a challenged signature $(\Pi_{\mathsf{gs}}^{*},\mathbf{c}_1^*,\mathbf{c}_2^*)\leftarrow\mathsf{Sign}(\mathsf{gpk},\mathsf{cert}_0,\mathsf{usk}_{0},M)$ in the challenge phase with $p_0=\mathbf{B}\cdot \mathsf{usk}_{0}$. So $\mathrm{Pr}[W_0]=\mathrm{Pr}[\mathbf{Exp}_{\mathsf{ATS},\mathcal{A}}^{\mathsf{AuT}-0}(\lambda)=1]$. \smallskip

\item[Game $1$:]  We modify Game~$0$ as follows: the challenger will keep decryption key $(s_2, \mathbf{e}_2)$ secret (by himself) instead of erasing it. However, the view of the adversary $\mathcal{A}$ is still the same as in Game $0$. Therefore, $\mathrm{Pr}[W_0]=\mathrm{Pr}[W_1]$. \smallskip

\item[Game~$2$:] This game is the same as Game~$1$ with one exception: it generates   simulated proofs for the opening oracle queries by programming the random oracle $\mc{H}_{\mathsf{FS}}$. Note that the challenger still follows the original game (that is, it uses $s_1$ to decrypt $\mathbf{c}_1$) to identify the real signer. The views of $\mathcal{A}$ in Game $1$ and Game $2$ are statistically close  due to the statistical zero-knowledge property of our argument system. Therefore $\mathrm{Pr}[W_1]\overset{s}{\approx}\mathrm{Pr}[W_2]$. \smallskip

\item[Game~$3$:] This game modifies Game~$2$ as follows. It uses $s_2$ instead of $s_1$ to  answer the opening oracle queries. In other words, it now uses $\mathbf{s}_2$ to decrypt $\mathbf{c}_2$ to identify the signer. The view of the adversary in this game is identical to that in Game~$2$ until event $F_1$, where $\mc{A}$ queries the opening oracle a valid signature $(\Pi_{\mathsf{gs}},\mathbf{c}_1,\mathbf{c}_2)$ with $\mathbf{c}_1,\mathbf{c}_2$ encrypting distinct messages, happens. Since the event $F_1$ violates the soundness of our argument system, we have $\vert\mathrm{Pr}[W_2]-\mathrm{Pr}[W_3]\vert\leq \mathrm{Pr}[F_1]\leq \mathbf{Adv}_{\Pi_{\mathsf{gs}}}^{\mathsf{sound}}(\lambda)=\mathsf{negl}(\lambda)$. \smallskip

\item[Game~$4$:] This game changes Game~$3$ as follows. It generates a simulated proof $\Pi_{\mathsf{gs}}^*$ in the  challenge  phase even though the challenger has the correct witness to generate a real proof. Due to the statistical zero-knowledge property of our argument system, this change is negligible to $\mathcal{A}$.  Therefore $\mathrm{Pr}[W_3]\overset{s}{\approx}\mathrm{Pr}[W_4]$. \smallskip

\item[Game~$5$:] In this game, we modify Game~$4$ by modifying the distribution of the challenged signature $\Sigma^*=(\Pi_{\mathsf{gs}}^*,\mathbf{c}_1^*,\mathbf{c}_2^*)$ as follows. For $i\in\{0,1\}$, parse $\mathsf{cert}_i=(p_i,\mathbf{a}'_{1,i},\mathbf{b}'_{1,i},\mathbf{a}'_{2,i},\mathbf{b}'_{2,i},t_i,\mathbf{r}_i,\mathbf{v}_i)$. Recall that in Game~$4$,  both $\mathbf{c}_1^*$ and $\mathbf{c}_2^*$ encrypt the same message, i.e., $\rdec(p_{0})$, under the randomized key $(\mathbf{a}'_{1,0},\mathbf{b}'_{1,0})$ and $(\mathbf{a}'_{2,0},\mathbf{b}'_{2,0})$, respectively. Here we change $\mathbf{c}_1^*$ to be encryption of $\rdec(p_{1})$ and keep $\mathbf{c}_2^*$ unchanged. By the semantic security under key randomization of our  key oblivious encryption scheme for public key $(\mathbf{a}_1^{(1)},\mathbf{b}_1^{(1)})$ (which is implied by the $\mathsf{RLWE}$ assumption since we no longer use $s_1$ to open signatures), the change made in this game is negligible to the adversary.  Therefore we have $\vert\mathrm{Pr}[W_4]-\mathrm{Pr}[W_5]\vert=\mathsf{negl}(\lambda)$. \smallskip
\item[Game~$6$:] In this game, we further modify  the distribution of the challenged signature $\Sigma^*$. We change $\mathbf{c}_1^*$ to be encryption of $\rdec(p_{1})$ under a fresh and then randomized key. By the property of key privacy under key randomization of our key-oblivious encryption scheme, the change made in this game is negligible to the adversary.  Therefore we have $\vert\mathrm{Pr}[W_5]-\mathrm{Pr}[W_6]\vert=\mathsf{negl}(\lambda)$. \smallskip
\item[Game~$7$:] In this game, we again modify  the distribution of the challenged signature $\Sigma^*$. We change $\mathbf{c}_1^*$ to be encryption of $\rdec(p_{1})$ under the randomized key $(\mathbf{a}'_{1,1},\mathbf{b}'_{1,1})$. By the same argument of indistinguishability between Game~$6$ and Game~$5$, we have $\vert\mathrm{Pr}[W_6]-\mathrm{Pr}[W_7]\vert=\mathsf{negl}(\lambda)$. \smallskip

\item[Game~$8$:]  This game is the same as Game~$7$ with one modification: it changes back to $s_1$ for the opening oracle queries and erases $(s_2,\mathbf{e}_2)$ again.     This change is indistinguishable to $\mathcal{A}$ until event $F_2$, where $\mc{A}$ queries a valid signature $(\Pi_{\mathsf{gs}},\mathbf{c}_1,\mathbf{c}_2
    )$ with $\mathbf{c}_1,\mathbf{c}_2$ encrypting different messages to the opening oracle, occurs. Since event $F_2$ violates the simulation soundness of our argument system, we have $\vert\mathrm{Pr}[W_7]-\mathrm{Pr}[W_8]\vert\leq \mathsf{Adv}_{\Pi_{\mathsf{gs}}}^{\mathsf{ss}}(\lambda)=\mathsf{negl}(\lambda)$.  \smallskip

\item[Game~$9$:] In this game, we modify Game~$8$ by modifying the distribution of the challenged signature $\Sigma^*=(\Pi_{\mathsf{gs}}^*,\mathbf{c}_1^*,\mathbf{c}_2^*)$ again.
It changes $\mathbf{c}_2^*$ to be encryption of $\rdec(p_{1})$ under the randomized key $(\mathbf{a}'_{2,1},\mathbf{b}'_{2,1})$ in the  {challenge} phase.  
By the same argument of indistinguishability from Game~$4$ to Game~$7$, we have $\vert\mathrm{Pr}[W_8]-\mathrm{Pr}[W_9]\vert=\mathsf{negl}(\lambda)$.
 \smallskip

\item[Game~$10$:] 
Note that in Game~$9$, both $\mathbf{c}_1^*$ and $\mathbf{c}_2^*$ encrypt the same message, i.e., $\rdec(p_{1})$, under the randomized key $(\mathbf{a}'_{1,1},\mathbf{b}'_{1,1})$ and $(\mathbf{a}'_{2,1},\mathbf{b}'_{2,1})$, respectively. Therefore, the challenger has correct witness to generate $\Pi_{\mathsf{gs}}^{*}$. In this game, we modify Game~$9$ by switching back to
a real proof $\Pi_\mathsf{gs}^*$ in the {challenge} phase. Then the views of $\mathcal{A}$ in Game~$9$ and Game~$10$ are statistically indistinguishable by the statistical zero-knowledge property of our argument system.  Hence $\mathrm{Pr}[W_9]\overset{s}{\approx}\mathrm{Pr}[W_{10}]$. \smallskip

\item[Game~$11$:] This game changes Game~$10$ in one aspect. It now generates  real proofs  for the opening oracle queries. Due to the statistical zero-knowledge property of our argument system,  Game~$10$ and Game~$11$  are statistically indistinguishable to $\mathcal{A}$. In other words, we have $\mathrm{Pr}[W_{10}]\overset{s}{\approx}\mathrm{Pr}[W_{11}]$. This is indeed the experiment $\mathbf{Exp}_{\mathsf{ATS},\mc{A}}^{\mathsf{AuT}-1}(\lambda)$. Hence, we have   $\mathrm{Pr}[W_{11}]=\mathrm{Pr}[\mathbf{Exp}_{\mathsf{ATS},\mc{A}}^{\mathsf{AuT}-1}(\lambda)=1]$.

\end{description}
As a result, we obtain  $$\vert\mathrm {Pr}[\mathbf{Exp}_{\mathsf{ATS},\mc{A}}^{\mathsf{AuT}-1}(\lambda)=1]-\mathrm{Pr}[\mathbf{Exp}_{\mathsf{ATS},\mc{A}}^{\mathsf{AuT}-0}(\lambda)=~1]\vert = \mathsf{negl}(\lambda),$$
and hence our scheme is anonymous under tracing.
\end{proof}

\begin{lemma}\label{lemma:ATS-Trace}
Assuming the  hardness of the $\mathsf{RSIS}$ problem, in the random oracle model, the given accountable tracing signature scheme is traceable .
\end{lemma}

\begin{proof}
We show that the success probability $\epsilon$ of $\mathcal{A}$ against traceability is negligible by the unforgeability of the Ducas-Micciancio signature recalled in Section~\ref{subsection:DM-signatures}, which in turn relies on the hardness of the $\mathsf{RSIS}$ problem, or by the hardness of solving a   $\mathsf{RSIS}$ instance directly.
\smallskip

Let $\mathcal{C}$ be the challenger and honestly run the experiment $\mathbf{Exp}_{\mathsf{ATS},\mathcal{A}}^{\mathsf{Trace}}(\lambda)$.
When $\mathcal{A}$ halts, it outputs $(M^*,\Pi_{\mathsf{gs}}^*,\mathbf{c}_1^*,\mathbf{c}_2^*)$. Let us consider the case that $\mathcal{A}$ wins.   Parse $\Pi_{\mathsf{gs}}^*=(\{\mathrm{CMT}_i^*\}_{i=1}^{\kappa},\mathrm{CH}^*,\{\mathrm{RSP}_i^*\}_{i=1}^{\kappa})$. Let
$$\xi^*=(\mathbf{A},\mathbf{A}_{[0]},\ldots,\mathbf{A}_{[d]},\mathbf{F},\mathbf{F}_0,\mathbf{F}_1,u,\mathbf{B},\mathbf{c}_1^*,\mathbf{c}_2^*). $$

Then $\mathrm{CH}^*=\mathcal{H}_{\mathsf{FS}}\big(M^*, \{\mathrm{CMT}_i^*\}_{i=1}^\kappa, \xi^* \big)$ and for each $i\in[\kappa]$, $\mathrm{RSP}_i^*$ is a valid response corresponding to $\mathrm{CMT}_i^*$ and $\mathrm{CH}_i^*$. This is due to  the fact that $\mathcal{A}$ wins and hence $\Pi_{\mathsf{gs}}^*$ passes the verification process.

\smallskip
We remark that  $\mathcal{A}$ had queried the tuple $\big(M^*, \{\mathrm{CMT}_i^*\}_{i=1}^\kappa, \xi^* \big)$  to the hash oracle $\mathcal{H}_{\mathsf{FS}}$ with all but negligible probability. Since we can only guess correctly the value $\mathcal{H}_{\mathsf{FS}}\big(M^*, \{\mathrm{CMT}_i^*\}_{i=1}^\kappa, \xi^* \big)$ with probability $3^{-\kappa}$, which is negligible. Therefore, $\mathcal{A}$ had queried the tuple $\big(M^*, \{\mathrm{CMT}_i^*\}_{i=1}^\kappa, \xi^* \big)$ to $\mathcal{H}_{\mathsf{FS}}$ with probability $\epsilon'=\epsilon-3^{-\kappa}$. Let  this tuple  be the $\theta^*$-th oracle query made by $\mathcal{A}$ and assume $\mathcal{A}$ had made $Q_H$ queries in total.

Up to this point, the challenger $\mc{C}$ then replays the behaviour of $\mathcal{A}$ for at most $32 \cdot Q_H/\epsilon'$ times. In each new replay, $\mathcal{A}$ is given the same hash answers $r_1,\ldots,r_{\theta^*-1}$ as in the original run for the first $\theta^*-1$ hash queries while it is given uniformly random and independent values $r'_{\theta^*},\ldots, r'_{Q_H}$ for the remaining hash queries.
According to the forking lemma of Brickell et al.~\cite{BPVY00}, with probability $\geq 1/2$, $\mc{B}$ obtains $3$-fork involving the same tuple $\big(M^*, \{\mathrm{CMT}_i^*\}_{i=1}^\kappa, \xi^* \big)$ with pairwise distinct hash values $\mathrm{CH}_{\theta^*}^{(1)}, \mathrm{CH}_{\theta^*}^{(2)}, \mathrm{CH}_{\theta^*}^{(3)}\in\{1,2,3\}^{\kappa}$  and corresponding valid responses $\mathrm{RSP}_{\theta^*}^{(1)}$, $\mathrm{RSP}_{\theta^*}^{(2)}$, $\mathrm{RSP}_{\theta^*}^{(3)}$. We observe that with probability $1-(\frac{7}{9})^{\kappa}$,  there exists  some $j\in\{1,2,\ldots,\kappa\}$ such that $\{\mathsf{CH}_{\theta^*,j}^{(1)},\mathsf{CH}_{\theta^*,j}^{(2)},\mathsf{CH}_{\theta^*,j}^{(3)}\}=\{1,2,3\}$. \smallskip

In other words, we obtain three valid responses
$\mathrm{RSP}_{\theta^*,j}^{(1)}$, $\mathrm{RSP}_{\theta^*,j}^{(2)}$, $\mathrm{RSP}_{\theta^*,j}^{(3)}$ for all the challenges $1,2,3$   with respect to the same commitment $\mathrm{CMT}_{j}^*$. Due to the computational binding property of the $\mathsf{COM}$ scheme, $\mathcal{C}$ is able to extract $\zeta^*$ of form
$$\zeta^*=(p^*,\mathbf{a}_1^*,\mathbf{b}_1^*,\mathbf{a}_2^*,\mathbf{b}_2^*, t^*,\mathbf{r}^*, \mathbf{v}^*,\mathbf{x}^*, g_1^*,\mathbf{e}_{1,1}^*,\mathbf{e}_{1,2}^*,g_2^*,\mathbf{e}_{2,1}^*,\mathbf{e}_{2,2}^*)
 $$
such that $t^*\in\mathcal{T}_d$, $\mathbf{r}^*, \mathbf{v}^*$ have infinity bound $\beta$, $\mathbf{g}_1^*,\mathbf{e}_{1,1}^*,\mathbf{e}_{1,2}^*,g_2^*,\mathbf{e}_{2,1}^*, \mathbf{e}_{2,2}^*$ have infinity bound $B$, $\mathbf{x}^*$ has infinity bound~$1$; and equations $\mathbf{B}\cdot \mathbf{x}^*=p^*$ and
$$\mathbf{A}_{t^*}\cdot \mathbf{v}^*=u+\mathbf{F}\cdot \rdec(\mathbf{F}_0\cdot \mathbf{r}^*+\mathbf{F}_1\cdot \rdec(p^*\|\mathbf{a}_1^*\|\mathbf{b}_1^*\|\mathbf{a}_2^*\|\mathbf{b}_2^*))$$ hold, and   $\mathbf{c}_1^*,\mathbf{c}_2^*$ are ciphertexts of $\rdec(p^*)$ under the key $(\mathbf{a}_1^*,\mathbf{b}_1^*)$ and $(\mathbf{a}_2^*,\mathbf{b}_2^*)$ with randomness $(g_1^*,\mathbf{e}_{1,1}^*,\mathbf{e}_{1,2}^*)$ and $(g_2^*,\mathbf{e}_{2,1}^*,\mathbf{e}_{2,2}^*)$, respectively.
\smallskip

Since $\mathcal{A}$ wins the game, then either (i) the $\mathsf{Open}$ algorithm outputs~$\bot$ or  (ii) the $\mathsf{Open}$ algorithm outputs $({p'},\Pi_{\mathsf{open}}^*)$ with ${p'}\neq \bot$ but the proof $\Pi_{\mathsf{open}}^*$ is not accepted by the $\mathsf{Judge}$ algorithm.
\smallskip

By the unforgeability of the underlying signature scheme, with overwhelming probability, $(p^*,\mathbf{a}_1^*,\mathbf{b}_1^*,\mathbf{a}_2^*,\mathbf{b}_2^*, t^*,\mathbf{r}^*, \mathbf{v}^*)$ is a certificate returned by the $\mathsf{Enroll}$ oracle. In other words, $p^*$ is a registered user. If $p^*$ is a non-traceable user, then $\mathcal{A}$ does not hold the user secret key of $p^*$, denoted as $\mathbf{x}'$. Note that this is ensured by the definition of traceability described in Section~\ref{subsection:ATS-security-model}. With probability $\geq 1/2$, $\mathbf{x}^*\neq \mathbf{x}'$ by Lemma~\ref{lemma:counting-argument}, in which case we obtain  a vector $\mathbf{y}=\mathbf{x}^*-\mathbf{x}'\neq \mathbf{0}$ so that $\mathbf{B}\cdot \mathbf{y}=0$ and $\|\mathbf{y}\|_{\infty}\leq \|\mathbf{x}^*\|_{\infty}+\|\mathbf{x}'\|_{\infty}\leq 2$. This solves a $\mathsf{RSIS}$ instance. Therefore, the $\mathsf{Open}$ algorithm outputs $\bot$ with negligible probability. In other words, case (i) happens with negligible probability. On the other hand, if $p^*$ is a traceable user. Then by the correctness of the underlying encryption scheme, the $\mathsf{Open}$ algorithm will output $p^*$. Furthermore, by the honest behaviour of decryption (performed by the honest challenger), the $\mathsf{Judge}$ algorithm always outputs $1$. This implies case (ii) occurs with negligible probability.  This concludes the proof.
\end{proof}

\begin{lemma}\label{lemma:ATS-NF}
Assuming the hardness of the $\mathsf{RSIS}$ problem, in the random oracle model,  the given accountable tracing signature scheme is non-frameable.
\end{lemma}

\begin{proof}
We show that the success probability $\epsilon$ of $\mathcal{A}$ against non-frameability is negligible assuming the hardness of solving a $\mathsf{RSIS}$ instance.

Let $\mathcal{C}$ be the challenger and faithfully run the experiment $\mathbf{Exp}_{\mathsf{ATS},\mathcal{A}}^{\mathsf{NF}}(\lambda)$. When $\mathcal{A}$ halts, it outputs the tuple $(M^*,\Pi_{\mathsf{gs}}^*,\mathbf{c}_1^*,\mathbf{c}_2^*, {p}^*,\Pi_{\mathsf{open}}^*)$. Let us consider the case that $\mathcal{A}$ wins.

The fact that $\mathcal{A}$ wins the game implies  $(\Pi_{\mathsf{gs}}^*,\mathbf{c}_1^*,\mathbf{c}_2^*)$ is a valid signature of the message $M^*$ that was not obtained from queries. By the same extraction technique as in Lemma~\ref{lemma:ATS-Trace}, we can extract witness $\mathbf{x}'\in R_q^{m}$ and ${p'}\in R_q$ such that $\|\mathbf{x}'\|_{\infty}\leq 1$, $\mathbf{B}\cdot \mathbf{x}'=p'$  and $\mathbf{c}_1^*,\mathbf{c}_2^*$ are correct encryptions of  $\rdec(p')$. By the correctness of the underlying encryption scheme, $\mathbf{c}_1^*$ will be decrypted to $p'$.

\smallskip

The fact that $\mathcal{A}$ wins the game also implies $\Pi_{\mathsf{open}}^*$ passes the verification process of the  $\mathsf{Judge}$ algorithm. Due to the soundness of the argument system that is used to generate $\Pi_{\mathsf{open}}^*$,  $\mathbf{c}_1^*$ will be decrypted to $p^*$. Hence we have $p'=p^*$. We observe that $\mathcal{A}$ wins the game also implies that $\mathcal{A}$ does not know the user secret key $\mathbf{x}^{*}$ that corresponds to $p^*$. Thus we obtain: $\mathbf{B}\cdot \mathbf{x'}=p'=p^*=\mathbf{B}\cdot \mathbf{x}^*$, where $\|\mathbf{x}^*\|_{\infty}\leq 1$. Lemma~\ref{lemma:counting-argument} implies that $\mathbf{x'}\neq \mathbf{x}^*$ with probability at least $1/2$.  If they are not equal, we obtain a  vector $\mathbf{y}=\mathbf{x'}-\mathbf{x}^*\neq \mathbf{0}$ such that $\mathbf{B}\cdot \mathbf{y}=0$ and $\|\mathbf{y}\|_{\infty}\leq \|\mathbf{x}^*\|_{\infty}+\|\mathbf{x}'\|_{\infty}\leq 2$. However, under the hardness of the $\mathsf{RSIS}$ problem, the success probability of $\mathcal{A}$ is negligible. This concludes the proof.
\end{proof}

\begin{lemma}\label{lemma:ATS-AwA}
Assuming the hardness of  the $\mathsf{RLWE}$ problem, in the random oracle model, the given accountable tracing signature scheme is anonymous with accountability.
\end{lemma}
\begin{proof}
The proof of this lemma is similar to Lemma~\ref{lemma:ATS-AuT} except that we do not need to switch between two decryption keys. This is because the randomized keys in the certificate of the challenged users are obtained from the  pairs $(\mathbf{a}_1^{(0)},\mathbf{b}_1^{(0)})$ and $(\mathbf{a}_2^{(0)},\mathbf{b}_2^{(0)})$, which are not related to the opening key.  The details are omitted here.
\end{proof}

\begin{lemma}\label{lemma:ATS-TO}
Assuming  the hardness of the $\mathsf{RLWE}$ problem, in the random oracle model,  the given accountable tracing signature  scheme is  trace-oblivious.

\end{lemma}
\begin{proof}
We proceed through a sequence of hybrids. Let $W_i$ be the event that adversary outputs $1$ in Game~$i$.
\begin{description}
\item[Game~$0$:] Let this game be the experiment $\mathbf{Exp}_{\mathsf{ATS},\mathcal{A}}^{\mathsf{TO}-0}(\lambda)$, where the adversary receives  $\mathsf{cert}$ for user $p$ of his choice.  Parse $\mathsf{cert}$ as $(p,\mathbf{a}'_{1},\mathbf{b}'_1,\mathbf{a}'_2,\mathbf{b}'_2,t,\mathbf{r},\mathbf{v})$. Note that $(\mathbf{a}'_{1},\mathbf{b}'_1)$ and $(\mathbf{a}'_2,\mathbf{b}'_2)$ are randomized keys from $(\mathbf{a}_{1}^{(0)},\mathbf{b}_1^{(0)})$ and $(\mathbf{a}_{2}^{(0)},\mathbf{b}_2^{(0)})$, respectively. We then have $\text{Pr}[W_0]=\text{Pr}[\mathbf{Exp}_{\mathsf{ATS},\mathcal{A}}^{\mathsf{TO}-0}(\lambda)=1]$.  

\item[Game~$1$:] We modify Game~$0$ by replacing $(\mathbf{a}'_{1},\mathbf{b}'_1)$ with a new fresh key $(\widetilde{\mathbf{a}}_1,\widetilde{\mathbf{b}}_1)$ generated by the $\mathsf{KeyGen}$ algorithm of our $\mathsf{KOE}$ scheme. It then follows from the key randomizability of our encryption scheme, this modification is negligible to the adversary. Therefore, we have $\vert\text{Pr}[W_0]-\text{Pr}[W_1]\vert=\mathrm{negl}(\lambda)$.  

\item[Game~$2$:]  We modify Game~$1$ by replacing $(\mathbf{a}'_{2},\mathbf{b}'_2)$ with a new fresh key $(\widetilde{\mathbf{a}}_2,\widetilde{\mathbf{b}}_2)$  as in Game~$1$. By the same argument, we have $\vert\text{Pr}[W_1]-\text{Pr}[W_2]\vert=\mathrm{negl}(\lambda)$.  

\item[Game~$3$:] We change Game~$2$ by replacing $(\widetilde{\mathbf{a}}_2,\widetilde{\mathbf{b}}_2)$ with $(\mathbf{a}'_2,\mathbf{b}'_2)$ that are randomized key from $(\mathbf{a}_{2}^{(1)},\mathbf{b}_2^{(1)})$. By the key randomizability of our encryption scheme, we have $\vert\text{Pr}[W_2]-\text{Pr}[W_3]\vert=\mathrm{negl}(\lambda)$.  

\item[Game~$4$:] We change Game~$3$ by replacing $(\widetilde{\mathbf{a}}_1,\widetilde{\mathbf{b}}_1)$ with $(\mathbf{a}'_1,\mathbf{b}'_1)$ that are randomized key from $(\mathbf{a}_{1}^{(1)},\mathbf{b}_1^{(1)})$. We then have $\vert\text{Pr}[W_3]-\text{Pr}[W_4]\vert=\mathrm{negl}(\lambda)$. This is exactly the experiment $\mathbf{Exp}_{\mathsf{ATS},\mathcal{A}}^{\mathsf{TO}-1}(\lambda)$. Therefore, we obtain  $\text{Pr}[W_4]=\text{Pr}[\mathbf{Exp}_{\mathsf{ATS},\mathcal{A}}^{\mathsf{TO}-1}(\lambda)=1]$.
\end{description}
Therefore, we obtain $\vert\text{Pr}[\mathbf{Exp}_{\mathsf{ATS},\mathcal{A}}^{\mathsf{TO}-1}(\lambda)=1]-\text{Pr}[\mathbf{Exp}_{\mathsf{ATS},\mathcal{A}}^{\mathsf{TO}-0}(\lambda)=1]\vert=\mathrm{negl}(\lambda)$. This implies that our scheme is trace-oblivious.
\end{proof}

\section*{Acknowledgements}
The research is supported by Singapore Ministry of Education under Research Grant MOE2016-T2-2-014(S). Khoa Nguyen is also supported by the Gopalakrishnan -- NTU Presidential Postdoctoral Fellowship 2018.

\bibliographystyle{abbrv}

\end{document}